\documentclass[notitlepage,12pt]{article}%
\usepackage{latexsym,epsfig,amsmath}
\usepackage{epsfig,graphicx}
\usepackage{amsfonts}
\usepackage{tabularx}
\usepackage{rotating}
\usepackage{harvard}
\usepackage{hyperref}
\usepackage{amsmath}
\usepackage{amssymb}
\usepackage{graphicx}%
\providecommand{\U}[1]{\protect\rule{.1in}{.1in}}
\newtheorem{theorem}{Theorem}

\newtheorem{condition}{Condition}

\newtheorem{corollary}{Corollary}

\newtheorem{example}{Example}

\newtheorem{lemma}{Lemma}

\newtheorem{remark}{Remark}

\newenvironment{proof}[1][Proof]{\textbf{#1.} }{\  \rule{0.5em}{0.5em}}

\begin{document}

\title{Central Limit Theory for Combined Cross-Section and Time Series with an
Application to Aggregate Productivity Shocks\thanks{We thank Peter Phillips
and referees in previous rounds for many helpful suggestions. }}
\author{Jinyong Hahn\thanks{UCLA, Department of Economics, 8283 Bunche Hall, Mail
Stop: 147703, Los Angeles, CA 90095, hahn@econ.ucla.edu}\\UCLA
\and Guido Kuersteiner\thanks{University of Maryland, Department of Economics,
Tydings Hall 3145, College Park, MD, 20742, kuersteiner@econ.umd.edu}\\University of Maryland
\and Maurizio Mazzocco\thanks{UCLA, Department of Economics, 8283 Bunche Hall, Mail
Stop: 147703, Los Angeles, CA 90095, mmazzocc@econ.ucla.edu}\\UCLA}
\maketitle

\begin{abstract}
Combining cross-section and time series data is a long and well established
practice in empirical economics. We develop a central limit theory that
explicitly accounts for possible dependence between the two data sets. We
focus on common factors as the mechanism behind this dependence. Using our
central limit theorem (CLT) we establish the asymptotic properties of
parameter estimates of a general class of models based on a combination of
cross-sectional and time series data, recognizing the interdependence between
the two data sources in the presence of aggregate shocks. Despite the
complicated nature of the analysis required to formulate the joint CLT, it is
straightforward to implement the resulting parameter limiting distributions
due to a formal similarity of our approximations with the standard Murphy and
Topel's (1985) formula.

\end{abstract}

\section{Introduction}

There is a long tradition in empirical economics of relying on information
from a variety of data sources to estimate model parameters. In this paper we
focus on a situation where cross-sectional and time-series data are combined.
This may be done for a variety of reasons. Some parameters may not be
identified in the cross section or time series alone. Alternatively,
parameters estimated from one data source may be used as first-step inputs in
the estimation of a second set of parameters based on a different data source.
This may be done to reduce the dimension of the parameter space or more
generally for computational reasons.

Data combination generates theoretical challenges, even when only
cross-sectional data sets or time-series data sets are combined. See Ridder
and Moffitt (2007) for example. We focus on dependence between cross-sectional
and time-series data produced by common aggregate factors. Andrews (2005)
demonstrates that even randomly sampled cross-sections lead to independently
distributed samples only conditionally on common factors, since the factors
introduce a possible correlation. This correlation extends inevitably to a
time series sample that depends on the same underlying factors. Andrews (2005)
only considers cross-sections and panels with fixed $T$ with sufficient
regularity conditions for consistent estimation. Our work in this and in our
companion papers Hahn, Kuersteiner and Mazzocco (2020, 2021) extends the
analysis in Andrews (2005) to situations where consistent estimation is only
possible by utilizing an additional time series data set. Examples in our
companion papers draw on literatures for the structural estimation of rational
expectations models as well as the program evaluation literature. Examples for
the former include cross-sectional studies of the consumption based asset
pricing model in Runkle (1991), Shea (1995) and Vissing-Jorgensen (2002) and a
recent example of the latter is Rosenzweig and Udry (2019).

The first contribution of this paper is to develop a central limit theory that
explicitly accounts for the dependence between the cross-sectional and
time-series data by using the notion of stable convergence. The second
contribution is to use the corresponding central limit theorem to derive the
asymptotic distribution of parameter estimators obtained from combining the
two data sources. Compared to our companion paper Hahn, Kuersteiner and
Mazzocco (2021, henceforth HKM21), we impose a martingale difference structure
as our fundamental moment condition on which our estimators are based in this
paper. This implies that the theory in this paper is suitable for correctly
specified maximum likelihood and conditional moment based estimators. Our
companion paper HKM21 does not impose a martingale structure for the time
series model but rather allows for mixingale processes. Thus, certain forms of
model misspecification can be handled by the theory of that paper. On the
other hand, in this paper we do not assume independence between time series
and cross-section data conditional on common factors, independence in the
cross-section conditional on common factors or stationarity or homogeneity in
the temporal direction of panel or time series data, as is done in HKM21. Due
to the strong assumptions in HKM21, the proof strategy for joint convergence
is relatively simple in that paper. Conditionally on common factors, a CLT for
independent cross-sections combined with a stable CLT for stationary
mixingales is sufficient to establish joint convergence of the time-series and
cross-sectional components. In this paper where conditional independence of
the cross-sectional sample is not assumed, a more complicated joint
convergence argument establishing the joint limiting behavior of time series
and cross-sectional averages is required.

Our analysis is inspired by a number of applied papers and in particular by
the discussion in Lee and Wolpin (2006, 2010). Econometric estimation based on
the combination of cross-sectional and time-series data is an idea that dates
back at least to Tobin (1950). More recently, Heckman and Sedlacek (1985) and
Heckman, Lochner, and Taber (1998) proposed to deal with the estimation of
equilibrium models by exploiting such data combination. It is, however, Lee
and Wolpin (2006, 2010) who develop the most extensive equilibrium model and
estimate it using similar intuition and panel data.

To derive the new central limit theorem and the asymptotic distribution of
parameter estimates, we extend the model developed in Lee and Wolpin (2016,
2010) to a general setting that involves two submodels. The first submodel
includes all the cross-sectional features, whereas the second submodel is
composed of all the time-series aspects. The two submodels are linked by a
vector of aggregate shocks and by the parameters that govern their dynamics.
Given the interplay between the two submodels, the aggregate shocks have
complicated effects on the estimation of the parameters of interest.

Another important literature that often uses micro data to calibrate model
parameters is the literature on dynamic stochastic general equilibrium (DSGE)
models in macro economics. The calibration approach was prominently advocated
by Kydland and Prescott (1982) who also emphasize the importance of aggregate
shocks to understand aggregate business cycle fluctuations. Other important
contributions to the DSGE literature emphasizing the persistence of aggregate
shocks are Long and Plosser (1983) and Smet and Wouters (2007), to name only
two in a large literature. More recently, Schorfheide (2000) proposed a formal
Bayesian approach that uses Bayesian priors to aid in the estimation of DSGE
model parameters, and An and Schorfheide (2007) use micro data to inform the
selection of prior distributions. To illustrate the contributions of this
paper we use the production function of Olley and Pakes (1996) as a micro
foundation to estimate production function parameters from cross-sectional
data. We then use the parameters estimated from the cross-section to compute
aggregate productivity shocks using time series data. The estimated
productivity shocks now form the basis for time series estimates of the
persistency parameter of aggregate shocks. The asymptotic theory developed in
this paper provides the theoretical foundation to quantify the additional
sampling uncertainty introduced by estimated, rather than observed aggregate
shocks. The challenge specifically arises from the combination of two
distinct, yet not independent data sources.

With the objective of creating a framework to perform inference in our general
model, we first derive a joint functional stable central limit theorem for
cross-sectional and time-series data. The central limit theorem explicitly
accounts for the factor-induced dependence between the two samples even when
the cross-sectional sample is obtained by random sampling, a special case
covered by our theory. We derive the central limit theorem under the condition
that the dimension of the cross-sectional data $n$ as well as the dimension of
the time series data $\tau$ go to infinity. Using our central limit theorem we
then derive the asymptotic distribution of the parameter estimators that
characterize our general model. To our knowledge, this is the first paper that
derives an asymptotic theory that combines cross sectional and time series
data. In order to deal with parameters estimated using two data sets of
completely different nature, we adopt the notion of stable convergence. Stable
convergence dates back to R\'{e}nyi (1963) and was recently used in
Kuersteiner and Prucha (2013) in a panel data context to establish joint
limiting distributions. Using this concept, we show that the asymptotic
distributions of the parameter estimators are a combination of asymptotic
distributions from cross-sectional analysis and time-series analysis.

Stable convergence has found wide applications in many areas of statistics.
Stable convergence was introduced into the econometrics literature by Phillips
and Ouliaris (1990) who discuss the concept in detail and draw the connection
to central limit theorems by McLeish (1975b) and Hall and Heyde (1980). The
insight that panel limit theory often involves invariant sigma algebras
generated by time series first appears in Phillips and Sul (2003) and is the
basis for work by Andrews (2005) and Kuersteiner and Prucha (2013).

Another area in econometrics where stable convergence plays a prominent role
are high frequency financial models\footnote{We thank one of the referees for
bringing this literature to our attention.}. Important references to the high
frequency literature include Barndorff-Nielsen, Hansen, Lunde and Shephard
(2008), Jacod, Podolskij and Vetter (2010) and Li and Xiu (2016). Monographs
treating the probability theoretic foundation for this line of research are
Jacod and Shiryaev (2003, henceforth JS), and Jacod and Protter (2012). In
line with the literature on high frequency financial econometrics, we base our
definition of stable convergence in function spaces on JS. In the introduction
to their book, JS outline two different strategies of proving central limit
theorems. One is what they term the `martingale method' which dates back at
least to Stroock and Varadhan (1979). The other is the approach followed by
Billingsley (1968) which consists of establishing tightness, finite
dimensional convergence and identification of the finite dimensional limiting
distribution. In this paper we follow the second approach while JS and the
cited papers in the high frequency literature rely on the martingale method.
The reason for the difference in approach lies in the fact that the primitive
parameters central to the martingale method, essentially a set of conditional
moments called the triplets of characteristics of the approximating and
limiting process in the language of JS, have no obvious analog in the models
we study. The data generating processes (DGP) producing our samples are not
embedded in limiting diffusion processes. Even if our data generating process
could be represented by or approximated with such diffusions, the parameters
of interest in our applications are not directly related to such approximations.

Our limit theory has a second connection to the high frequency literature.
Barndorff-Nielsen et. al. (2008, Proposition 5) and Li and Xiu (2016, Lemma
A3), develop methods to deduce joint stable convergence of two random
sequences from stable convergence of the first random sequence and conditional
convergence in law of the second random sequence. The assumptions made in
these papers, namely that the noise terms are mean zero conditional on the
entire time series are stronger than the assumptions we make in our paper. It
is often unnatural to condition on the entire time series of common shocks,
which we avoid in this paper.

While the formal derivation of the asymptotic distribution may appear
complicated, the asymptotic formulae that we produce are straightforward to
implement in some special cases and very similar to the standard Murphy and
Topel's (1985) formula.

We also derive a novel result related to the unit root literature. We show
that, when the time-series data are characterized by unit roots, the
asymptotic distribution is a combination of a normal distribution and the
distribution found in the unit root literature. Therefore, the asymptotic
distribution exhibits mathematical similarities to the inferential problem in
predictive regressions, as discussed by Campbell and Yogo (2006). However, the
similarity is superficial in that Campbell and Yogo's (2006) result is about
an estimator based on a single data source. But, similarly to Campbell and
Yogo's analysis, we need to address problems of uniform inference. Phillips
(2014) proposes a method of uniform inference for predictive regressions,
which we adopt and modify to our own estimation problem in the unit root case.

Our results should be of interest to both applied microeconomists and
macroeconomists. Data combination is common practice in the macro calibration
literature where typically a subset of parameters is determined based on
cross-sectional studies. It is also common in structural microeconomics where
the focus is more directly on identification issues that cannot be resolved in
the cross-section alone. In a companion paper, Hahn, Kuersteiner, and Mazzocco
(2020, HKM20 henceforth), we discuss in detail specific examples from the
micro literature. In that paper, we focus on identification and provide an
intuitive explanation of inference with combined cross-sectional and
time-series data when aggregate factors are present. The purpose of this paper
is to prove the asymptotic theory needed for inference rigorously.

The remainder of the paper is organized as follows. In Section 2, we present
the Olley and Pakes model and illustrate how cross-section data can be used to
help estimate the persistence of aggregate productivity shocks from
time-series data. In Section \ref{section-models}, we introduce the general
statistical model. Our main central limit theorem is presented in Section
\ref{Section_JointCLT}. In Section \ref{section-intuition} we discuss
inference. Section \ref{sec-unit-root-time-series-models} contains the
analysis of the unit root case.

\section{Aggregate Productivity Shocks\label{OP}}

In Hahn, Kuersteiner and Mazzocco (2020), we considered Olley and Pakes'
(1996) method of estimating production functions, and argued that a subset of
parameters may be consistently estimable from the cross-section alone even
under the presence of aggregate shocks. In this section, we argue that such
a\ result is limited to only a subset of parameters and we should in general
have access to both time series and cross sectional data sets to identify the
full set of parameters, especially when the interest is in parameters that
characterize the aggregate shock process.

The general model introduced in Section \ref{section-models} is flexible
enough to both cover cases where the aggregate shocks are treated as nuisance
quantities and cases where the aggregate shocks are treated as parameters or
variables to be estimated. The example in this section can be understood to be
the latter case.\textbf{ }In our companion papers HKM20 and HKM21 the
interested reader can find additional fully worked examples, some of them
simpler and others more complex than the model discussed in this paper.

The essence of the example we discuss in this paper is as follows. The
parameter of interest is the first order autoregressive parameter
$\alpha^{\left(  A\right)  }$ of an aggregate shock process $\nu_{t}.$
Aggregate shocks $\nu_{t}$ are unobserved but can be recovered from aggregate
output $Y_{t}^{\ast}$ and aggregate capital $K_{t}^{\ast}$ through the
relationship $\nu_{t}=Y_{t}^{\ast}-\beta_{k}K_{t}^{\ast}.$ The parameter
$\beta$ is only identified from cross-sectional data. This approach is
justified by the theory of Olley and Pakes (1996) and explained in more detail
below. Inference for $\alpha^{\left(  A\right)  }$ then is complicated by two
aspects: the estimator $\hat{\alpha}^{\left(  A\right)  }$ is based on
estimated data $\hat{\nu}_{t}=Y_{t}^{\ast}-\hat{\beta}_{k}K_{t}^{\ast}$ and
$\hat{\beta}$ is estimated on a cross-sectional data set that is different
from the aggregate data $\left(  Y_{t}^{\ast},K_{t}^{\ast}\right)  .$ Our
paper provides the rigorous foundation for inference in this setting.\textbf{
}

We now present a simplified version of Olley and Pakes' (1996) model. A
profit-maximizing firm $j$ produces a product $Y_{j,t}$ in period $t,$
employing a production function that depends on the logarithm of labor
$l_{j,t}$, the logarithm of capital $k_{j,t}$, and a productivity shock
$\omega_{j,t}$. By denoting the logarithm of $Y_{j,t}$ by $\mathfrak{y}_{j,t}%
$, the production function takes the following form:
\begin{equation}
\mathfrak{y}_{j,t}=\beta_{l}l_{j,t}+\beta_{k}k_{j,t}+\omega_{j,t}+\eta_{j,t},
\label{production function}%
\end{equation}
where $\omega_{j,t}$ is a productivity shock and $\eta_{j,t}$ is a zero mean
measurement error with finite variance and iid over $j$ and $t.$\textbf{ }The
intercept term is normalized to be zero.\textbf{ }In each period capital
accumulates according to the equation $k_{j,t+1}=\left(  1-\delta\right)
k_{j,t}+i_{j,t}$, where $\delta$ is the rate at which capital depreciates. We
abstract from age heterogeneity and exit decision. As in Olley and Pakes
(1996), we assume that the optimal investment decision in period $t$ is a
function of the current stock of capital and of the productivity shock, i.e.
\begin{equation}
i_{j,t}=i_{t}\left(  \omega_{j,t},k_{j,t}\right)  .
\label{investment decision}%
\end{equation}
Olley and Pakes (1996) use the result that the investment decision
(\ref{investment decision}) is strictly increasing in the productivity shock
for every value of capital to invert (\ref{investment decision}), solve for
the productivity shock, and obtain
\begin{equation}
\omega_{j,t}=h_{t}\left(  i_{j,t},k_{j,t}\right)  .
\label{productivity shock production function}%
\end{equation}
One can then replace the productivity shock in the production function using
equation (\ref{productivity shock production function}) to obtain
\begin{equation}
\mathfrak{y}_{j,t}=\beta_{l}l_{j,t}+\phi_{t}\left(  i_{j,t},k_{j,t}\right)
+\eta_{j,t}, \label{production function no productity shock}%
\end{equation}
where
\begin{equation}
\phi_{t}\left(  i_{j,t},k_{j,t}\right)  =\beta_{k}k_{j,t}+h_{t}\left(
i_{j,t},k_{j,t}\right)  . \label{definition of phi}%
\end{equation}
For simplicity, we will assume that $\beta_{l}$ and $\phi_{t}\left(
i_{j,t},k_{j,t}\right)  $ are known\footnote{Olley and Pakes (1996) identifies
the parameter $\beta_{l}$ by
\[
\beta_{l}=\frac{E\left[  \left(  l_{j,t}-E\left[  \left.  l_{j,t}\right\vert
i_{j,t},k_{j,t}\right]  \right)  \left(  \mathfrak{y}_{j,t}-E\left[  \left.
\mathfrak{y}_{j,t}\right\vert i_{j,t},k_{j,t}\right]  \right)  \right]
}{E\left[  \left(  l_{j,t}-E\left[  \left.  l_{j,t}\right\vert i_{j,t}%
,k_{j,t}\right]  \right)  ^{2}\right]  },
\]
which can be consistently estimated by cross sectional variation. In Hahn,
Kuersteiner and Mazzocco (2020), it was shown that the cross sectional
variation identifies the $\beta_{l}$ even under the presence of aggregate
shocks. We abstract away from the estimation of $\beta_{l}$ because the above
method of identification was critiqued for substantive economic reasons, for
example, by Ackerberg, Caves and Frazer (2015), and as such, researchers may
prefer other methods of estimation.}, and work with $\mathfrak{y}_{j,t}^{\ast
}\equiv\mathfrak{y}_{j,t}-\beta_{l}l_{j,t}$.

We now introduce an aggregate shock\footnote{See Appendix \ref{OP_detail} for
intuition of the moment condition for the simple case without any aggregate
shock.}, and assume that the productivity shock at $t$ is the sum of an
aggregate shock $\nu_{t}$ drawn from a distribution $F\left(  \nu
|\alpha\right)  $ and of an idiosyncratic shock $\varepsilon_{j,t}$
independent of $\nu_{t}$, i.e.,
\begin{equation}
\omega_{j,t}=\nu_{t}+\varepsilon_{j,t}. \label{composit_shock}%
\end{equation}
Unlike HKM20, we assume that the firm observes $\omega_{j,t}$ but not $\nu
_{t}$ and $\varepsilon_{j,t}$ separately\footnote{In Hahn, Kuersteiner and
Mazzocco (2020), we assumed that $\nu_{t}$ and $\varepsilon_{j,t}$ are both
Markov processes and that the firm observes the realization of the aggregate
shock and, separately, of the idiosyncratic shock. This is an assumption of
convenience to be consistent with Olley Pakes' (1996) assumption that the
problem solved by the firm is Markovian. To understand why, consider a case in
which $\nu_{t}$ and $\varepsilon_{j,t}$ are both AR(1) processes. If we only
use their sum as a state variable, the Markovian assumption is generally
violated, because the sum of AR(1) processes is in general not an AR(1) but an
ARMA(2,1) process. However, if we include $\nu_{t}$ and $\varepsilon_{j,t}$ as
separate state variables -- both observed by the firm -- the Markovian
structure is preserved.}. We assume that $\nu_{t}$ and $\varepsilon_{j,t}$ are
both Markov processes and in particular, we assume that both $\nu_{t}$ and
$\varepsilon_{j,t}$ are AR(1):%
\begin{align*}
\nu_{t}  &  =\alpha^{\left(  A\right)  }\nu_{t-1}+e_{t}^{\left(  A\right)
},\\
\varepsilon_{j,t}  &  =\alpha^{\left(  C\right)  }\varepsilon_{j,t-1}%
+e_{j,t}^{\left(  C\right)  },
\end{align*}
where we assumed that the intercepts are zero for notational simplicity such
that $\nu_{t}$ and $\varepsilon_{j,t}$ have mean zero.

It can be shown\footnote{See Appendix \ref{OP_detail} for details that lead to
(\ref{C-GMM}).} that some of the parameter can be identified by using the
cross section GMM estimator based on the moments
\begin{equation}
0=E\left[  z_{j,t}\left(  \mathfrak{y}_{j,t+1}^{\ast}-\left(  \beta
_{0,t+1}^{\ast}+\beta_{k}k_{j,t+1}+\alpha^{\left(  C\right)  }\left(  \phi
_{t}\left(  i_{j,t},k_{j,t}\right)  -\beta_{k}k_{j,t}\right)  \right)
\right)  \right]  , \label{C-GMM}%
\end{equation}
where $\beta_{0,t+1}^{\ast}\equiv\nu_{t+1}-\alpha^{\left(  C\right)  }\nu_{t}$
and the $z_{j,t}$ is an instrument uncorrelated with the error $e_{j,t+1}%
^{\left(  C\right)  }+\eta_{j,t+1}$. Note that identification of the
parameters $\left(  \beta_{0,t+1}^{\ast},\beta_{k},\alpha^{\left(  C\right)
}\right)  $ requires that the $z_{j,t}$ should contain at least three
components. The key cross-sectional parameter of interest for the application
we have in mind is $\beta_{k}$ while the remaining cross-sectional parameters
$\left(  \beta_{0,t+1}^{\ast},\alpha^{\left(  C\right)  }\right)  $ are
incidental to our final goal of estimating the degree of persistence of the
aggregate shock.

The parameter $\alpha^{\left(  A\right)  }$ is not identified by the above
procedure based on cross sectional variation. On the other hand,
$\alpha^{\left(  A\right)  }$ can be estimated consistently, possibly with the
help of production function parameters estimated in the cross-section, if
aggregate time series data with information about $\nu_{t}$ is available. For
example, if an econometrician observes $\left\{  \left(  Y_{t}^{\ast}%
,K_{t}^{\ast}\right)  ,t=1,\ldots,\tau\right\}  $,\textbf{ }where $\tau$ is
the time series sample size and $Y_{t}^{\ast}\equiv\operatorname*{plim}%
_{n\rightarrow\infty}n^{-1}\sum_{j=1}^{n}\mathfrak{y}_{j,t}^{\ast}$ and
$K_{t}^{\ast}\equiv\operatorname*{plim}_{n\rightarrow\infty}n^{-1}\sum
_{j=1}^{n}k_{j,t}^{\ast}$, then equation (\ref{production function}) combined
with (\ref{composit_shock}) implies an aggregate relationship between output
and capital of the form $Y_{t}^{\ast}-\beta_{k}K_{t}^{\ast}=\nu_{t}$. Implicit
in this formulation is the assumption that firm-specific shocks $\varepsilon
_{j,t}$ and $\eta_{j,t}$ average out in the aggregate. If $\nu_{t}$ is
estimated by $\hat{\nu}_{t}=Y_{t}^{\ast}-\hat{\beta}_{k}K_{t}^{\ast}$, then
the parameter $\alpha^{\left(  A\right)  }$ can be consistently estimated by
an AR(1) regression of $Y_{t}^{\ast}-\hat{\beta}_{k}K_{t}^{\ast}$ on a
constant and its own lagged value (using $\hat{\beta}_{k}$ estimated from the
cross sectional data) as long as $\tau$ is sufficiently large. The example
illustrates how cross-sectional data and micro parameters recovered from it
can be used to understand dynamic aggregate parameters.

The fact that $\alpha^{\left(  A\right)  }$ is estimated based on $\hat{\nu
}_{t}$ rather than $\nu_{t}$ creates an estimated regressor problem that
affects the limiting distribution of the estimator for $\alpha^{\left(
A\right)  }.$ Unlike in classical estimated regressor problems, this paper
considers the case where two samples that are not necessarily independent of
each other are used to construct $\hat{\nu}_{t}$ and estimate $\alpha^{\left(
A\right)  }.$

\section{Model and Probability Space\label{section-models}}

In this section we present a general modeling framework that includes the
example in Section \ref{OP} as well as models considered in HKM20 and HKM21 as
special cases. We assume that our cross-sectional data consist of $\left\{
y_{i,t},\;i=1,\ldots,n,t=1,...,T\right\}  $,\footnote{We do not consider
models with estimated fixed effects in this paper because we assume that the
parameter space is finite dimensional\textbf{.}} where the start time of the
cross-section or panel, $t=1,$ is an arbitrary normalization of time. Pure
cross-sections are handled by allowing for $T=1.$ Note that $T$ is fixed and
finite throughout our discussion while our asymptotic approximations are based
on $n$ tending to infinity. The need to keep $T$ fixed is motivated by short
cross-sectional panels and is critical for our theoretical development.
Without this assumption, more complicated asymptotic approximations allowing
for expanding parameter spaces as well as different limit theorems are
required. The extensions are left for future work.

Our time series data consist of $\left\{  z_{s},\;s=\tau_{0}+1,\ldots,\tau
_{0}+\tau\right\}  $ where the time series sample size $\tau$ tends to
infinity jointly with $n$. The start point of the time series sample is either
fixed at an arbitrary time $\tau_{0}$ such that $-\infty<-K\leq\tau_{0}\leq
K<\infty$ for some bounded $K$ and $\tau\rightarrow\infty$ or $\tau_{0}%
=\tau_{0}\left(  \tau\right)  $ depends on $\tau$ such that $\tau_{0}\left(
\tau\right)  =-\upsilon\tau+\tau_{0,f}+T$ for $\upsilon\in\left[  0,1\right]
$ and $\tau_{0,f}$ a fixed constant. In the latter case we use the short hand
notation $\tau_{0}$ when no confusion arises. The fixed $\tau_{0}$ scenario
corresponds to a situation where a (hypothetical) time series sample is
observed into the infinite future. The specification $\tau_{0}\left(
\tau\right)  $ covers the case $\tau_{0}=-\tau+T$. The case where $\tau_{0}$
varies with the time series sample size in the prescribed way can be used to
model situations where the asymptotics are carried out `backwards' in time
(when $\upsilon=1$) or where the panel data are located at a fixed fraction of
the time series sample as the time series sample size tends to infinity
$\left(  0<\upsilon<1\right)  .$When $\upsilon=1$ such that $\tau
_{0}\rightarrow-\infty$ the end of the time series sample is fixed at
$\tau_{0,f}+T.$ A hybrid case arises when $\upsilon\in\left(  0,1\right)  $
such that the start point $\tau_{0}$ of the time series sample extends back in
time simultaneously with the last observation in the sample $\tau_{0}+\tau$
tending to infinity. For simplicity we refer to both scenarios $\upsilon
\in\left(  0,1\right)  $ and $\upsilon=1$ as backwards asymptotics. Backwards
asymptotics may be more realistic in cases where recent panel data is
augmented with long historical records of time series data. We show that under
a mild additional regularity condition both forward and backward asymptotics
lead to the same limiting distribution. Conventional asymptotics are covered
by setting $\upsilon=0$. The vector $y_{i,t}$ includes all information related
to the cross-sectional submodel, where $i$ is an index for individuals,
households or firms, and $t$ denotes the time period when the cross-sectional
unit is observed. The second vector $z_{s}$ contains aggregate data.

The technical assumptions for our CLT, detailed in Section
\ref{Section_JointCLT}, do not directly restrict the data, nor do they impose
restrictions on how the data were sampled. For example, we do not assume that
the cross-sectional sample was obtained by randomized sampling, although this
is a special case that is covered by our assumptions. Rather than imposing
restrictions directly on the data we postulate that there are two parametrized
models that implicitly restrict the data. The function $f\left(  \left.
y_{i,t}\right\vert \beta,\nu_{t},\rho\right)  $ is used to model $y_{i,t}$ as
a function of cross-sectional parameters $\beta$ and common shocks $\nu
\equiv\left(  \nu_{1},...,\nu_{T}\right)  $\textbf{ }which are treated as
parameters to be estimated, and time series parameters $\rho$. In the same way
the function $g\left(  \left.  z_{s}\right\vert \beta,\rho\right)  $ restricts
the behavior of some time series variables $z_{s}$.\footnote{The function $g$
may naturally arise if the $\nu_{t}$ is an unobserved component that can be
estimated from the aggregate time series once the parameters $\beta$ and
$\rho$ are known, i.e., if $\nu_{t}\equiv\nu_{t}\left(  \beta,\rho\right)  $
is a function of $\left(  z_{t},\beta,\rho\right)  $ and the behavior of
$\nu_{t}$ is expressed in terms of $\rho$. Later, we allow for the possibility
that $g$ in fact is derived from the conditional density of $\nu_{t}$ given
$\nu_{t-1}$, i.e., the possibility that $g$ may depend on both the current and
lagged values of $z_{t}$. For notational simplicity, we simply write $g\left(
\left.  z_{s}\right\vert \beta,\rho\right)  $ here for now.}

Depending on the exact form of the underlying economic model, the functions
$f$ and $g$ may have different interpretations. They could be the
log-likelihoods of $y_{i,t},$ conditional on $\nu_{t},$ and $z_{s}$
respectively. In a likelihood setting, $f$ and $g$ impose restrictions on
$y_{i,t}$ and $z_{s}$ because of the implied martingale properties of the
score process evaluated at the true parameter values. More generally, the
functions $f$ and $g$ may be the basis for method of moments (the exactly
identified case) or GMM (the overidentified case) estimation. In these
situations parameters are identified from the conditions $E_{C}\left[
f\left(  \left.  y_{i,t}\right\vert \beta,\nu_{t},\rho\right)  \right]  =0$
given the shock $\nu_{t}$ and $E_{\tau}\left[  g\left(  z_{s}|\beta
,\rho\right)  \right]  =0.$ The first expectation, $E_{C},$ is understood as
being over the cross-section population distribution holding $\nu=\left(
\nu_{1},...,\nu_{T}\right)  $ fixed, while the second, $E_{\tau},$ is over the
distribution of the time-series data generating process. The moment conditions
follow from martingale assumptions we directly impose on $f$ and $g.$ In our
companion paper HKM20 we discuss examples of economic models that rationalize
these assumptions.

Whether we are dealing with likelihoods or moment functions, the central limit
theorem is directly formulated for the estimating functions that define the
parameters. We use the notation $F_{n}\left(  \beta,\nu,\rho\right)  $ and
$G_{\tau}\left(  \beta,\rho\right)  $ to denote the criterion function based
on the cross-section and time series respectively. When the model specifies a
log-likelihood these functions are defined as $F_{n}\left(  \beta,\nu
,\rho\right)  =\frac{1}{n}\sum_{t=1}^{T}\sum_{i=1}^{n}f\left(  \left.
y_{i,t}\right\vert \beta,\nu_{t},\rho\right)  $ and $G_{\tau}\left(
\beta,\rho\right)  =\frac{1}{\tau}\sum_{s=\tau_{0}+1}^{\tau_{0}+\tau}g\left(
\left.  z_{s}\right\vert \beta,\rho\right)  .$ When the model specifies moment
conditions we let $h_{n}\left(  \beta,\nu,\rho\right)  $ $=$ $\frac{1}{n}%
\sum_{t=1}^{T}\sum_{i=1}^{n}$ $f\left(  \left.  y_{i,t}\right\vert \beta
,\nu_{t},\rho\right)  $ and $k_{\tau}\left(  \beta,\rho\right)  $ $=$
$\frac{1}{\tau}\sum_{s=\tau_{0}+1}^{\tau_{0}+\tau}g\left(  \left.
z_{s}\right\vert \beta,\rho\right)  $. The GMM\ or moment based criterion
functions are then given by $F_{n}\left(  \beta,\nu,\rho\right)
=-h_{n}\left(  \beta,\nu,\rho\right)  ^{\prime}W_{n}^{C}h_{n}\left(  \beta
,\nu,\rho\right)  $ and $G_{\tau}\left(  \beta,\rho\right)  =-k_{\tau}\left(
\beta,\rho\right)  ^{\prime}W_{\tau}^{\tau}k_{\tau}\left(  \beta,\rho\right)
$ with $W_{n}^{C}$ and $W_{\tau}^{\tau}$ two almost surely positive definite
weight matrices. The use of two separate objective functions is helpful in our
context because it enables us to discuss which issues arise if only
cross-sectional variables or only time-series variables are used in the
estimation.\footnote{Note that our framework covers the case where the joint
distribution of $\left(  y_{it},z_{t}\right)  $ is modeled. Considering the
two components separately adds flexibility in that data is not required for
all variables in the same period.}

We formally justify the use of two data sets by imposing restrictions on the
identifiability of parameters through the cross-section and time series
criterion functions alone. Let $\Theta$ be a compact set constructed from the
product $\Theta=\Theta_{\beta}\times\Theta_{\nu}\times\Theta_{\rho}$ where
$\Theta_{\beta}$ is a compact set that contains the true parameter value
$\beta_{0},$ $\Theta_{\nu}$ is a compact set that contains the true parameter
value $\nu_{0}$ and $\Theta_{\rho}$ is a compact set that contains the true
parameter value $\rho_{0}.$ We denote the probability limit of the objective
functions by $F\left(  \beta,\nu_{t},\rho\right)  $ and $G\left(  \beta
,\rho\right)  ,$ in other words,
\begin{align*}
F\left(  \beta,\nu,\rho\right)   &  =\operatorname*{plim}_{n\rightarrow\infty
}F_{n}\left(  \beta,\nu,\rho\right)  ,\\
G\left(  \beta,\rho\right)   &  =\operatorname*{plim}_{\tau\rightarrow\infty
}G_{\tau}\left(  \beta,\rho\right)  .
\end{align*}
The true or pseudo true parameters are defined as the maximizers of these
probability limits%
\begin{align}
\left(  \beta\left(  \rho\right)  ,\nu\left(  \rho\right)  \right)   &
\equiv\operatorname*{argmax}_{\beta,\nu\in\Theta_{\beta}\times\Theta_{\nu}%
}F\left(  \beta,\nu,\rho\right)  ,\label{Cross-Par}\\
\rho\left(  \beta\right)   &  \equiv\operatorname*{argmax}_{\rho\in
\Theta_{\rho}}G\left(  \beta,\rho\right)  , \label{Time-Par}%
\end{align}
and we denote with $\beta_{0}$ and $\rho_{0}$ the solutions to
(\ref{Cross-Par}) and (\ref{Time-Par}). The idea that neither $F$ nor $G$
alone are sufficient to identify both parameters is formalized as follows. If
the function $F$ is constant in $\rho$ at the parameter values $\beta$ and
$\nu$ that maximize it then $\rho$ is not identified by the criterion $F$
alone. Formally we state that%
\begin{equation}
\max_{\beta,\nu\in\Theta_{\beta}\times\Theta_{\nu}}F\left(  \beta,\nu
,\rho\right)  =\max_{\beta,\nu\in\Theta_{\beta}\times\Theta_{\nu}}F\left(
\beta,\nu,\rho_{0}\right)  \quad\text{for all }\rho\in\Theta_{\rho}
\label{no cross-sectional identification1}%
\end{equation}
It is easy to see that (\ref{no cross-sectional identification1}) is not a
sufficient condition to restrict identification in a desirable way. For
example (\ref{no cross-sectional identification1}) is satisfied in a setting
where $F$ does not depend at all on $\rho$. In that case the maximizers in
(\ref{Cross-Par}) also do not depend on $\rho$ and by definition coincide with
$\beta_{0}$ and $\nu_{0}.$ To rule out this case we require that $\rho_{0}$ is
needed to identify $\beta_{0}$ and $\nu_{0}$. Formally, we impose the
condition that
\begin{equation}
\left(  \beta\left(  \rho\right)  ,\nu\left(  \rho\right)  \right)
\neq\left(  \beta_{0},\nu_{0}\right)  \quad\text{for all }\rho\neq\rho_{0}.
\label{no cross-sectional identification2}%
\end{equation}
Similarly, we impose restrictions on the time series criterion functions that
insure that the parameters $\beta$ and $\rho$ cannot be identified solely as
the maximizers of $G.$ Formally, we require that
\begin{align}
\max_{\rho\in\Theta_{\rho}}G\left(  \beta,\rho\right)   &  =\max_{\rho
\in\Theta_{\rho}}G\left(  \beta_{0},\rho\right)  \quad\text{for all }\beta
\in\Theta_{\beta},\label{no time series identification}\\
\rho\left(  \beta\right)   &  \neq\rho_{0}\quad\text{for all }\beta\neq
\beta_{0}.\nonumber
\end{align}
To insure that the parameters can be identified from a combined
cross-sectional and time-series data set we impose the following condition.
Define $\theta\equiv\left(  \beta^{\prime},\nu^{\prime}\right)  ^{\prime}$ and
assume that (i) there exists a unique solution to the system of equations:
\begin{equation}
\left[  \frac{\partial F\left(  \beta,\nu,\rho\right)  }{\partial
\theta^{\prime}},\;\frac{\partial G\left(  \beta,\rho\right)  }{\partial
\rho^{\prime}}\right]  =0, \label{moment-of-the-estimator}%
\end{equation}
and (ii) the solution is given by the true value of the parameters. In
summary, our model is characterized by the high level assumptions in
(\ref{no cross-sectional identification1}),
(\ref{no cross-sectional identification2}),
(\ref{no time series identification}), and by the assumption that
(\ref{moment-of-the-estimator}) only has one solution at the true parameter
values.\footnote{Througout this paper, we assume that the parameters are not
identified using cross section and time series datasets alone. This seems to
be the main reason for data combination. However, the proposed method of
inference in this paper is effectively based on the moments in
(\ref{moment-of-the-estimator}). This implies that our method of inference
still works even if the parameters are identified from just one data set.}

In order to accurately describe the theory that follows, we start with a
precise definition of the probability space used for our theory. Let $\left(
\Omega^{\prime},\mathcal{G}^{\prime},P^{\prime}\right)  $ be a probability
space with random sequences $\left\{  z_{t}\right\}  _{t=-\infty}^{\infty}%
$\textbf{ }and\textbf{ }$\left\{  y_{it}\right\}  _{i=1,t=-\infty}%
^{\infty,\infty}$.\textbf{ }The observed sample $\left\{  z_{t}\right\}
_{t=\tau_{0}+1}^{\tau_{0}+\tau}$ and $\left\{  y_{i1},\ldots,y_{iT}\right\}
_{i=1}^{n}$ is a subset of these random sequences. The process we analyze
consists of a triangular array of panel data $\psi_{n,it}^{y}$ where
$\psi_{n,it}^{y}$ typically is a function of $y_{it}$ and parameters, observed
for $i=1,...,n$ and $t=1,\ldots,T.$ We let $n\rightarrow\infty$ while $T$ is
fixed and $t=1$ is an arbitrary normalization of time at the beginning of the
cross-sectional sample. It also consists of a separate triangular array of
time series $\psi_{\tau,s}^{\nu},$ where $\psi_{\tau,s}^{\nu}$ is a function
of $z_{s}$ and parameters for $s=\tau_{0}+1,\ldots,\tau_{0}+\tau.$ In typical
applications, $\psi_{n,it}^{y}$ and $\psi_{\tau,s}^{\nu}$ are the influence
functions of the cross-section and time series estimators.

We now form the triangular array of filtrations similarly to Kuersteiner and
Prucha (2013). The filtrations are a theoretical construct defined on the
probability space in such a way that the observed sample is a strict subset of
the random variables that generate the filtrations. We proceed by first
collecting information about all of the common shocks, $\left(  \nu
_{1},...,\nu_{T}\right)  $, then we pick the initial time-series realization
$z_{\min\left\{  1,\tau_{0}\right\}  }$ in such a way that it predates or
coincides with the initial period of the panel data set. We then sequentially
add the cross-sectional units for the same time period, starting from $i=1$ to
$i=n$.\footnote{The filtrations are constructed based on the ordering of the
cross-sectional sample. However, at the cost of slightly stronger moment
conditions, the $\sigma$-fields can be constructed in a way that is invariant
to reordering of the cross-sectional sample. We refer the interested reader to
Kuersteiner and Prucha (2013), in particular Definition 1 and the related
discussion for details. Note that the stronger moment conditions needed for
the invariance property hold in situations where $y_{it}$ is cross-sectionally
independent conditional on $\left\{  z_{t},z_{t-1},...\right\}  \vee
\mathcal{C}$. The latter is a leading case in our examples.} Subsequently, the
same procedure is repeated by shifting the time index ahead by one period. The
process ends once the time index reaches $\max(T,\tau).$ If $\tau>T,$ which
eventually happens in forward asymptotics, but may also arise in backward
asymptotics, enlarge the filtration by adding random variables $y_{it}$ even
when $t>T$. These random variables are generated from the same distribution
that produces the observed cross-sectional sample but are not actually
included in the cross-sectional sample. This enlargement is used mostly for
notational convenience.

Formally, the filtrations are defined as follows. We use the binary operator
$\vee$ to denote the smallest $\sigma$-field that contains the union of two
$\sigma$-fields. Setting $\mathcal{C}=\sigma\left(  \nu_{1},...,\nu
_{T}\right)  $ we define
\begin{align}
\mathcal{G}_{\tau n,0}  &  =\mathcal{C}\label{Information Sets}\\
\mathcal{G}_{\tau n,i}  &  =\sigma\left(  z_{\min\left(  1,\tau_{0}\right)
},\left\{  y_{j,\min\left(  1,\tau_{0}\right)  }\right\}  _{j=1}^{i}\right)
\vee\mathcal{C}\nonumber\\
&  \vdots\nonumber\\
\mathcal{G}_{\tau n,n+i}  &  =\sigma\left(  \left\{  y_{j,\min\left(
1,\tau_{0}\right)  }\right\}  _{j=1}^{n},\left\{  z_{\min\left(  1,\tau
_{0}\right)  +1},z_{\min\left(  1,\tau_{0}\right)  }\right\}  ,\left\{
y_{j,\min\left(  1,\tau_{0}\right)  +1}\right\}  _{j=1}^{i}\right)
\vee\mathcal{C}\nonumber\\
&  \vdots\nonumber\\
\mathcal{G}_{\tau n,\left(  t-\min\left(  1,\tau_{0}\right)  \right)  n+i}  &
=\sigma\left(  \left\{  y_{j,t-1},y_{j,t-2},\ldots,y_{j,\min\left(  1,\tau
_{0}\right)  }\right\}  _{j=1}^{n},\left\{  z_{t},z_{t-1},\ldots
,z_{\min\left(  1,\tau_{0}\right)  }\right\}  ,\left\{  y_{j,t}\right\}
_{j=1}^{i}\right)  \vee\mathcal{C}\text{.}\nonumber
\end{align}
As noted before, the filtration $\mathcal{G}_{\tau n,tn+i}$ is generated by a
set of random variables that contain the observed sample as a strict subset.
More specifically, we note that the cross-section $y_{j,t}$ is only observed
in the sample for a finite number of time periods while the filtrations range
over the entire expanding time series sample period.\textbf{ }We use the
convention that $\mathcal{G}_{\tau n,\left(  t-\min\left(  1,\tau_{0}\right)
\right)  n}=\mathcal{G}_{\tau n,\left(  t-\min\left(  1,\tau_{0}\right)
-1\right)  n+n}.$ This implies that $z_{t}$ and $y_{1t}$ are added
simultaneously to the filtration $\mathcal{G}_{\tau n,\left(  t-\min\left(
1,\tau_{0}\right)  \right)  n+1}$. Also note that $\mathcal{G}_{\tau n,i}$
predates the time series sample by at least one period. To simplify notation
define the function $q_{n}\left(  t,i\right)  =(t-\min(1,\tau_{0}))n+i$ that
maps the two-dimensional index $\left(  t,i\right)  $ into the integers and
note that for $q=q_{n}\left(  t,i\right)  $ it follows that $q\in\left\{
0,\ldots,\max(T,\tau)n\right\}  $. The index $q$ orders the filtrations from
smallest to largest. The filtrations $\mathcal{G}_{\tau n,q}$ are increasing
in the sense that $\mathcal{G}_{\tau n,q}\subset\mathcal{G}_{\tau n,q+1}$ for
all $q,$ $\tau$ and $n$. However, they are not nested in the sense of Hall and
Heyde's (1980) Condition (3.21) for two reasons. One is the fact that we are
considering what essentially amounts to a panel structure generating the
filtration. Kuersteiner and Prucha (2013) provide a detailed discussion of
this aspect. A second reason has to do with the possible `backwards'
asymptotics adopted in this paper. Even in a pure time series setting, i.e.
omitting $y_{i,t}$ from the definitions (\ref{Information Sets}), backwards
asymptotics lead to a violation of Hall and Heyde's Condition (3.21) because
the definition of $\mathcal{G}_{\tau n,q}$ changes as $q$ is held fixed but
$\tau$ increases. The consequence of this is that unlike in Hall and Heyde
(1980) stable convergence cannot be established for the entire probability
space, but rather is limited to the invariant $\sigma$-field $\mathcal{C}$.
This limitation is the same as in Kuersteiner and Prucha (2013) and also
appears in Eagleson (1975), albeit due to very different technical reasons.
The fact that the definition of $\mathcal{G}_{\tau n,q}$ changes for fixed $q$
does not pose any problems for the proofs that follow, because the underlying
proof strategy explicitly accounts for triangular arrays and does not use Hall
and Heyde's Condition (3.21).

To better understand the construction of $\mathcal{G}_{\tau n,q}$ we refer the
reader to Kuersteiner and Prucha (2013, pp.112-114) for a discussion of why
the cross-sectional sample needs to be added to the filtration one at a time,
why the ordering of the cross-sectional sample is irrelevant under certain
regularity conditions and why the nesting condition (3.21)\ of Hall and Heyde
(1980) must fail in a panel context. Kuersteiner and Prucha (2013, Section
2.3) also provide a number of worked examples. The construction of
$\mathcal{G}_{\tau n,q}$ proposed in this paper extends Kuersteiner and Prucha
(2013) in two directions. On the one hand, an additional time series component
$z_{s}$ is part of the generating mechanism for $\mathcal{G}_{\tau n,q}.$ On
the other hand, the filtration expands because two indices, $\tau$ and $n$,
rather than just $n$ in the case of Kuersteiner and Prucha (2013), tend to
infinity. The construction of $\mathcal{G}_{\tau n,q}$ is specific to the type
of central limit theorem we prove and the fact that the joint process of
$\psi_{n,it}^{y}$ and $\psi_{\tau,t}^{\nu}$ needs to satisfy a martingale
difference property relative to the filtration $\mathcal{G}_{\tau n,q}$ for
our proof to be valid. The fact that $z_{t}$ and $y_{1t}$ are added
simultaneously to the filtration before the cross-section observations
$y_{2t},...,y_{nt}$ is necessitated by the possibility that $z_{t}$ and
$y_{jt}$ are not independent. To understand this point, consider a
hypothetical situation where $z_{t}$ were added at the end of the
cross-section sample together with $y_{nt}$. In such a scenario, it would no
longer be credible to impose the moment condition $E\left[  \psi_{\tau,t}%
^{\nu}|\mathcal{G}_{\tau n,q-1}\right]  =0$ for $q=q_{n}\left(  t,n\right)  $
with $t>T$ because $\mathcal{G}_{\tau n,q-1}$ now would depend on
$y_{1t},...,y_{n-1t}.$\textbf{ }These variables in turn may predict
$\psi_{\tau,t}^{\nu}$. Similarly, the need to develop partial sums over the
index $i$ for the component $\psi_{n,it}^{y}$ requires that $y_{it}$ be added
one at a time to the filtration $\mathcal{G}_{\tau n,q}$, a point also
explained in Kuersteiner and Prucha (2013).

\section{Joint Panel-Time Series Limit Theory\label{Section_JointCLT}}

In this section we first establish a generic joint limiting result for a
combined panel-time series process and then specialize it to the limiting
distributions of parameter estimates under stationarity and, in a later
section, non-stationarity.

We develop asymptotic theory for the sums of some generic random vectors
$\psi_{n,it}^{y}$ and $\psi_{\tau,t}^{\nu}$. Typically, $\psi_{n,it}^{y}$and
$\psi_{\tau,t}^{\nu}$ are the scores or moment functions of a cross-section
and time series criterion function based on observed data $y_{it}$ and $z_{t}%
$.\ Below we introduce general regularity conditions for these generic random
vectors $\psi_{n,it}^{y}$ and $\psi_{\tau,t}^{\nu}$. Let $k_{\theta}$ be the
dimension of the parameter $\theta$ and $k_{\rho}$ be the dimension of the
parameter $\rho.$ With some abuse of notation we also denote by $k_{\theta}$
the number of moment conditions used to identify $\theta$ when GMM estimators
are used, with a similar convention applying to $k_{\rho}.$ With this
notation, $\psi_{n,it}^{y}$ takes values in $R^{k_{\theta}}$ and $\psi
_{\tau,t}^{\nu}$ takes values in $R^{k_{\rho}}.$ We assume that $T\leq\tau
_{0}+\tau$. \textbf{ }Throughout we assume that $\left(  \psi_{n,it}^{y}%
,\psi_{\tau,t}^{\nu}\right)  $ is a type of a vector mixingale sequence
relative to a filtration $\mathcal{G}_{\tau n,q}$. The concept of mixingales
was introduced by Gordin (1969, 1973) and McLeish (1975a). We derive the joint
limiting distribution and a related functional central limit theorem for
$\frac{1}{\sqrt{n}}\sum_{t=1}^{T}\sum_{i=1}^{n}\psi_{n,it}^{y}$ and $\frac
{1}{\sqrt{\tau}}\sum_{t=\tau_{0}+1}^{\tau_{0}+\tau}\psi_{\tau,t}^{\nu}$.

The central limit theorem we develop needs to establish joint convergence for
terms involving both $\psi_{n,it}^{y}$ and $\psi_{\tau,t}^{\nu}$ with both the
time series and the cross-sectional dimension becoming large simultaneously.
Let $\left[  a\right]  $ be the largest integer less than or equal $a$. Joint
convergence is achieved by stacking both moment vectors into a single sum that
extends over both $t$ and $i.$ Let $r\in\left[  0,1\right]  $ and define
\begin{equation}
\tilde{\psi}_{it}^{\nu}\left(  r\right)  \equiv\frac{\psi_{\tau,t}^{\nu}%
}{\sqrt{\tau}}1\left\{  \tau_{0}+1\leq t\leq\tau_{0}+\left[  \tau r\right]
\right\}  1\left\{  i=1\right\}  , \label{psi_tilde_z}%
\end{equation}
which depends on $r$ in a non-trivial way. This dependence will be of
particular interest when we specialize our models to the near unit root case.
For cross-sectional data define%
\begin{equation}
\tilde{\psi}_{it}^{y}\left(  r\right)  \equiv\tilde{\psi}_{it}^{y}\equiv
\frac{\psi_{n,it}^{y}}{\sqrt{n}}1\left\{  1\leq t\leq T\right\}
\label{psi_tilde_y}%
\end{equation}
where $\tilde{\psi}_{it}^{y}\left(  r\right)  =\tilde{\psi}_{it}^{y}$ is
constant as a function of $r\in\left[  0,1\right]  .$ In turn, this implies
that functional convergence of the component (\ref{psi_tilde_y}) is the same
as the finite dimensional limit. It also means that the limiting process is
degenerate (i.e. constant) when viewed as a function of $r$. However, this
does not matter in our applications as we are only interested in the sample
averages
\[
\frac{1}{\sqrt{n}}\sum_{t=1}^{T}\sum_{i=1}^{n}\psi_{n,it}^{y}=\sum
_{t=\min(1,\tau_{0}+1)}^{\max(T,\tau_{0}+\tau)}\sum_{i=1}^{n}\tilde{\psi}%
_{it}^{y}\equiv X_{n\tau}^{y}.
\]
Define the stacked vector $\tilde{\psi}_{it}\left(  r\right)  =\left(
\tilde{\psi}_{it}^{y}\left(  r\right)  ^{\prime},\tilde{\psi}_{it}^{\nu
}\left(  r\right)  ^{\prime}\right)  ^{\prime}\in\mathbb{R}^{k_{\phi}}$ where
$\phi=\left(  \theta,\rho\right)  $ and $k_{\phi}$ is the dimension of $\phi
$\textbf{. }Consider the stochastic process%
\begin{equation}
X_{n\tau}\left(  r\right)  =\sum_{t=\min(1,\tau_{0}+1)}^{\max(T,\tau_{0}%
+\tau)}\sum_{i=1}^{n}\tilde{\psi}_{it}\left(  r\right)  ,\text{\quad}X_{n\tau
}\left(  0\right)  =\left(  X_{n\tau}^{y\prime},0\right)  ^{\prime}.
\label{CLT_Process}%
\end{equation}

We derive a functional central limit theorem which establishes joint
convergence between the panel and time series portions of the process
$X_{n\tau}\left(  r\right)  $. The result is useful in analyzing both trend
stationary and unit root settings. In the latter, we specialize the model to a
linear time series setting. The functional CLT is then used to establish
proper joint convergence between stochastic integrals and the cross-sectional
component of our model.

For the stationary case we are mostly interested in $X_{n\tau}\left(
1\right)  $ where in particular
\[
\frac{1}{\sqrt{\tau}}\sum_{t=\tau_{0}+1}^{\tau_{0}+\tau}\psi_{\tau,t}^{\nu
}=\sum_{t=\min(1,\tau_{0}+1)}^{\max(T,\tau_{0}+\tau)}\sum_{i=1}^{n}\tilde
{\psi}_{it}^{\nu}\left(  1\right)
\]
and we used the fact that $\sum_{i=1}^{n}\tilde{\psi}_{it}^{\nu}\left(
1\right)  =\tilde{\psi}_{1t}^{\nu}\left(  1\right)  =\frac{\psi_{\tau,t}^{\nu
}}{\sqrt{\tau}}1\left\{  \tau_{0}+1\leq t\leq\tau_{0}+\tau\right\}  $ by
(\ref{psi_tilde_z}). The limiting distribution of $X_{n\tau}\left(  1\right)
$ is a simple corollary of the functional CLT for $X_{n\tau}\left(  r\right)
.$ We note that our treatment differs from Phillips and Moon (1999), who
develop functional CLT's for the time series dimension of the panel data set.
In our case, since $T$ is fixed and finite, a similar treatment is not applicable.

We introduce the following general regularity conditions for generic random
vectors $\psi_{n,it}^{y}$ and $\psi_{\tau,t}^{\nu}.$ Similarly, the central
limit theorem established in this section is for generic random vectors and
empirical processes satisfying the regularity conditions. In later sections,
these conditions will be specialized to the particular models considered
there. To apply the general theory in this section to specific models we will
evaluate the score or moment function at the true parameter value. In those
instances $\psi_{n,it}^{y}$ and $\psi_{\tau,t}^{\nu}$ will be interpreted as
the score or moment function evaluated at the true parameter value. We use
$\left\Vert .\right\Vert $ to denote the Euclidean norm.

\begin{condition}
\label{Diag_CLT_Cond}Assume that \newline i) $\psi_{n,it}^{y}$ is measurable
with respect to $\mathcal{G}_{\tau n,\left(  t-\min\left(  1,\tau_{0}\right)
\right)  n+i}.$\newline ii) $\psi_{\tau,t}^{\nu}$ is measurable with respect
to $\mathcal{G}_{\tau n,\left(  t-\min\left(  1,\tau_{0}\right)  \right)
n+i}$ for all $i=1,...,n.$\newline iii) for some $\delta>0$ and $C<\infty
,$\ $\sup_{it}E\left[  \left\Vert \psi_{n,it}^{y}\right\Vert ^{2+\delta
}\right]  \leq C$ for all $n\geq1.$\newline iv) for some $\delta>0$ and
$C<\infty,$\ $\sup_{t}E\left[  \left\Vert \psi_{\tau,t}^{\nu}\right\Vert
^{2+\delta}\right]  \leq C$ for all $\tau\geq1.$\newline v) $E\left[  \left.
\psi_{n,it}^{y}\right\vert \mathcal{G}_{\tau n,\left(  t-\min\left(
1,\tau_{0}\right)  \right)  n+i-1}\right]  =0.$ \newline vi) $E\left[  \left.
\psi_{\tau,t}^{\nu}\right\vert \mathcal{G}_{\tau n,\left(  t-\min\left(
1,\tau_{0}\right)  -1\right)  n+i}\right]  =0$ for $t>T$ and all
$i=1,...,n.$\newline vii) $\left\Vert E\left[  \left.  \psi_{\tau,t}^{\nu
}\right\vert \mathcal{G}_{\tau n,\left(  t-\min\left(  1,\tau_{0}\right)
-1\right)  n+i}\right]  \right\Vert _{2}\leq\vartheta_{t}$\textbf{ }for $t<0$
and all $i=1,...,n$ where
\[
\vartheta_{t}\leq C\left(  \left\vert t\right\vert ^{1+\delta}\right)  ^{-1/2}%
\]
for the same $\delta$ as in (iv) and some bounded constant $C$ and where for a
vector of random variables $x=\left(  x_{1},..,x_{d}\right)  ,$ $\left\Vert
x\right\Vert _{2}=\left(  \sum_{j=1}^{d}E\left[  \left\vert x_{j}\right\vert
^{2}\right]  \right)  ^{1/2}$ is the $L_{2}$ norm.
\end{condition}

\begin{remark}
Conditions \ref{Diag_CLT_Cond}(i), (iii) and (v) can be justified in a variety
of ways. One is the subordinated process theory employed in Andrews (2005)
which arises when $y_{it}$ are random draws from a population of outcomes $y$.
A sufficient condition for Conditions \ref{Diag_CLT_Cond}(v) to hold is that
$E\left[  \left.  \psi\left(  y|\theta,\rho,\nu_{t}\right)  \right\vert
\mathcal{C}\right]  =0$ holds in the population. This would be the case, for
example, if $\psi$ were the correctly specified score for the population
distribution of $y_{i,t}$ given $\nu_{t}$. See Andrews (2005, pp. 1573-1574).
\end{remark}

Conditions \ref{Diag_CLT_Cond}(i), (iii) and (v) impose a martingale
difference property for $\psi_{n,it}^{y}$ both in the time as well as in the
cross-section dimension. More specifically,\textbf{ }$E\left[  \psi_{n,it}%
^{y}|\psi_{n,i_{1}t}^{y},\psi_{n,i_{2}t_{2}}^{y},...,\psi_{n,i_{k}t_{k}}%
^{y},\mathcal{C}\right]  =0$\textbf{ }for any collection $\left(
i_{1},t\right)  ,...,\left(  i_{k},t_{k}\right)  $ with $i_{1}<i,$
$t_{2},...,t_{k}<t$ and $i_{r}\in\{1,...,n\}$ for $2\leq r\leq k.$ The central
limit theorem is established by letting the sums over the sample increase
sequentially in line with the information contained in\textbf{ }%
$\mathcal{G}_{\tau n,q}$\textbf{ }and thus mapping the sums over $t$ and $i$
into a single sum over an index set on the real line. This approach preserves
the information structure in\textbf{ }$\mathcal{G}_{\tau n,q}$ and in
particular takes into account that in general $E\left[  \psi_{n,it}^{y}%
|\psi_{n,js}^{y},\mathcal{C}\right]  \neq0$ for $s>t$ and all $j\in
\{1,...,n\}.$ The score of a correctly specified likelihood is a leading
example where the information structure takes the form implied by our conditions.

The construction in this paper is in contrast to some of the joint limit
theory for panels that is based on mds assumptions in the time direction only
and weak convergence of time series averages of cross-sectional aggregates
$\breve{\psi}_{n,t}^{y}=$ $n^{-1/2}\sum_{i=1}^{n}\psi_{n,it}^{y},$ driven by
$T$ going to infinity. Our theory on the other hand depends on the
cross-section sample size $n$ tending to infinity, while $T$ is kept fixed.
The pure cross-section case is covered by allowing for $T=1$. Other examples
include cases where the cross-section is sampled at random conditional
on\textbf{ }$\mathcal{C}$\textbf{ }and $\psi_{n,it}^{y}$ is a conditional
moment function in a rational agent model.

Conditions \ref{Diag_CLT_Cond}(ii), (iv) and (vi) impose a martingale property
for $\psi_{\tau,t}^{\nu}$ in the time dimension. In addition the condition
also implies that\textbf{ }$E\left[  \psi_{\tau,t}^{\nu}|\psi_{n,is}%
^{y},\mathcal{C}\right]  =0$\textbf{ }for any $i\in\{1,...,n\}$ and $s<t$ and
$t>T$. We note that this condition is weaker than assuming independence
between $\psi_{\tau,t}^{\nu}$ and $\psi_{n,is}^{y}$, even conditionally on
$\mathcal{C}$.\textbf{ }While the examples in HKM20 do satisfy such a
conditional independence restriction, it is not required for the CLT developed
in this paper.

We note that $E\left[  \left.  \psi_{\tau,t}^{\nu}\right\vert \mathcal{G}%
_{\tau n,q_{n}\left(  t-1,i\right)  }\right]  =0$ for all $i$ only holds for
$t>T$ because we condition not only on $z_{t-1},z_{t-2}...$ but also on
$\nu_{1},...,\nu_{T}$, where the latter may have non-trivial overlap with the
former. When $\tau_{0}$ is fixed the number of time periods $t$ where
$E\left[  \left.  \psi_{\tau,t}^{\nu}\right\vert \mathcal{G}_{\tau
n,q_{n}\left(  t-1,i\right)  }\right]  \neq0$ is finite and thus can be
neglected asymptotically. On the other hand, when $\tau_{0}$ varies with
$\tau$ there is an asymptotically non-negligible number of time periods where
the condition may not hold. To handle this latter case we impose an additional
mixingale type condition that is satisfied for typical time series models.

Condition \ref{Diag_CLT_Cond}(vii) is an additional restriction needed to
handle situations where $\tau_{0}$, the starting point of the time series
sample, is allowed to diverge to $-\infty.$ We call this situation backward
asymptotics. Since it generally is the case that $E\left[  \psi_{\tau,t}^{\nu
}|\mathcal{C}\right]  \neq0$\textbf{ }for $t\leq T$ because $\mathcal{C}$
contains information about future realizations of $\psi_{\tau,t}^{\nu}$ we
need a condition that limits this dependence as $t\rightarrow-\infty.$ The
following example illustrates that the condition naturally holds in linear
time series models\textbf{.}

\begin{example}
\label{Example_AR_Mixing}Assume $u_{s}$ is iid $N\left(  0,1\right)  $ and
$z_{s}=\sum_{j=0}^{\infty}\rho^{j}u_{s-j}$ with $\left\vert \rho\right\vert
<1$ is the stationary solution to $z_{s+1}=\rho z_{s}+u_{s+1}.$ Use the
convention that $\nu_{s}=z_{s}.$ Then the score of the Gaussian likelihood is
$\psi_{\tau,s}^{\nu}=z_{s}u_{s+1}.$ Assuming that $\tau_{0}=-\tau+1,$ $T=1$
and that $z_{s}$ is independent of $y_{i,t}$ conditional on $\mathcal{C}%
=\sigma\left(  \nu_{1}\right)  $, it is sufficient for this example to define
$\mathcal{G}_{\tau n,\left(  t-\min\left(  1,\tau_{0}\right)  \right)
n+i}=\sigma\left(  \left\{  z_{t},z_{t-1},\ldots,z_{\tau_{0}}\right\}
\right)  \vee\sigma\left(  \nu_{1}\right)  $. Then\footnote{Detailed
derivations are in Section \ref{Example_AR_Mxing_Calc}.}
\[
\left\Vert E\left[  \left.  \psi_{\tau,s}^{\nu}\right\vert \mathcal{G}_{\tau
n,\left(  s-\min\left(  1,\tau_{0}\right)  -1\right)  n+i}\right]  \right\Vert
_{2}=O\left(  \left\vert \rho\right\vert ^{\left\vert s\right\vert /2}\right)
=o\left(  \left\vert s\right\vert ^{-\left(  1+\delta\right)  /2}\right)  .
\]

\end{example}

The following conditions, Condition \ref{Omega_z} for the time series sample
and Condition \ref{Omega_y} below for the cross-section sample are put in
place to ensure that, in combination with Condition \ref{Diag_CLT_Cond}, the
variance of $X_{n\tau}\left(  r\right)  $ converges as $n$ and $\tau$ tend to
$\infty$ jointly. Conditions \ref{Omega_z} and \ref{Omega_y} correspond to
Condition 3.19 in Hall and Heyde (1980, Theorem 3.2). We note in particular
that the martingale structure in conjunction with uniform moment bounds
imposed in Condition \ref{Diag_CLT_Cond} are sufficient to guarantee that
cross-covariance terms over different time periods converge to zero.

\begin{condition}
\label{Omega_z}Assume that: \textbf{\newline}i) for any $r\in\left[
0,1\right]  ,$%
\[
\frac{1}{\tau}\sum_{t=\tau_{0}+1}^{\tau_{0}+\left[  \tau r\right]  }\psi
_{\tau,t}^{\nu}\psi_{\tau,t}^{\nu\prime}\overset{p}{\rightarrow}\Omega_{\nu
}\left(  r\right)  \text{ as }\tau\rightarrow\infty
\]
where $\Omega_{\nu}\left(  r\right)  $ is positive definite a.s. and
measurable with respect to $\sigma\left(  \nu_{1},...,\nu_{T}\right)  $ for
all $r\in(0,1]$. \newline ii) The elements of $\Omega_{\nu}\left(  r\right)  $
are bounded continuously differentiable functions of $r>s\in\lbrack0,1]$. The
derivatives $\dot{\Omega}_{\nu}\left(  r\right)  =\partial\Omega_{\nu}\left(
r\right)  /\partial r$ are positive definite almost surely.\newline iii) There
is a fixed constant $M<\infty$ such that $\sup_{\left\Vert \lambda_{\nu
}\right\Vert =1,\lambda_{\nu}\in\mathbb{R}^{k_{\rho}}}\sup_{t}\lambda_{\nu
}^{\prime}\dot{\Omega}_{\nu}\left(  t\right)  \lambda_{\nu}\leq M$ a.s.
\end{condition}

\begin{remark}
Note that by construction $\Omega_{\nu}\left(  0\right)  =0.$
\end{remark}

Condition \ref{Omega_z} is weaker than the conditions of Billingsley's (1968,
Theorem 23.1) functional CLT for strictly stationary mds\textbf{ }because we
neither assume strict stationarity nor homoskedasticity.\textbf{ }We do not
assume that $E\left[  \psi_{\tau,t}^{\nu}\psi_{\tau,t}^{\nu\prime}\right]  $
is constant. Brown (1971) allows for time varying variances, but uses stopping
times to achieve a standard Brownian limit$.$ Even more general treatments
with random stopping times are possible - see Gaenssler and Haeussler (1979).
On the other hand, here convergence to a Gaussian process (not a standard
Wiener process) with the same methodology (i.e. establishing convergence of
finite dimensional distributions and tightness) as in Billingsley, but without
assuming homoskedasticity is pursued. Related results with heteroskedastic
errors in the high frequency and stochastic process literature can be found
for example in Jacod and Shiryaev (2003, Theorem IX 7.28).

Heteroskedastic errors are explicitly used in Section
\ref{sec-unit-root-time-series-models} where $\psi_{\tau,t}^{\nu}=\exp\left(
\left(  t-s\right)  \gamma/\tau\right)  \eta_{s}$. Even if $\eta_{s}$ is
iid$\left(  0,\sigma^{2}\right)  $ it follows that $\psi_{\tau,t}^{\nu}$ is a
heteroskedastic triangular array that depends on $\tau$. It can be shown that
the variance kernel $\Omega_{\nu}\left(  r\right)  $ is $\Omega_{\nu}\left(
r\right)  =\left.  \sigma^{2}\left(  1-\exp\left(  -2r\gamma\right)  \right)
\right/  2\gamma$ in this case. See equation (\ref{CW2.3}).

\begin{condition}
\label{Omega_y}Assume that
\[
\frac{1}{n}\sum_{i=1}^{n}\psi_{n,it}^{y}\psi_{n,it}^{y\prime}%
\overset{p}{\rightarrow}\Omega_{ty}%
\]
where $\Omega_{ty}$ is positive definite a.s. and measurable with respect to
$\sigma\left(  \nu_{1},...,\nu_{T}\right)  .$
\end{condition}

Condition \ref{Omega_z} holds under a variety of conditions that imply some
form of weak dependence of the process $\psi_{\tau,t}^{\nu}$. These include,
in addition to Condition \ref{Diag_CLT_Cond}(ii) and (iv), mixing or near
epoch dependence assumptions on the temporal dependence properties of the
process $\psi_{\tau,t}^{\nu}.$ Condition \ref{Omega_y} holds under appropriate
moment bounds and random sampling in the cross-section even if the underlying
population distribution is not independent (see Andrews, 2005, for a detailed treatment).

\subsection{Stable Functional CLT}

This section details the probabilistic setting we use to accommodate the
results that Jacod and Shiryaev (2002) (shorthand notation JS) develop for
general Polish spaces. Let $\left(  \Omega^{\prime},\mathcal{G}^{\prime
},P^{\prime}\right)  $ be a probability space with increasing filtrations
$\mathcal{G}_{k_{n},q}\subset\mathcal{G}^{\prime}$ and $\mathcal{G}_{k_{n}%
,q}\subset\mathcal{G}_{k_{n},q+1}$ for any $q=1,...,k_{n}$ and an increasing
sequence $k_{n}\rightarrow\infty$ as $n\rightarrow\infty$.\footnote{In our
case $k_{n}=\max(T,\tau)n$ where both $n\rightarrow\infty$ and $\tau
\rightarrow\infty$ such that clearly $k_{n}\rightarrow\infty.$} Let
$D_{\mathbb{R}^{k_{\theta}}\times\mathbb{R}^{k_{\rho}}}\left[  0,1\right]  $
be the space of functions $\left[  0,1\right]  \rightarrow\mathbb{R}%
^{k_{\theta}}\times\mathbb{R}^{k_{\rho}}$ that are right continuous and have
left limits (see Billingsley (1968, p.109)) with discontinuities synchronized
across all elements in the vectors. Let $\mathcal{C}$ be a sub-sigma field of
$\mathcal{G}^{\prime}$. Let $\left(  \zeta,Z^{n}\left(  \omega,t\right)
\right)  :\Omega^{\prime}\times\left[  0,1\right]  \rightarrow\mathbb{R\times
R}^{k_{\theta}}\times\mathbb{R}^{k_{\rho}}$ be random variables or random
elements in $\mathbb{R}$ and $D_{\mathbb{R}^{k_{\theta}}\times\mathbb{R}%
^{k_{\rho}}}\left[  0,1\right]  $, respectively defined on the common
probability space $\left(  \Omega^{\prime},\mathcal{G}^{\prime},P^{\prime
}\right)  $ and assume that $\zeta$ is bounded and measurable with respect to
$\mathcal{C}$.

As in JS, p.512, let $Z\left(  \omega^{\prime},x\right)  =x$ be the canonical
element on $D_{\mathbb{R}^{k_{\theta}}\times\mathbb{R}^{k_{\rho}}}\left[
0,1\right]  $ and let $Q\left(  \omega^{\prime},dx\right)  $ be a version of
the distribution of $Z$ conditional on $\mathcal{C}$. Similarly, let
$Q_{n}\left(  \omega^{\prime},dx\right)  $ be a version of the conditional (on
$\mathcal{C}$) distribution of $Z^{n}.$ Following JS (Definition VI1.1 and
Theorem VI1.14) we define the $\sigma$-field generated by all coordinate
projections on $D_{\mathbb{R}^{k_{\theta}}\times\mathbb{R}^{k_{\rho}}}\left[
0,1\right]  $ as $\mathcal{D}_{\mathbb{R}^{k_{\theta}}\times\mathbb{R}%
^{k_{\rho}}}$. Then define the joint probability space $\left(  \Omega
,\mathcal{G},P\right)  $ with $\Omega=\Omega^{\prime}\times D_{\mathbb{R}%
^{k_{\theta}}\times\mathbb{R}^{k_{\rho}}}\left[  0,1\right]  ,$ $\mathcal{G}%
=\mathcal{G}^{\prime}\otimes\mathcal{D}_{\mathbb{R}^{k_{\theta}}%
\times\mathbb{R}^{k_{\rho}}}$ and%
\begin{equation}
P\left(  d\omega^{\prime},dx\right)  =P^{\prime}\left(  d\omega^{\prime
}\right)  Q\left(  \omega^{\prime},dx\right)  . \label{P-Measure}%
\end{equation}
Following JS (p.512, Definition 5.28) we say that $Z^{n}$ converges
$\mathcal{C}$-stably to $Z$ if for all bounded, $\mathcal{C}$-measurable
$\zeta$ and any continuous bounded functional $f:D_{\mathbb{R}^{k_{\theta}%
}\times\mathbb{R}^{k_{\rho}}}\left[  0,1\right]  \rightarrow\mathbb{R}$
\begin{equation}
E\left[  \zeta f\left(  Z^{n}\right)  \right]  \rightarrow E\left[  \zeta
Q\left[  f\left(  Z\right)  \right]  \right]
\label{Def Functional C Stable Convergence}%
\end{equation}
where $Q\left[  f\left(  Z\right)  \right]  $ is the expectation of $f\left(
Z\right)  $ conditional on $\mathcal{C}$. More specifically, if $W\left(
r\right)  $ is standard Brownian motion, we say that $Z^{n}\Rightarrow
W\left(  r\right)  $ $\mathcal{C}$-stably where the notation means that
(\ref{Def Functional C Stable Convergence}) holds when $Q$ is Wiener measure
(for a definition and existence proof see Billingsley (1968, Chapter 2,
Section 9)). Our proof strategy is based on JS Proposition VIII, 5.33 which
shows that $Z^{n}$ converges $\mathcal{C}$-stably iff $Z^{n}$ is tight and for
all $A\in\mathcal{C}$, $E\left[  1_{A}f\left(  Z^{n}\right)  \right]  $ converges.

The concept of stable convergence was introduced by R\'{e}nyi (1963) and has
found wide application in probability and statistics. Most relevant to the
discussion here are the stable central limit theorem of Hall and Heyde (1980,
Theorem 3.2) and Kuersteiner and Prucha (2013) who extend the result in Hall
and Heyde (1980, Theorem 3.2) to panel data with fixed $T$. Related to our
work, stable functional limit theorems were obtained previously for different
settings by Rootzen (1983), Feigin (1985) and Dedecker and Merlevede
(2002).\textbf{ }For example, Dedecker and Merlevede (2002) established a
related stable functional CLT for strictly stationary martingale differences
while we allow for heterogeneity and non-stationarity.

\begin{theorem}
\label{FCLT}Assume that Conditions \ref{Diag_CLT_Cond}, \ref{Omega_z} and
\ref{Omega_y} hold. Then it follows that for $\tilde{\psi}_{it}$ defined in
(\ref{CLT_Process}), and as $\tau,n\rightarrow\infty,$ and $T$ fixed,
\[
X_{n\tau}\left(  r\right)  \Rightarrow\left[
\begin{array}
[c]{c}%
B_{y}\left(  1\right) \\
B_{\nu}\left(  r\right)
\end{array}
\right]  \text{ (}\mathcal{C}\text{-stably)}%
\]
where $B_{y}\left(  r\right)  =\Omega_{y}{}^{1/2}W_{y}(r),$ $B_{\nu}\left(
r\right)  =\int_{0}^{r}\dot{\Omega}_{\nu}\left(  s\right)  ^{1/2}dW_{\nu}(s)$
and $\Omega\left(  r\right)  =\operatorname*{diag}\left(  \Omega_{y}%
,\Omega_{\nu}\left(  r\right)  \right)  $ is $\mathcal{C}$-measurable,
$\dot{\Omega}_{\nu}\left(  s\right)  =\partial\Omega_{\nu}\left(  s\right)
/\partial s$ and $\left(  W_{y}\left(  r\right)  ,W_{\nu}\left(  r\right)
\right)  $ is a vector of standard $k_{\phi}$-dimensional, mutually
independent, Brownian processes independent of $\Omega$.
\end{theorem}

\begin{proof}
In Appendix \ref{proof-section-jointCLT}.
\end{proof}

\begin{remark}
Note that $W_{y}\left(  r\right)  =W_{y}\left(  1\right)  $ for each
$r\in\left[  0,1\right]  $ by construction. Thus, $W_{y}\left(  1\right)  $ is
simply a vector of standard Gaussian random variables, independent both of
$W_{\nu}\left(  r\right)  $ and any random variable measurable with respect to
$\mathcal{C}$.
\end{remark}

The limiting random variables $B_{y}\left(  r\right)  $ and $B_{\nu}\left(
r\right)  $ both depend on $\mathcal{C}$ and thus are mutually dependent.
However, conditional on $\mathcal{C}$, the limiting random variables are
independent because of the mutual independence of $W_{y}\left(  r\right)  $
and $W_{\nu}\left(  r\right)  .$ The representation $B_{y}\left(  1\right)
=\Omega_{y}{}^{1/2}W_{y}(1),$ where a stable limit is represented as the
product of an independent Gaussian random variable and a scale factor that
depends on $\mathcal{C}$, is common in the literature on stable convergence.
Results similar to the one for $B_{\nu}\left(  r\right)  $ were obtained by
Phillips (1987, 1988) for cases where $\dot{\Omega}_{\nu}\left(  s\right)  $
is non-stochastic and has an explicitly functional form, notably for near unit
root processes and when convergence is marginal rather than stable. Rootzen
(1983) establishes stable convergence but gives a representation of the
limiting process in terms of standard Brownian motion obtained by a stopping
time transformation. The representation of $B_{\nu}\left(  r\right)  $ in
terms of a stochastic integral over the random scale process $\dot{\Omega
}_{\nu}\left(  s\right)  $ is obtained by utilizing a technique mentioned in
Rootzen (1983, p. 10) but not utilized there, namely establishing finite
dimensional convergence using a stable martingale CLT. This technique combined
with a tightness argument establishes the characteristic function of the
limiting process. The representation for $B_{\nu}\left(  r\right)  $ is then
obtained by utilizing isometry properties of the stochastic integral. Rootzen
(1983, p.13) similarly utilizes characteristic functions to identify the
limiting distribution in the case of standard Brownian motion. Similar
representations have been obtained in the high frequency time series
literature, see Jacod et al. (2010), Jacod and Protter (2012). Finally, the
results of Dedecker and Merlevede (2002) differ from ours in that they only
consider asymptotically homoskedastic and strictly stationary processes. In
our case, heteroskedasticity is explicitly allowed because of $\dot{\Omega
}_{\nu}\left(  s\right)  .$ An important special case of Theorem \ref{FCLT} is
the near unit root model discussed in more detail in Section
\ref{sec-unit-root-time-series-models}.

More importantly, our results innovate over the literature by establishing
joint convergence between cross-sectional and time series averages that are
generally not independent and whose limiting distributions are not
independent. This result is obtained by a novel construction that embeds both
data sets in a random field. A careful construction of information filtrations
$\mathcal{G}_{\tau n,n+i}$ allows to map the field into a martingale array.
Similar techniques were used in Kuersteiner and Prucha (2013) for panels with
fixed $T.$ In this paper we extend their approach to handle an additional and
distinct time series data-set and by allowing for both $n$ and $\tau$ to tend
to infinity jointly. In addition to the more complicated data-structure, we
extend Kuersteiner and Prucha (2013) by considering functional central limit theorems.

The following corollary is useful for possibly non-linear but trend stationary models.

\begin{corollary}
\label{Diag_CLT}Assume that Conditions \ref{Diag_CLT_Cond}, \ref{Omega_z} and
\ref{Omega_y} hold. Then it follows that for $\tilde{\psi}_{it}$ defined in
(\ref{CLT_Process}), and as $\tau,n\rightarrow\infty$ and $T$ fixed,
\[
X_{n\tau}\left(  1\right)  \overset{d}{\rightarrow}B\equiv\Omega^{1/2}W\text{
(}\mathcal{C}\text{-stably)}%
\]
where $\Omega=\operatorname*{diag}\left(  \Omega_{y},\Omega_{\nu}\left(
1\right)  \right)  $ is $\mathcal{C}$-measurable and $W=\left(  W_{y}\left(
1\right)  ,W_{\nu}\left(  1\right)  \right)  $ is a vector of standard
$d$-dimensional Gaussian random variables independent of $\Omega$. The
variables $\Omega_{y},\Omega_{\nu}\left(  .\right)  ,W_{y}\left(  .\right)  $
and $W_{\nu}\left(  .\right)  $ are as defined in Theorem \ref{FCLT}.
\end{corollary}

\begin{proof}
In Appendix \ref{proof-section-jointCLT}.
\end{proof}

The result of Corollary \ref{Diag_CLT} is equivalent to the statement that
$X_{n\tau}\left(  1\right)  \overset{d}{\rightarrow}N\left(  0,\Omega\right)
$ conditional on positive probability events in $\mathcal{C}$. As noted
earlier, no simplification of the technical arguments are possible by
conditioning on $\mathcal{C}$ except in the trivial case where $\Omega$ is a
fixed constant. Eagleson (1975, Corollary 3), see also Hall and Heyde (1980,
p. 59), establishes a simpler result where $X_{n\tau}\left(  1\right)
\overset{d}{\rightarrow}B$ weakly but not ($\mathcal{C}$-stably). Such results
could in principle be obtained here as well, but they would not be useful for
the analysis in Sections \ref{sec-trend-stationary-models} and
\ref{sec-unit-root-time-series-models} because the limiting distributions of
our estimators not only depend on $B$ but also on other $\mathcal{C}%
$-measurable scaling matrices. Since the continuous mapping theorem requires
joint convergence, a weak limit for $B$ alone is not sufficient to establish
the results we obtain below.\textbf{ }

Theorem \ref{Diag_CLT} establishes what Phillips and Moon (1999) call diagonal
convergence, a special form of joint convergence.\footnote{The discussion
assumes that $0<\kappa<\infty.$ The cases where $\kappa=0$ or $\kappa=\infty$
allow for a simpler treatment where either the time series or cross-section
sample can be ignored. In those situations considerations of joint convergence
play only a minor role.\textbf{ }} To see that sequential convergence where
first $n$ or $\tau$ go to infinity, followed by the other index, is generally
not useful in our set up, consider the following example. Assume that
$d=k_{\phi}$ is the dimension of the vector $\tilde{\psi}_{it}$. This would
hold for just identified moment estimators and likelihood based procedures.
Consider the double indexed process
\begin{equation}
X_{n\tau}\left(  1\right)  =\sum_{t=\min(1,\tau_{0}+1)}^{\max(T,\tau_{0}%
+\tau)}\sum_{i=1}^{n}\tilde{\psi}_{it}\left(  1\right)  .
\label{Double_Index_Proc}%
\end{equation}
For each $\tau$ fixed, convergence in distribution of $X_{n\tau}$ as
$n\rightarrow\infty$ follows from the central limit theorem in Kuersteiner and
Prucha (2013). Let $X_{\tau}$ denote the \textquotedblleft large $n$, fixed
$\tau$\textquotedblright\ limit. For each $n$ fixed, convergence in
distribution of $X_{n\tau}$ as $\tau\rightarrow\infty$ follows from a standard
martingale central limit theorem for Markov processes. Let $X_{n}$ be the
\textquotedblleft large $\tau$, fixed $n$\textquotedblright\ limit. It is
worth pointing out that the distributions of both $X_{n}$ and $X_{\tau}$ are
unknown because the limits are trivial in one direction. For example, when
$\tau$ is fixed and $n$ tends to infinity, the component $\tau^{-1/2}%
\sum_{t=\tau_{0}+1}^{\tau_{0}+\tau}\psi_{\tau,t}^{\nu}$ trivially converges in
distribution (it does not change with $n$) but the distribution of
$\tau^{-1/2}\sum_{t=\tau_{0}+1}^{\tau_{0}+\tau}\psi_{\tau,t}^{\nu}$ is
generally unknown. More importantly, application of a conventional CLT for the
cross-section alone will fail to account for the dependence between the time
series and cross-sectional components. Sequential convergence arguments thus
are not recommended even as heuristic justifications of limiting distributions
in our setting.

\subsection{Trend Stationary Models\label{sec-trend-stationary-models}}

This section provides the theoretical foundation for the inference methods
proposed in Section 6 of HKM20. Let $\theta=\left(  \beta,\nu_{1},...,\nu
_{T}\right)  $ and define the shorthand notation $f_{it}\left(  \theta
,\rho\right)  =f\left(  y_{it}|\theta,\rho\right)  $, $g_{t}\left(  \beta
,\rho\right)  =g\left(  \nu_{t}|\nu_{t-1},\beta,\rho\right)  ,$ $f_{\theta
,it}\left(  \theta,\rho\right)  =\partial f_{it}\left(  \theta,\rho\right)
/\partial\theta$ and $g_{\rho,t}\left(  \beta,\rho\right)  =\partial
g_{t}\left(  \beta,\rho\right)  /\partial\rho.$ Also let $f_{it}=f_{it}\left(
\theta_{0},\rho_{0}\right)  ,$ $f_{\theta,it}=f_{\theta,it}\left(  \theta
_{0},\rho_{0}\right)  ,$ $g_{t}=g_{t}\left(  \beta_{0},\rho_{0}\right)  $ and
$g_{\rho,t}=g_{\rho,t}\left(  \beta_{0},\rho_{0}\right)  .$ Depending on
whether the estimator under consideration is maximum likelihood or moment
based we assume that either $\left(  f_{\theta,it},g_{\rho,t}\right)  $ or
$\left(  f_{it},g_{t}\right)  $ satisfy the same Assumptions as $\left(
\psi_{it}^{y},\psi_{\tau,t}^{\nu}\right)  $ in Condition \ref{Diag_CLT_Cond}.
We recall that $\nu_{t}\left(  \beta,\rho\right)  $ is a function of $\left(
z_{t},\beta,\rho\right)  $, where $z_{t}$ are observable macro variables. For
the CLT, the process $\nu_{t}=\nu_{t}\left(  \beta_{0},\rho_{0}\right)  $ is
evaluated at the true parameter values and treated as observed. In
applications, $\nu_{t}$ will be replaced by an estimate which potentially
affects the limiting distribution of $\rho.$ This dependence is analyzed in a
step separate from the CLT.

The next step is to use Corollary \ref{Diag_CLT} to derive the joint limiting
distribution of estimators for $\phi=\left(  \theta^{\prime},\rho^{\prime
}\right)  ^{\prime}$. Define $s_{ML}^{\nu}\left(  \beta,\rho\right)
=\tau^{-1/2}\sum_{t=\tau_{0}+1}^{\tau_{0}+\tau}\partial g\left(  \nu
_{t}\left(  \beta,\rho\right)  |\nu_{t-1}\left(  \beta,\rho\right)
,\beta,\rho\right)  /\partial\rho$ and $s_{ML}^{y}\left(  \theta,\rho\right)
=n^{-1/2}\sum_{t=1}^{T}\sum_{i=1}^{n}\partial f\left(  y_{it}|\theta
,\rho\right)  /\partial\theta$ for maximum likelihood, and
\[
s_{M}^{\nu}\left(  \beta,\rho\right)  =-\left(  \partial k_{\tau}\left(
\beta,\rho\right)  /\partial\rho\right)  ^{\prime}W_{\tau}^{\tau}\tau
^{-1/2}\sum_{t=\tau_{0}+1}^{\tau_{0}+\tau}g\left(  \nu_{t}\left(  \beta
,\rho\right)  |\nu_{t-1}\left(  \beta,\rho\right)  ,\beta,\rho\right)
\]
and $s_{M}^{y}\left(  \theta,\rho\right)  =-\left(  \partial h_{n}\left(
\theta,\rho\right)  /\partial\theta\right)  ^{\prime}W_{n}^{C}n^{-1/2}%
\sum_{t=1}^{T}\sum_{i=1}^{n}f\left(  y_{it}|\theta,\rho\right)  $ for moment
based estimators. We use $s^{\nu}\left(  \beta,\rho\right)  $ and
$s^{y}\left(  \theta,\rho\right)  $ generically for arguments that apply to
both maximum likelihood and moment based estimators. The estimator $\hat{\phi
}$ jointly satisfies the moment restrictions using time series data
\begin{equation}
s^{\nu}\left(  \hat{\beta},\hat{\rho}\right)  =0. \label{Moment Cond z}%
\end{equation}
and cross-sectional data
\begin{equation}
s^{y}\left(  \hat{\theta},\hat{\rho}\right)  =0. \label{Moment Cond y}%
\end{equation}
\textbf{ }Defining $s\left(  \phi\right)  =\left(  s^{y}\left(  \phi\right)
^{\prime},s^{\nu}\left(  \phi\right)  ^{\prime}\right)  ^{\prime}$ the
estimator $\hat{\phi}$ satisfies $s\left(  \hat{\phi}\right)  =0$. A first
order Taylor series expansion around $\phi_{0}$ is used to obtain the limiting
distribution for $\hat{\phi}$. We impose the following additional assumption.

\begin{condition}
\label{Hessian}Let $\phi=\left(  \theta^{\prime},\rho^{\prime}\right)
^{\prime}\in\mathbb{R}^{k_{\phi}},$ $\theta\in\mathbb{R}^{k_{\theta}},$ and
$\rho\in\mathbb{R}^{k_{\rho}}.$ Define $D_{n\tau}=\operatorname*{diag}\left(
n^{-1/2}I_{y},\tau^{-1/2}I_{\nu}\right)  ,$where $I_{y}$ is an identity matrix
of dimension $k_{\theta}$ and $I_{\nu}$ is an identity matrix of dimension
$k_{\rho}.$ Assume that for some $\varepsilon>0,$ \textbf{\newline}%
\[
\sup_{\phi:\left\Vert \phi-\phi_{0}\right\Vert \leq\varepsilon}\left\Vert
\frac{\partial s\left(  \phi\right)  }{\partial\phi^{\prime}}D_{n\tau
}-A\left(  \phi\right)  \right\Vert =o_{p}\left(  1\right)
\]
where $A\left(  \phi\right)  $ is $\mathcal{C}$-measurable and $A=A\left(
\phi_{0}\right)  $ is full rank almost surely. Let $\kappa=\lim n/\tau,$
\[
A=\left[
\begin{array}
[c]{cc}%
A_{y,\theta} & \sqrt{\kappa}A_{y,\rho}\\
\frac{1}{\sqrt{\kappa}}A_{\nu,\theta} & A_{\nu,\rho}%
\end{array}
\right]
\]
with $A_{y,\theta}=\operatorname*{plim}n^{-1/2}\partial s^{y}\left(  \phi
_{0}\right)  /\partial\theta^{\prime}$, $A_{y,\rho}=\operatorname*{plim}%
n^{-1/2}\partial s^{y}\left(  \phi_{0}\right)  /\partial\rho^{\prime}$,
$A_{\nu,\theta}=\operatorname*{plim}\tau^{-1/2}\partial s^{\nu}\left(
\phi_{0}\right)  /\partial\theta^{\prime}$ and $A_{\nu,\rho}%
=\operatorname*{plim}\tau^{-1/2}\partial s^{\nu}\left(  \phi_{0}\right)
/\partial\rho^{\prime}$.
\end{condition}

\begin{condition}
\label{s-var-ML}For maximum likelihood criteria the following holds: \newline
i)\ for any $r\in\left[  0,1\right]  $, $\frac{1}{\tau}\sum_{t=\tau_{0}%
+1}^{\tau_{0}+\left[  \tau r\right]  }g_{\rho,t}g_{\rho,t}^{\prime
}\overset{p}{\rightarrow}\Omega_{\nu}\left(  r\right)  $ as $\tau
\rightarrow\infty$ and where $\Omega_{\nu}\left(  r\right)  $ satisfies the
same regularity conditions as in Condition \ref{Omega_z}(ii). \newline ii)
$\frac{1}{n}\sum_{i=1}^{n}f_{\theta,it}f_{\theta,it}^{\prime}%
\overset{p}{\rightarrow}\Omega_{ty}$ for all $t\in\left[  1,...,T\right]  $
and where $\Omega_{ty}$ is positive definite a.s. and measurable with respect
to $\sigma\left(  \nu_{1},...,\nu_{T}\right)  .$ Let $\Omega_{y}=\sum
_{t=1}^{T}\Omega_{ty}$.
\end{condition}

\begin{condition}
\label{s-var-GMM}Let $W^{C}=\operatorname*{plim}_{n}W_{n}^{C}$ and $W^{\tau
}=\operatorname*{plim}_{\tau}W_{\tau}^{\tau}$ and assume the limits to be
positive definite and $\mathcal{C}$-measurable. Define $h\left(  \theta
,\rho\right)  =\operatorname*{plim}_{n}h_{n}\left(  \beta,\nu_{t},\rho\right)
$ and $k\left(  \beta,\rho\right)  =\operatorname*{plim}_{\tau}k_{\tau}\left(
\beta,\rho\right)  .$ For moment based criteria the following holds: \newline
i)\ $\frac{1}{\tau}\sum_{t=\tau_{0}+1}^{\tau_{0}+\left[  \tau r\right]  }%
g_{t}g_{t}^{\prime}\overset{p}{\rightarrow}\Omega_{g}\left(  r\right)  $ as
$\tau\rightarrow\infty$ \textbf{ }and where $\Omega_{g}\left(  r\right)  $
satisfies the same regularity conditions as in Condition \ref{Omega_z}(ii).
\newline ii) $\frac{1}{n}\sum_{i=1}^{n}f_{it}f_{it}^{\prime}%
\overset{p}{\rightarrow}\Omega_{t,f}$ for all $t\in\left[  1,...,T\right]  .$
Let $\Omega_{f}=\sum_{t=1}^{T}\Omega_{t,f}.$ Assume that $\Omega_{f}$ is
positive definite a.s. and measurable with respect to $\sigma\left(  \nu
_{1},...,\nu_{T}\right)  $.\newline Assume that for some $\varepsilon>0,$
\textbf{\newline}iii) $\sup_{\phi:\left\Vert \phi-\phi_{0}\right\Vert
\leq\varepsilon}\left\Vert \left(  \partial k_{\tau}\left(  \beta,\rho\right)
/\partial\rho\right)  ^{\prime}W_{\tau}^{\tau}-\partial k\left(  \beta
,\rho\right)  ^{\prime}/\partial\rho W^{\tau}\right\Vert =o_{p}\left(
1\right)  ,$ \newline iv) $\sup_{\phi:\left\Vert \phi-\phi_{0}\right\Vert
\leq\varepsilon}\left\Vert \left(  \partial h_{n}\left(  \theta,\rho\right)
/\partial\theta\right)  ^{\prime}W_{n}^{C}-\left(  \partial h\left(
\theta,\rho\right)  /\partial\theta\right)  ^{\prime}W^{C}\right\Vert
=o_{p}\left(  1\right)  .$
\end{condition}

It is easy to see that the regularity conditions laid out in Conditions
\ref{Diag_CLT_Cond}, \ref{Hessian}, and \ref{s-var-ML} are satisfied if the
requirements in Footnote 32 of HKM20 are imposed on the estimating functions
$f$ and $g$ defined in that paper. The following result establishes the joint
limiting distribution of $\hat{\phi}.$

\begin{theorem}
\label{CLT_MLE}In the case of likelihood based estimators assume that
Conditions \ref{Diag_CLT_Cond}, \ref{Hessian}, and \ref{s-var-ML} hold with
$\left(  \psi_{it}^{y},\psi_{\tau,t}^{\nu}\right)  =\left(  f_{\theta
,it},g_{\rho,t}\right)  $. In the case of moment based estimators, assume that
Conditions \ref{Diag_CLT_Cond}, \ref{Hessian}, and \ref{s-var-GMM} hold with
$\left(  \psi_{it}^{y},\psi_{\tau,t}^{\nu}\right)  =\left(  \frac{\partial
h\left(  \theta_{0},\rho_{0}\right)  ^{\prime}}{\partial\theta}W^{C}%
f_{it},\frac{\partial k\left(  \beta_{0},\rho_{0}\right)  ^{\prime}}%
{\partial\rho}W^{\tau}g_{t}\right)  $. Assume that $\hat{\phi}-\phi_{0}%
=o_{p}\left(  1\right)  $ and that (\ref{Moment Cond z}) and
(\ref{Moment Cond y}) hold. Then,
\[
D_{n\tau}^{-1}\left(  \hat{\phi}-\phi_{0}\right)  \overset{d}{\rightarrow
}-A^{-1}\Omega^{1/2}W\text{ (}\mathcal{C}\text{-stably)}%
\]
where $A$ is full rank almost surely, $\mathcal{C}$-measurable and is defined
in Condition \ref{Hessian}.
\end{theorem}

\begin{proof}
In Appendix \ref{proof-section-jointCLT}.
\end{proof}

\begin{remark}
The distribution of $\Omega^{1/2}W$ is given in Corollary \ref{Diag_CLT}. In
particular, $\Omega=$ $\operatorname*{diag}\left(  \Omega_{y},\Omega_{\nu
}\left(  1\right)  \right)  $. When the criterion is maximum likelihood,
$\Omega_{y}$ and $\Omega_{\nu}\left(  1\right)  $ are given in Condition
\ref{s-var-ML}. When the criterion is moment based, $\Omega_{y}=\frac{\partial
h\left(  \theta_{0},\rho_{0}\right)  ^{\prime}}{\partial\theta}W^{C}\Omega
_{f}W^{C\prime}\frac{\partial h\left(  \theta_{0},\rho_{0}\right)  }%
{\partial\theta}$ and $\Omega_{\nu}\left(  1\right)  =\frac{\partial k\left(
\beta_{0},\rho_{0}\right)  ^{\prime}}{\partial\rho}W^{\tau}\Omega_{g}\left(
1\right)  W^{\tau\prime}\frac{\partial k\left(  \beta_{0},\rho_{0}\right)
}{\partial\rho}$ with $\Omega_{f}$ and $\Omega_{g}\left(  1\right)  $ defined
in Condition \ref{s-var-GMM}.
\end{remark}

The theorem provides formulas for the joint asymptotic variance covariance
matrix of $\hat{\phi}$ in two scenarios. The first scenario obtains when $f$
and $g$ are either the scores of a likelihood function, or if they are
estimating functions in a just identified set of moment conditions. The second
scenario covers GMM\ estimators in a scenario where $f$ and $g$ are moment
functions in an overidentified set of moment conditions. The methods reported
in Section 6 of HKM20 use an exactly identified moment based approach. There
may be cases where one wants to estimate the cross-section model using a
likelihood approach and the time series model using a moment approach, or vice
versa. These cases can be handled as a special case of the second scenario,
where $f$ or $g$ is an exactly identified moment condition while the other one
may be an overidentified moment condition.

\section{Asymptotic Inference \label{section-intuition}}

Our asymptotic framework is such that standard textbook level analysis
suffices for the discussion of consistency of the estimators. In standard
analysis with a single data source, one typically restricts the moment
equation to ensure identification, and imposes further restrictions such that
the sample analog of the moment function converges uniformly to the population
counterpart. Because these arguments are well known we simply impose as a
high-level assumption that our estimators are consistent. In this section we
illustrate how the rigorous technical results of Section
\ref{Section_JointCLT} can be applied to statistical inference problems for
specific examples.

\subsection{Intuition \label{Intuition}}

For expositional purposes, suppose that the time series $z_{s}$ is such that
the log of its conditional probability density function given $z_{s-1}$ is
$g\left(  \left.  z_{s}\right\vert \beta,\rho\right)  $. To simplify the
exposition in this section we assume that the cross-section model does not
depend on the macro parameter $\rho$.\textbf{ }We denote the consistent first
stage estimator of $\theta=\left(  \beta,\nu_{1},\ldots,\nu_{T}\right)  $ by
$\widetilde{\theta}$.\footnote{In order to emphasize the fact that $\theta$ is
estimated using only the cross-section data, we use the symbol
$\widetilde{\theta}$. In more complicated models, $\theta$ needs to be
estimated using both cross-section and time series data, and we reserve the
notation $\hat{\theta}$ for the general joint estimator.}

We assume that the dimension of the cross-section data is $n$. Implicit in
this representation is the idea that we are given a short panel for estimation
of $\theta=\left(  \beta,\nu_{1},\ldots,\nu_{T}\right)  $, where $T$ denotes
the time series dimension of the panel data. In order to emphasize that $T$ is
small, we use the term 'cross-section' for the short panel data set, and adopt
asymptotics where $T$ is fixed. Then, assume that $\widetilde{\theta}$ is a
regular estimator with\textbf{ }influence function $\varphi_{it}$ such that
\begin{equation}
\sqrt{n}\left(  \widetilde{\theta}-\theta\right)  =\frac{1}{\sqrt{n}}%
\sum_{i=1}^{n}\sum_{t=1}^{T}\varphi_{it}+o_{p}\left(  1\right)
\label{rho-influence}%
\end{equation}
with $E\left[  \varphi_{it}\right]  =0$. Using $\widetilde{\theta}$ from the
cross-section data, we can then consider maximizing the criterion $G_{\tau
}\left(  \theta,\rho\right)  =\frac{1}{\tau}\sum_{s=\tau_{0}+1}^{\tau_{0}%
+\tau}g\left(  \left.  z_{s}\right\vert \theta,\rho\right)  $ with respect to
$\rho$. Here, $\tau_{0}+1$ denotes the beginning of the time series data,
which is allowed to differ from the beginning of the panel data. The moment
equation then is%
\[
\frac{\partial G_{\tau}\left(  \tilde{\theta},\hat{\rho}\right)  }%
{\partial\rho}=0
\]
and the asymptotic distribution of $\widehat{\rho}$ is characterized by
\[
\sqrt{\tau}\left(  \widehat{\rho}-\rho\right)  =-\left(  \frac{\partial
^{2}G\left(  \tilde{\theta},\rho\right)  }{\partial\rho\partial\rho^{\prime}%
}\right)  ^{-1}\left(  \sqrt{\tau}\frac{\partial G_{\tau}\left(  \tilde
{\theta},\rho\right)  }{\partial\rho}\right)  +o_{p}\left(  1\right)  .
\]
Because $\sqrt{\tau}\left(  \partial G_{\tau}\left(  \tilde{\theta}%
,\rho\right)  /\partial\rho-\partial G_{\tau}\left(  \theta,\rho\right)
/\partial\rho\right)  \approx\left(  \partial^{2}G\left(  \theta,\rho\right)
/\partial\theta\partial\rho^{\prime}\right)  \frac{\sqrt{\tau}}{\sqrt{n}}%
\sqrt{n}\left(  \widetilde{\theta}-\theta\right)  $ \textbf{ }we obtain
\begin{equation}
\sqrt{\tau}\left(  \widehat{\rho}-\rho\right)  =-A_{\nu,\rho}^{-1}\sqrt{\tau
}\frac{\partial G_{\tau}\left(  \theta,\rho\right)  }{\partial\rho}%
-A_{\nu,\rho}^{-1}A_{\nu,\theta}\frac{\sqrt{\tau}}{\sqrt{n}}\left(  \frac
{1}{\sqrt{n}}\sum_{i=1}^{n}\sum_{t=1}^{T}\varphi_{it}\right)  +o_{p}\left(
1\right)  \label{expansion-theta-hat-reflecting-tilde-rho}%
\end{equation}
with
\[
A_{\nu,\rho}^{-1}\equiv\frac{\partial^{2}G\left(  \theta,\rho\right)
}{\partial\rho\partial\rho^{\prime}},\quad A_{\nu,\theta}\equiv\frac
{\partial^{2}G\left(  \theta,\rho\right)  }{\partial\rho\partial\theta
^{\prime}}.
\]

Because both $A_{\nu,\rho}$ and $A_{\nu,\theta}$ are $\mathcal{C}$-measurable
random variables in the limit the continuous mapping theorem can only be
applied if joint convergence of $\sqrt{\tau}\partial G_{\tau}\left(
\theta,\rho\right)  /\partial\theta,n^{-1/2}\sum_{i=1}^{n}\sum_{t=1}%
^{T}\varphi_{it}$ and any $\mathcal{C}$-measurable random variable is
established. Joint stable convergence of both components delivers exactly
that. We also point out that it is perfectly possible to consistently estimate
parameters, in our case $\left(  \nu_{1},\ldots,\nu_{T}\right)  $, that remain
random in the limit. For related results, see the recent work of Kuersteiner
and Prucha (2020).

Assume that the unconditional distribution is such that
\[
\frac{1}{\sqrt{n}}\sum_{i=1}^{n}\sum_{t=1}^{T}\varphi_{it}%
\overset{d}{\rightarrow}MN\left(  0,\Omega_{y}\right)
\]
where $\Omega_{y}$ generally does depend on $\left(  \nu_{1},\ldots,\nu
_{T}\right)  $ through the parameter $\theta$ and as a result the distribution
is mixed normal in general. Let's also assume that
\[
\sqrt{\tau}\frac{\partial G_{\tau}\left(  \theta,\rho\right)  }{\partial\rho
}\overset{d}{\rightarrow}N\left(  0,\Omega_{\nu}\right)
\]
where we assume that $\Omega_{\nu}$ is a fixed constant that does not depend
on $\left(  \nu_{1},\ldots,\nu_{T}\right)  $.

We note that in general $\varphi_{it}$ is a function of $\left(  \nu
_{1},\ldots,\nu_{T}\right)  $.\textbf{ }If there is overlap between $\left(
1,\ldots,T\right)  $ and $\left(  \tau_{0}+1,\ldots,\tau_{0}+\tau\right)  $,
we need to worry about the asymptotic distribution of $\sqrt{\tau}\partial
G_{\tau}\left(  \theta,\rho\right)  /\partial\rho$ conditional on $\left(
\nu_{1},\ldots,\nu_{T}\right)  $. However, because in this example the only
connection between $y$ and $\varphi$ is assumed to be through $\theta$ and
because $T$ is assumed fixed, the two terms $\sqrt{\tau}\partial G_{\tau
}\left(  \theta,\rho\right)  /\partial\rho$ and $\frac{1}{\sqrt{n}}\sum
_{i=1}^{n}\sum_{t=1}^{T}\varphi_{it}$ are expected to be asymptotically
independent in the trend stationary case and when $\Omega_{\nu}$ does not
depend on $\left(  \nu_{1},\ldots,\nu_{T}\right)  $. Even in this simple
setting, independence between the two samples does not hold, and asymptotic
conditional or unconditional independence as well as joint convergence with
$\mathcal{C}$-measurable random variables needs to be established formally.
This follows from $\mathcal{C}$-stable convergence established in Section
\ref{sec-trend-stationary-models} and is summarized in the following Corollary.

\begin{corollary}
\label{CLT_Time_Series}Under the same conditions as in Theorem \ref{CLT_MLE}
it follows that%
\[
\sqrt{\tau}\left(  \widehat{\rho}-\rho\right)  \overset{d}{\rightarrow}%
-A_{\nu,\rho}^{-1}\Omega_{\nu}^{1/2}\left(  1\right)  W_{\nu}\left(  1\right)
-\frac{1}{\sqrt{\kappa}}A_{\nu,\rho}^{-1}A_{y,\theta}\Omega_{y}^{1/2}%
W_{y}\left(  1\right)  \text{ (}\mathcal{C}\text{-stably)}.
\]
For
\[
\Omega_{\rho}=A_{\nu,\rho}^{-1}\Omega_{\nu}A_{\nu,\rho}^{-1}+\frac{1}{\kappa
}A_{\nu,\rho}^{-1}A_{\nu,\theta}\Omega_{y}A_{\nu,\theta}^{\prime}A_{\nu,\rho
}^{-1}%
\]
it follows that
\[
\sqrt{\tau}\Omega_{\rho}^{-1/2}\left(  \widehat{\rho}-\rho\right)
\overset{d}{\rightarrow}N\left(  0,I\right)  \text{ (}\mathcal{C}%
\text{-stably).}%
\]

\end{corollary}

Corollary \ref{CLT_Time_Series} follows directly from Theorem \ref{CLT_MLE}.
It implies that
\begin{equation}
\sqrt{\tau}\left(  \widehat{\rho}-\rho\right)  \overset{d}{\rightarrow
}MN\left(  0,A_{\nu,\rho}^{-1}\Omega_{\nu}A_{\nu,\rho}^{-1}+\frac{1}{\kappa
}A_{\nu,\rho}^{-1}A_{\nu,\theta}\Omega_{y}A_{\nu,\theta}^{\prime}A_{\nu,\rho
}^{-1}\right)  , \label{stationary-result}%
\end{equation}
where $0<\kappa\equiv\lim n/\tau<\infty$ and where the limiting distribution
on the RHS of (\ref{stationary-result}) is a mixed Gaussian distribution. This
means that a practitioner would use the square root of%
\[
\frac{1}{\tau}\left(  A_{\nu,\rho}^{-1}\Omega_{\nu}A_{\nu,\rho}^{-1}+\frac
{1}{\kappa}A_{\nu,\rho}^{-1}A_{\nu,\theta}\Omega_{y}A_{\nu,\theta}^{\prime
}A_{\nu,\rho}^{-1}\right)  \approx\frac{1}{\tau}A_{\nu,\rho}^{-1}\Omega_{\nu
}A_{\nu,\rho}^{-1}+\frac{1}{n}A_{\nu,\rho}^{-1}A_{\nu,\theta}\Omega_{y}%
A_{\nu,\theta}^{\prime}A_{\nu,\rho}^{-1}%
\]
as the standard error when formulating a $t$-ratio. This result looks similar
to Murphy and Topel's (1985) formula, except that we need to make an
adjustment to the second component to address the differences in sample sizes.

The assumption that $0<\kappa<\infty$ is used as a technical device to obtain
an asymptotic approximation that accounts for estimation errors stemming both
from the cross-section and time series samples. Simulation results in Hahn
et.al (2016, 2020) for data and sample sizes calibrated to actual macro data
show that our approximation provides good control for estimator bias and test
size. The knife edge case $\kappa=\infty$ corresponds to situations where the
estimation of cross-section parameters can be neglected for inference about
$\widehat{\rho},$ and where now $\sqrt{\tau}\left(  \widehat{\rho}%
-\rho\right)  \overset{d}{\rightarrow}MN\left(  0,A_{\nu,\rho}^{-1}\Omega
_{\nu}A_{\nu,\rho}^{-1}\right)  .$ The expansion in
(\ref{expansion-theta-hat-reflecting-tilde-rho}) also shows that the case
$\kappa=0$ leads to a scenario where uncertainty from the cross-section
dominates such that the rate of convergence of $\hat{\rho}$ now is $\sqrt{n}$
rather than $\sqrt{\tau}$ and where $\sqrt{n}\left(  \widehat{\rho}%
-\rho\right)  \overset{d}{\rightarrow}MN\left(  0,A_{\nu,\rho}^{-1}%
A_{\nu,\theta}\Omega_{y}A_{\nu,\theta}^{\prime}A_{\nu,\rho}^{-1}\right)  .$
However, in what follows we focus on the case most relevant in practice where
$0<\kappa<\infty.$

The asymptotic variance formula is such that the noise of the cross-section
estimator $\widetilde{\theta}$ can make quite a difference if $\kappa$ is
small, i.e., if the cross-section size $n$ is small relative to the time
series size $\tau$. Obviously, this calls for larger cross-sections for
accurate estimation of the time series parameter $\rho$. We also note that
cross-section estimation asymptotically has no impact on macro estimation if
$A_{\nu,\theta}=0$. One scenario where $A_{\nu,\theta}=0$ is the case where
the model is additively separable in $\theta$ and $\rho$ such that $G\left(
\theta,\rho\right)  =G_{1}\left(  \theta\right)  +G_{2}\left(  \rho\right)
.$\textbf{ }

For completeness, we also present a result that focuses on the limiting
distribution of the subset of parameters that are associated with the
cross-sectional model. Here we no longer impose the restriction that the
cross-sectional model does not depend on time series parameters.
Cross-sectional parameters are the main object of interest in HKM20 while in
Section \ref{OP} of this paper we consider a model where the main parameter of
interest is a time series parameter.

\begin{corollary}
\label{Corollary_CLT_MLE}Under the same conditions as in Theorem \ref{CLT_MLE}
it follows that
\begin{equation}
\sqrt{n}\left(  \hat{\theta}-\theta_{0}\right)  \overset{d}{\rightarrow
}-A^{y,\theta}\Omega_{y}^{1/2}W_{y}\left(  1\right)  -\sqrt{\kappa}A^{y,\rho
}\Omega_{\nu}^{1/2}\left(  1\right)  W_{\nu}\left(  1\right)  \text{
(}\mathcal{C}\text{-stably)}. \label{Corollary_CLT_MLE_1}%
\end{equation}
where
\begin{align*}
A^{y,\theta}  &  =A_{y,\theta}^{-1}+A_{y,\theta}^{-1}A_{y,\rho}\left(
A_{\nu,\rho}-A_{\nu,\theta}A_{y,\theta}^{-1}A_{y,\rho}\right)  ^{-1}%
A_{\nu,\theta}A_{y,\theta}^{-1}\\
A^{y,\rho}  &  =-A_{y,\theta}^{-1}A_{y,\rho}\left(  A_{\nu,\rho}-A_{\nu
,\theta}A_{y,\theta}^{-1}A_{y,\rho}\right)  ^{-1}.
\end{align*}
For
\begin{equation}
\Omega_{\theta}=A^{y,\theta}\Omega_{y}A^{y,\theta\prime}+\kappa A^{y,\rho
}\Omega_{\nu}\left(  1\right)  A^{y,\rho\prime} \label{Omega_theta}%
\end{equation}
it follows that
\begin{equation}
\sqrt{n}\Omega_{\theta}^{-1/2}\left(  \hat{\theta}-\theta_{0}\right)
\overset{d}{\rightarrow}N\left(  0,I\right)  \text{ (}\mathcal{C}%
\text{-stably).} \label{Corollary_CLT_MLE_2}%
\end{equation}

\end{corollary}

The corollary develops the asymptotic distribution of the estimators for the
general case where neither the time series nor the cross-section parameters
are identified separately. We note that the exposition in Section 6 of
HKM20\ does impose the additional restriction that $A_{v,\theta}=0$ which
significantly simplifies (\ref{Omega_theta}). When $A_{v,\theta}=0$ the
distributional approximation reported in HKM20, Section 6, Eq (35) corresponds
to the result obtained in (\ref{Corollary_CLT_MLE_2}).\footnote{We also note
that the online appendix of HKM20\ contains explicit formulas for
$\Omega_{\theta}$ for the general equilibrium model considered in that paper.
Section \ref{WE} of this paper contains similar explicit formulas for a
version of the Olley and Pakes' model considered here.}\textbf{. }

Note that $\Omega_{\theta},$ the asymptotic variance of $\sqrt{n}\left(
\hat{\theta}-\theta_{0}\right)  $ conditional on $\mathcal{C}$, in general is
a random variable, and the asymptotic distribution of $\hat{\theta}$ is mixed
normal. However, as in Andrews (2005), the result in
(\ref{Corollary_CLT_MLE_2}) can be used to construct an asymptotically pivotal
test statistic. For a consistent estimator $\hat{\Omega}_{\theta}$ the
statistic $\sqrt{n}\hat{\Omega}_{\theta}^{-1/2}\left(  R\hat{\theta}-r\right)
$ is asymptotically distribution free under the null hypothesis $R\theta-r=0$
where $R$ is a conforming matrix of dimension $q\times k_{\theta}$ and and $r$
a $q\times1$ vector. These insights form the basis for the standard errors
proposed in Section 6 of HKM20.

We note that when two datasets are combined, the variance estimate of the
estimators have to reflect the estimation error of the nuisance parameters.
The formula above boils down to the usual variance formula of the two-step
estimator. In fact, the formula is somewhat simpler than the usual two step
formula, at least in the stationary scenario. The reason is that the
covariance of the moments of the two steps is zero when the moments are based
on cross-section and time series data respectively. This is generally not the
case for two step procedures based only on one sample.

An additional theoretical difficulty that arises in this paper are common
factors that remain random in the limit and affect the limiting variance
$\Omega_{\theta}.$ Theoretically, we handle this difficulty by relying on the
concept of stable convergence to establish the limit distribution of our
estimators. While inference based on pivotal statistics such as the $t$-ratio
is not affected by stable limits, caution needs to be exercised when
interpreting standard errors and confidence intervals for $\hat{\theta}.$ The
reason is that these quantities remain data-dependent through their dependence
on common shocks even in the limit and may not be comparable across different
empirical studies. This point is emphasized in HKM20, p. 1390.

\subsection{A Worked Example}

We now discuss how the model in Section \ref{OP} fits into our theory and
discuss how to obtain valid standard errors. The estimator introduced in
Section \ref{OP} is defined in terms of moment (not likelihood) based
criterion functions. Using the Taylor series expansion based intuition, we
discuss how the asymptotic distribution can be understood. Our discussion in
this section parallels and complements the material in Section 6 of HKM20.

Unlike the model in Section \ref{section-models}, the moment (\ref{C-GMM}) in
Section \ref{OP} does not identify all the $\nu_{1},...,\nu_{T}$, and it only
identifies $\left(  \beta_{0,t+1}^{\ast},\beta_{k},\alpha^{\left(  C\right)
}\right)  $, where we define $\beta_{0,t+1}^{\ast}\equiv\nu_{t+1}%
-\alpha^{\left(  C\right)  }\nu_{t}$. Therefore, it is convenient to define a
finite dimensional parameter that is identified from the cross-section as
$\theta\left(  \nu\right)  $, which may depend on the aggregate shocks
$\nu=\left(  \nu_{1},...,\nu_{T}\right)  $ instead of working with $\left(
\beta,\nu\right)  $. For the model in Section \ref{OP}, $\theta\left(
\nu\right)  $ is equal to $\left(  \nu_{t+1}-\alpha^{\left(  C\right)  }%
\nu_{t},\beta_{k},\alpha^{\left(  C\right)  }\right)  $. The parameter $\beta$
in Section \ref{Intuition} denotes the collection of cross section parameters
that do not depend on $\nu$. Since $\theta\left(  \nu\right)  =\left(
\nu_{t+1}-\alpha^{\left(  C\right)  }\nu_{t},\beta_{k},\alpha^{\left(
C\right)  }\right)  $ in Section \ref{OP}, only the parameters $\left(
\beta_{k},\alpha^{\left(  C\right)  }\right)  $ do not depend on $\nu$. Thus,
the $\left(  \beta_{k},\alpha^{\left(  C\right)  }\right)  $ in Section
\ref{OP} plays the role of $\beta$ in Section \ref{Intuition}.

We consider the following GMM\ estimation functions in the cross-section and
time series samples. Following the notational convention in Section
\ref{section-models}, we define $h_{n}\left(  \theta\right)  $ $=$ $\frac
{1}{n}\sum_{t=1}^{T}\sum_{j=1}^{n}$ $f\left(  \left.  y_{j,t}\right\vert
\theta\right)  $ with $\theta=\theta\left(  \nu\right)  $ and $k_{\tau}\left(
\beta,\rho\right)  $ $=$ $\frac{1}{\tau}\sum_{s=\tau_{0}+1}^{\tau_{0}+\tau
}g\left(  \left.  z_{s}\right\vert \beta,\rho\right)  $, where in Section
\ref{OP} the parameters are $\theta=\left(  \nu_{t+1}-\alpha^{\left(
C\right)  }\nu_{t},\beta_{k},\alpha^{\left(  C\right)  }\right)  $ and
$\rho=\alpha^{\left(  A\right)  }$\textbf{$.$} The reason why it is sufficient
to focus on $\left(  \beta_{k},\alpha^{\left(  C\right)  }\right)  $ is that
the main interest lies in $\alpha^{\left(  A\right)  }$ which can be
identified in the time series sample with knowledge of $\beta_{k}$ alone. Also
note that for this model, $y_{j,t}$ $=$ $\left(  i_{j,t},k_{j,t}%
,l_{j,t},\mathfrak{y}_{j,t}^{\ast},i_{j,t-1},k_{j,t-1},l_{j,t-1}%
,\mathfrak{y}_{j,t-1}^{\ast}\right)  $ and $z_{s}=\left(  Y_{s}^{\ast}%
,K_{s}^{\ast}\right)  $ is the vector of aggregate observed data.

The cross-sectional moment function $f\left(  y_{i,t}|\theta\right)  $ can be
specified as
\[
f\left(  y_{j,t}|\theta\right)  =\left(  \mathfrak{y}_{j,t}^{\ast}-\left(
\beta_{0,t}^{\ast}+\beta_{k}k_{j,t}+\alpha^{\left(  C\right)  }\left(
\phi_{t}\left(  i_{j,t-1},k_{j,t-1}\right)  -\beta_{k}k_{j,t-1}\right)
\right)  \right)  z_{j,t},
\]
where $z_{j,t}$ can be chosen as the vector $z_{j,t}=\left(  1,k_{j,t-1}%
,i_{j,t-1}\right)  ^{\prime}$, for example.

Similarly, specialize the generically defined function $g\left(  \left.
z_{s}\right\vert \beta,\rho\right)  $ for the aggregate time series model to
the score of the conditional pseudo-likelihood for the aggregate shock
process, denoted by $g\left(  \nu_{s}\left(  \beta\right)  |\nu_{s-1}\left(
\beta\right)  ,\beta,\rho\right)  \equiv g\left(  \left.  z_{s}\right\vert
\beta,\rho\right)  ,$ and where the aggregate shock $\nu_{s}\left(  \beta
_{k}\right)  =Y_{s}^{\ast}-\beta_{k}K_{s}^{\ast}$ depends on $z_{s}=\left(
Y_{s}^{\ast},K_{s}^{\ast}\right)  $ through the parameter $\beta_{k}.$ When
$\beta_{k}$ is evaluated at the true parameter value $\beta_{k,0},$ we use the
shorthand notation $\nu_{s}\equiv\nu_{s}\left(  \beta_{k,0}\right)  $. For the
AR(1) model we postulate for $\nu_{s},$ the function $g\left(  \nu_{s}\left(
\beta\right)  |\nu_{s-1}\left(  \beta\right)  ,\beta,\rho\right)  $ can be
written explicitly as $g\left(  \nu_{s}\left(  \beta\right)  |\nu_{s-1}\left(
\beta\right)  ,\beta,\rho\right)  =\left(  \nu_{s}\left(  \beta\right)
-\alpha^{\left(  A\right)  }\nu_{s-1}\left(  \beta\right)  \right)  \nu
_{s-1}\left(  \beta\right)  $.

Differentiating the counterparts of $F_{n}$ and $G\tau$ discussed in Section
\ref{section-models}, we can see that the GMM estimator for $\phi$ solves the
two moment conditions
\begin{align}
s_{M}^{y}\left(  \theta\right)   &  =-\left(  \partial h_{n}\left(
\theta,\rho\right)  /\partial\theta\right)  ^{\prime}W_{n}^{C}n^{-1/2}%
\sum_{t=1}^{T}\sum_{i=1}^{n}f\left(  y_{it}|\theta,\rho\right)
=0\label{Empirical_Joint_GMM}\\
s_{M}^{\nu}\left(  \beta,\rho\right)   &  =-\left(  \partial k_{\tau}\left(
\beta,\rho\right)  /\partial\rho\right)  ^{\prime}W_{\tau}^{\tau}\tau
^{-1/2}\sum_{t=\tau_{0}+1}^{\tau_{0}+\tau}g\left(  \nu_{t}\left(  \beta
,\rho\right)  |\nu_{t-1}\left(  \beta,\rho\right)  ,\beta,\rho\right)
=0.\nonumber
\end{align}
Proceeding as in HKM20, let $J_{n\tau}\left(  \phi\right)  =\left[
n^{-1/2}s_{M}^{y}\left(  \theta\right)  ,\tau^{-1/2}s_{M}^{\nu}\left(
\beta,\rho\right)  \right]  $ and $D_{n\tau}$ $=$ $\operatorname*{diag}\left(
n^{-1/2}I_{f},\tau^{-1/2}\right)  $, where $I_{f}$ is an identity matrix equal
to the dimension of $\theta.$ A Taylor series expansion of $J_{n\tau}\left(
\phi\right)  $ around $\phi_{0}$ leads to
\begin{equation}
0=D_{n\tau}^{-1}J_{n\tau}\left(  \phi_{0}\right)  +AD_{n\tau}^{-1}\left(
\hat{\phi}-\phi_{0}\right)  +o_{p}\left(  1\right)  , \label{phi-Expansion}%
\end{equation}
where $A=\operatorname*{plim}\left(  D_{n\tau}^{-1}\partial J_{n\tau}\left(
\phi\right)  /\partial\phi^{\prime}D_{n\tau}\right)  .$ The elements of the
matrix $A$ for the example are obtained as
\[
A=\left[
\begin{array}
[c]{cc}%
A_{y,\theta} & 0\\
\frac{1}{\sqrt{\kappa}}A_{\nu,\theta} & A_{\nu,\rho}%
\end{array}
\right]  ,
\]
where the upper right corner of $A$ is zero because the cross-sectional model
does not depend on the time series parameters $\rho$. This feature of the
model implies that a plug in estimator using the first step cross-sectional
estimate $\tilde{\beta}_{k}$ for the time series problem estimating
$\alpha^{\left(  A\right)  }$ is equivalent to an estimator $\hat{\phi}$
obtained jointly on the two samples.\footnote{In HKM20 and in Section
\ref{sec-trend-stationary-models} we show that these simplifications are not
generic features of the problem we study and that joint estimation is needed
except in special cases.\textbf{ }} The non-zero elements of the matrix $A$
are defined as
\[
A_{y,\theta}=-\operatorname*{plim}n^{-1/2}\frac{\partial s_{M}^{y}\left(
\theta\right)  }{\partial\theta^{\prime}},\text{ }A_{\nu,\rho}%
=-\operatorname*{plim}\tau^{-1/2}\frac{\partial s_{M}^{\nu}\left(
\theta\right)  }{\partial\rho^{\prime}},\text{ }A_{\nu,\theta}%
=-\operatorname*{plim}\tau^{-1/2}\frac{\partial s_{M}^{\nu}\left(
\theta\right)  }{\partial\theta^{\prime}}.
\]

Let $\mathcal{C=\sigma}\left(  \nu_{1},...,\nu_{T}\right)  $ be the sigma
field generated by the aggregate shocks of the cross-section sample. It can be
shown\footnote{See Appendix \ref{WE}.}
\[
D_{n\tau}^{-1}J_{n\tau}\left(  \phi_{0}\right)  \rightarrow_{d}N\left(
0,\Omega\right)  \text{ }\mathcal{C}\text{-stably}%
\]
where $\Omega=\operatorname*{diag}\left(  \Omega_{y},\Omega_{\nu}\right)  $ is
the asymptotic variance covariance matrix of the moment functions defined in
(\ref{Empirical_Joint_GMM}).

By the continuous mapping theorem and (\ref{phi-Expansion}), the limiting
distribution of $\hat{\phi}$ then is characterized by $D_{n\tau}^{-1}\left(
\hat{\phi}-\phi_{0}\right)  \rightarrow_{d}N\left(  0,V\right)  $
$\mathcal{C}$-stably where $V=A^{-1}\Omega\left(  A^{\prime}\right)  ^{-1}$.
Suppose that the \textquotedblleft conventional\textquotedblright\ weight
matrices are chosen so that $\operatorname*{plim}W_{n}^{C}=\Omega_{f}^{-1}$
and $\operatorname*{plim}W_{\tau}^{\tau}=\Omega_{g}^{-1}$, where $\Omega_{f}$
and $\Omega_{g}$ denote the asymptotic variance of $\frac{1}{\sqrt{n}}%
\sum_{t=1}^{T}\sum_{j=1}^{n}f\left(  \left.  y_{j,t}\right\vert \theta
_{0}\right)  $ and $\frac{1}{\sqrt{\tau}}\sum_{s=\tau_{0}+1}^{\tau_{0}+\tau
}g\left(  \left.  z_{s}\right\vert \beta_{0},\rho_{0}\right)  $.\footnote{See
(\ref{Def_Omega_f}) and (\ref{Def_Omega_g}) in Appendix \ref{WE}.} We would
then have $\Omega=\operatorname*{diag}\left(  A_{y,\theta},A_{\nu,\rho
}\right)  $. With straightforward algebra, it can be shown that\textbf{ }%
\[
V=\left[
\begin{array}
[c]{cc}%
A_{y,\theta}^{-1} & 0\\
0 & A_{\nu,\rho}^{-1}+\frac{1}{\kappa}A_{\nu,\rho}^{-1}A_{\nu,\theta
}A_{y,\theta}^{-1}A_{\nu,\theta}^{\prime}A_{\nu,\rho}^{-1}%
\end{array}
\right]  ,
\]
which shows that the two sets of estimators are asymptotically independent in
our example. The form of the limiting variance for $\rho$ confirms the
intuitive derivation in Section \ref{Intuition}. Note in particular that
$A_{y,\theta}^{-1}=\Omega_{y}$ when GMM with the optimal weight matrix is
used. In general, $V$ is a random variable measurable with respect to
$\mathcal{C}$. A further application of the continuous mapping theorem shows
that $V^{-1/2}D_{n\tau}^{-1}\left(  \hat{\phi}-\phi_{0}\right)  \ $converges
to a standard Gaussian random vector.

Standard errors can now be computed based on this distributional
approximation. To this end use the following estimator $\hat{V}$ for the
asymptotic variance-covariance matrix $V$. Let $\hat{\phi}=\left(  \hat
{\theta},\hat{\rho}\right)  $ be the joint solution to the moment conditions
(\ref{Empirical_Joint_GMM}). Note that $\hat{\phi}=\left(  \hat{\beta}%
_{0,1}^{\ast},...,\hat{\beta}_{0,T}^{\ast},\hat{\beta}_{k},\hat{\alpha
}^{\left(  C\right)  },\hat{\alpha}^{\left(  A\right)  }\right)  .$ Obtain the
residuals $\hat{u}_{j,t}=\mathfrak{y}_{j,t}^{\ast}-\left(  \hat{\beta}%
_{0,t}^{\ast}+\hat{\beta}_{k}k_{j,t}+\hat{\alpha}^{\left(  C\right)  }\left(
\phi_{t}\left(  i_{j,t-1},k_{j,t-1}\right)  -\hat{\beta}_{k}k_{j,t-1}\right)
\right)  $ as well as $\hat{\nu}_{s}=Y_{s}^{\ast}-\hat{\beta}_{k}K_{s}^{\ast}$
and $\hat{e}_{s}^{\left(  A\right)  }=\hat{\nu}_{s}-\hat{\alpha}^{\left(
A\right)  }\hat{\nu}_{s-1}$ and form the matrices
\begin{equation}
\hat{\Omega}_{f}=\frac{1}{n}\sum_{t=1}^{T}\sum_{j=1}^{n}\hat{u}_{j,t}%
^{2}z_{j,t}z_{j,t}^{\prime},\text{ }\hat{\Omega}_{g}=\frac{1}{\tau}%
\sum_{s=\tau_{0}+1}^{\tau_{0}+\tau}\left(  \hat{e}_{s}^{\left(  A\right)
}\right)  ^{2}\hat{\nu}_{s-1}^{2}. \label{def-weight-matrix}%
\end{equation}
Similarly, obtain
\begin{align*}
\frac{\partial\hat{k}\left(  \beta,\rho\right)  }{\partial\theta}  &
=\frac{1}{\tau}\sum_{s=\tau_{0}+1}^{\tau_{0}+\tau}\frac{\partial g\left(
\left.  z_{s}\right\vert \hat{\beta},\hat{\rho}\right)  }{\partial\theta
},\text{ }\frac{\partial\hat{k}\left(  \beta,\rho\right)  }{\partial\rho
}=\frac{1}{\tau}\sum_{s=\tau_{0}+1}^{\tau_{0}+\tau}\frac{\partial g\left(
\left.  z_{s}\right\vert \hat{\beta},\hat{\rho}\right)  }{\partial\rho}\\
\frac{\partial\hat{h}\left(  \theta\right)  }{\partial\theta}  &  =\frac{1}%
{n}\sum_{t=1}^{T}\sum_{j=1}^{n}\frac{\partial f\left(  \left.  y_{j,t}%
\right\vert \hat{\theta}\right)  }{\partial\theta}%
\end{align*}
and
\[
\hat{A}_{y,\theta}=\frac{\partial\hat{h}\left(  \theta\right)  ^{\prime}%
}{\partial\theta}\hat{\Omega}_{f}^{-1}\frac{\partial\hat{h}\left(
\theta\right)  }{\partial\theta^{\prime}},\text{ }\hat{A}_{\nu,\rho}%
=\frac{\partial\hat{k}\left(  \beta,\rho\right)  ^{\prime}}{\partial\rho}%
\hat{\Omega}_{g}^{-1}\frac{\partial\hat{k}\left(  \beta,\rho\right)
}{\partial\rho^{\prime}},\text{ }\hat{A}_{\nu,\theta}=\frac{\partial\hat
{k}\left(  \beta,\rho\right)  ^{\prime}}{\partial\rho}\hat{\Omega}_{g}%
^{-1}\frac{\partial\hat{h}\left(  \theta\right)  }{\partial\theta^{\prime}}.
\]
The asymptotic variance-covariance matrix then can be estimated as
\[
\hat{V}=\left[
\begin{array}
[c]{cc}%
\hat{A}_{y,\theta}^{-1} & 0\\
0 & \hat{A}_{\nu,\rho}^{-1}+\frac{1}{\kappa}\hat{A}_{\nu,\rho}^{-1}\hat
{A}_{\nu,\theta}\hat{A}_{y,\theta}^{-1}\hat{A}_{\nu,\theta}^{\prime}\hat
{A}_{\nu,\rho}^{-1}%
\end{array}
\right]  .
\]
Now let $\phi_{j}$ be the $j$-th element of $\phi$ with estimator $\hat{\phi
}_{j}$. Then, a t-ratio for $\hat{\phi}_{j}$ based on the asymptotic
approximation for $\hat{\phi}$ can be constructed as $\hat{\phi}_{j}\left/
\left(  d_{j,j}\sqrt{\hat{V}_{j,j}}\right)  \right.  $ where $\hat{V}_{j,j}$
is the $j$-th diagonal element of $\hat{V}$ and $d_{j,j}$ is the $j$-th
diagonal element of $D_{n\tau}^{-1}$\textbf{.}

Focusing on the time-series parameter $\alpha^{\left(  A\right)  }$ one
obtains the following t-ratio%
\[
\hat{\alpha}^{\left(  A\right)  }\left/  \sqrt{\frac{1}{\tau}\hat{A}_{\nu
,\rho}^{-1}+\frac{1}{n}\hat{A}_{\nu,\rho}^{-1}\hat{A}_{\nu,\theta}\hat
{A}_{y,\theta}^{-1}\hat{A}_{\nu,\theta}^{\prime}\hat{A}_{\nu,\rho}^{-1}%
}\right.  ,
\]
which corresponds to the asymptotic variance formula obtained in
(\ref{stationary-result}).

\section{Unit Root Time Series Models\label{sec-unit-root-time-series-models}}

\subsection{Unit Root Problems\label{sec-unit-root}}

When the simple trend stationary paradigm does not apply, the limiting
distribution of our estimators may be more complicated. A general treatment is
beyond the scope of this paper and likely requires a case by case analysis. In
this subsection we consider a simple unit root model where initial conditions
can be neglected. We use it to exemplify additional inferential difficulties
that arise even in this relatively simple setting. In Section
\ref{Unit_Root_Limit_Theory} we consider a slightly more complex version of
the unit root model where initial conditions cannot be ignored. We show that
more complicated dependencies between the asymptotic distributions of the
cross-section and time series samples manifest. The result is a cautionary
tale of the difficulties that may present themselves when nonstationary time
series data are combined with cross-sections. We leave the development of
inferential methods for this case to future work.

We again consider the model in the previous section, except with the twist
that (i) $\rho$ is the AR(1) coefficient in the time series regression of
$z_{t}$ on $z_{t-1}$ with independent error; and (ii) $\rho$ is at
unity.\textbf{ }In the same way that led to
(\ref{expansion-theta-hat-reflecting-tilde-rho}), we obtain
\[
\sqrt{n}\left(  \widehat{\theta}-\theta\right)  \approx-A^{-1}\sqrt{n}%
\frac{\partial F_{n}\left(  \theta,\rho\right)  }{\partial\theta}-A^{-1}%
B\frac{\sqrt{n}}{\tau}\tau\left(  \widetilde{\rho}-\rho\right)
\]
For simplicity, again assume that the two terms on the right are
asymptotically independent.\textbf{ }The first term converges in distribution
to a normal distribution $N\left(  0,A^{-1}\Omega_{y}A^{-1}\right)  $, but
with $\rho=1$ and i.i.d. AR(1) errors the second term converges to
\[
\xi A^{-1}B\frac{W\left(  1\right)  ^{2}-1}{2\int_{0}^{1}W\left(  r\right)
^{2}dr},
\]
where $\xi\equiv\lim\left.  \sqrt{n}\right/  \tau$ and $W\left(  \cdot\right)
$ is the standard Wiener process. In contrast to the result in
(\ref{stationary-result}) when $\rho$ is away from unity,\textbf{ }$\sqrt
{n}/\tau$ rather than $n/\tau$ is assumed to converge to a constant. Because
$\tilde{\rho}$ is superconsistent under the unit root scenario, from a
theoretical point of view, the correction term is relevant only in cases where
$n$ is much larger than $\tau$ such that $\xi>0$ in the limit.\textbf{ }The
result is formalized in Section \ref{Unit_Root_Limit_Theory}.

The fact that the limiting distribution of $\hat{\theta}$ is no longer
Gaussian complicates inference. This discontinuity is mathematically similar
to Campbell and Yogo's (2006) observation, which leads to a question of how
uniform inference could be conducted. In principle, the problem here can be
analyzed by modifying the proposal in Phillips (2014, Section 4.3).\footnote{A
rigorous proof of the validity of the proposed uniform inference procedure is
beyond the scope of this paper and left for future research.} First, construct
the $1-\alpha_{1}$ confidence interval for $\rho$ using Mikusheva (2007). Call
it $\left[  \rho_{L},\rho_{U}\right]  $. Second, compute $\widehat{\theta
}\left(  \rho\right)  \equiv\operatorname*{argmax}_{\theta}F_{n}\left(
\theta,\rho\right)  $ for $\rho\in\left[  \rho_{L},\rho_{U}\right]  $.
Assuming that $\rho$ is fixed, characterize the asymptotic variance
$\Sigma\left(  \rho\right)  $, say, of $\sqrt{n}\left(  \widehat{\theta
}\left(  \rho\right)  -\theta\left(  \rho\right)  \right)  $, which is
asymptotically normal in general. Third, construct the $1-\alpha_{2}$
confidence region, say $CI\left(  \alpha_{2};\rho\right)  $, using asymptotic
normality and $\Sigma\left(  \rho\right)  $. Our confidence interval for
$\theta_{1}$ is then given by $\bigcup_{\rho\in\left[  \rho_{L},\rho
_{U}\right]  }CI\left(  \alpha_{2};\rho\right)  $. By Bonferroni, its
asymptotic coverage rate is expected to be at least $1-\alpha_{1}-\alpha_{2}$.

There are some cases where standard asymptotics obtain for certain parameters
in nonstationary scenarios, see for example Inoue and Kilian (2020). We expect
that such results will carry over to the case of joint cross-section and time
series inference, in which case the results in Section
\ref{sec-trend-stationary-models} could be applied. We leave the detailed
technical analysis of these cases for future research.\textbf{\ }

\subsection{Unit Root Limit Theory\label{Unit_Root_Limit_Theory}}

In this section we consider the special case where $\nu_{t}$ follows an
autoregressive process of the form $\nu_{t+1}=\rho\nu_{t}+\eta_{t}$. As in
Hansen (1992), Phillips (1987, 1988, 2014) we allow for nearly integrated
processes where $\rho=\exp\left(  \left.  \gamma\right/  \tau\right)  $ is a
scalar parameter localized to unity such that%
\begin{equation}
\nu_{\tau,t+1}=\exp\left(  \left.  \gamma\right/  \tau\right)  \nu_{\tau
,t}+\eta_{t+1} \label{Unit Root Gen Mech}%
\end{equation}
and the notation $\nu_{\tau,t}$ emphasizes that $\nu_{\tau,t}$ is a sequence
of processes indexed by $\tau$. We assume that $\tau_{0}=0$ is fixed
and\textbf{ }
\[
\tau^{-1/2}\nu_{\tau,\min\left(  1,\tau_{0}\right)  }=V\left(  0\right)
\text{ a.s.}%
\]
where $V\left(  0\right)  $ is a potentially nondegenerate random variable. In
other words, the initial condition for (\ref{Unit Root Gen Mech}) is
$\nu_{\tau,\min\left(  1,\tau_{0}\right)  }=\tau^{1/2}V\left(  0\right)  $. We
explicitly allow for the case where $V\left(  0\right)  =0$, to model a
situation where the initial condition can be ignored. This assumption is
similar, although more parametric than, the specification considered in Kurtz
and Protter (1991). We limit our analysis to the case of maximum likelihood
criterion functions. Results for moment based estimators can be developed
along the same lines as in Section \ref{sec-trend-stationary-models} but for
ease of exposition we omit the details. For the unit root version of our model
we assume that $\nu_{t}$ is observed in the data and that the only parameter
to be estimated from the time series data is $\rho.$ Further assuming a
Gaussian quasi-likelihood function we note that the score function now is
\begin{equation}
g_{\rho,t}\left(  \beta,\rho\right)  =\nu_{\tau,t-1}\left(  \nu_{\tau,t}%
-\nu_{\tau,t-1}\rho\right)  . \label{log_lik_unit_root}%
\end{equation}
The estimator $\hat{\rho}$ solving sample moment conditions based on
(\ref{log_lik_unit_root}) is the conventional OLS estimator given by
\[
\hat{\rho}=\frac{\sum_{t=\tau_{0}+1}^{\tau}\nu_{\tau,t-1}\nu_{\tau,t}}%
{\sum_{t=\tau_{0}+1}^{\tau}\nu_{\tau,t-1}^{2}}.
\]
We continue to use the definition for $f_{\theta,it}\left(  \theta
,\rho\right)  $ in Section \ref{sec-trend-stationary-models} but now consider
the simplified case where $\theta_{0}=\left(  \beta,V\left(  0\right)
\right)  $. We note that in this section, $V\left(  0\right)  $ rather than
$\nu_{\tau,\min\left(  1,\tau_{0}\right)  }$ is the common shock used in the
cross-sectional model. The implicit scaling of $\nu_{\tau,\min\left(
1,\tau_{0}\right)  }$ by $\tau^{-1/2}$ is necessary in the cross-sectional
specification to maintain a well defined model even as $\tau\rightarrow\infty$.

Consider the joint process $\left(  V_{\tau n}\left(  r\right)  ,s_{ML}%
^{y}\right)  $ where $V_{\tau n}\left(  r\right)  \equiv\tau^{-1/2}\nu
_{\tau\left[  \tau r\right]  }$, and
\[
s_{ML}^{y}\equiv s_{ML}\left(  \theta_{0},\rho_{0}\right)  \equiv\sum
_{t=1}^{T}\sum_{i=1}^{n}\frac{f_{\theta,it}}{\sqrt{n}}.
\]
\textbf{ }Note that
\[
\int_{0}^{r}V_{\tau n}\left(  u\right)  dW_{\tau n}\left(  u\right)
=\tau^{-1}\sum_{t=\tau_{0}+1}^{\tau_{0}+\left[  \tau r\right]  }\nu_{\tau
,t-1}\eta_{t}%
\]
with $W_{\tau n}\left(  r\right)  \equiv\tau^{-1/2}\sum_{t=\tau_{0}+1}%
^{\tau_{0}+\left[  \tau r\right]  }\eta_{t}$. We define the limiting process
for $V_{\tau n}\left(  r\right)  $ as%
\begin{equation}
V_{\gamma,V\left(  0\right)  }\left(  r\right)  =e^{\gamma r}V\left(
0\right)  +\int_{0}^{r}\sigma e^{\gamma\left(  r-s\right)  }dW_{\nu}\left(
s\right)  \label{Limit_V}%
\end{equation}
where $W_{\nu}$ is defined in Theorem \ref{FCLT}. When $V\left(  0\right)
=0,$ Theorem \ref{FCLT} directly implies that $e^{-\gamma\left[  r\tau\right]
/\tau}V_{\tau n}\left(  r\right)  \Rightarrow\int_{0}^{r}\sigma e^{-s\gamma
}dW_{\nu}\left(  s\right)  $ $\mathcal{C}$-stably noting that in this case
$\Omega_{\nu}\left(  s\right)  =\sigma^{2}\left(  1-\exp\left(  -2s\gamma
\right)  \right)  /2\gamma$ and $\dot{\Omega}_{\nu}\left(  s\right)
^{1/2}=\sigma e^{-s\gamma}$. The familiar result (cf. Phillips 1987) that
$V_{\tau n}\left(  r\right)  \Rightarrow\int_{0}^{r}\sigma e^{\gamma\left(
r-s\right)  }dW_{\nu}\left(  s\right)  $ then is a consequence of the
continuous mapping theorem. The case in (\ref{Limit_V}) where $V\left(
0\right)  $ is a $\mathcal{C}$-measurable random variable now follows from
$\mathcal{C}$-stable convergence of $V_{\tau n}\left(  r\right)  $. In this
section we establish joint $\mathcal{C}$-stable convergence of the triple
$\left(  V_{\tau n}\left(  r\right)  ,s_{ML}^{y},\int_{0}^{r}V_{\tau n}\left(
u\right)  dW_{\tau n}\left(  u\right)  \right)  .$

Let $\phi=\left(  \theta^{\prime},\rho\right)  ^{\prime}\in\mathbb{R}%
^{k_{\phi}},$ $\theta\in\mathbb{R}^{k_{\theta}},$ and $\rho\in\mathbb{R}.$ The
true parameters are denoted by $\theta_{0}$ and $\rho_{\tau0}=\exp\left(
\gamma_{0}/\tau\right)  $ with $\gamma_{0}\in\mathbb{R}$ and both $\theta_{0}$
and $\gamma_{0}$ bounded. We impose the following modified assumptions to
account for the the specific features of the unit root model.

\begin{condition}
\label{Diag_Unit Root CLT_Cond}Define $\mathcal{C}=\sigma\left(  V\left(
0\right)  \right)  $. Define the $\sigma$-fields $\mathcal{G}_{n,\left(
t-\min\left(  1,\tau_{0}\right)  \right)  n+i}$ in the same way as in
(\ref{Information Sets}) except that here $\tau=\kappa n$ such that dependence
on $\tau$ is suppressed and that $\nu_{t}$ is replaced with $\eta_{t}$ as in
\[
\mathcal{G}_{n,\left(  t-\min\left(  1,\tau_{0}\right)  \right)  n+i}%
=\sigma\left(  \left\{  y_{jt-1},y_{jt-2},\ldots,y_{j\min\left(  1,\tau
_{0}\right)  }\right\}  _{j=1}^{n},\left\{  \eta_{t},\eta_{t-1},\ldots
,\eta_{\min\left(  1,\tau_{0}\right)  }\right\}  ,\left(  y_{j,t}\right)
_{j=1}^{i}\right)  \vee\mathcal{C}\text{.}%
\]
Assume that \newline i) $f_{\theta,it}$ is measurable with respect to
$\mathcal{G}_{n,\left(  t-\min\left(  1,\tau_{0}\right)  \right)  n+i}%
.$\newline ii) $\eta_{t}$ is measurable with respect to $\mathcal{G}%
_{n,\left(  t-\min\left(  1,\tau_{0}\right)  \right)  n+i}$ for all
$i=1,...,n$\newline iii) for some $\delta>0$ and $C<\infty,$\ $\sup
_{it}E\left[  \left\Vert f_{\theta,it}\right\Vert ^{2+\delta}\right]  \leq
C$\newline iv) for some $\delta>0$ and $C<\infty,$\ $\sup_{t}E\left[
\left\Vert \eta_{t}\right\Vert ^{2+\delta}\right]  \leq C$\newline v)
$E\left[  f_{\theta,it}|\mathcal{G}_{n,\left(  t-\min\left(  1,\tau
_{0}\right)  \right)  n+i-1}\right]  =0$\newline vi) $E\left[  \eta
_{t}|\mathcal{G}_{n,\left(  t-\min\left(  1,\tau_{0}\right)  -1\right)
n+i}\right]  =0$ for $t>T$ and all $i=\{1,...,n\}.$\newline vii) For any
$1>r>s\geq0$ fixed let $\Omega_{\tau,\eta}^{r,s}=\tau^{-1}\sum_{t=\min\left(
1,\tau_{0}\right)  +\left[  \tau s\right]  +1}^{\tau_{0}+\left[  \tau
r\right]  }E\left[  \eta_{t}^{2}|\mathcal{G}_{n,\left(  t-\min\left(
1,\tau_{0}\right)  \right)  n}\right]  .$ Then, $\Omega_{\tau,\eta}%
^{r,s}\rightarrow_{p}\left(  r-s\right)  \sigma^{2}.$\newline viii) Assume
that $\frac{1}{n}\sum_{i=1}^{n}f_{\theta,it}f_{\theta,it}^{\prime
}\overset{p}{\rightarrow}\Omega_{ty}$ where $\Omega_{ty}$ is positive definite
a.s. and measurable with respect to $\mathcal{C}$. Let $\Omega_{y}=\sum
_{t=1}^{T}\Omega_{ty}$.
\end{condition}

Conditions \ref{Diag_Unit Root CLT_Cond}(i)-(vi) are the same as Conditions
\ref{Diag_CLT_Cond} (i)-(vi) adapted to the unit root model. Condition
\ref{Diag_Unit Root CLT_Cond}(vii) replaces Condition \ref{Omega_z}. It is
slightly more primitive in the sense that if $\eta_{t}^{2}$ is homoskedastic,
Condition \ref{Diag_Unit Root CLT_Cond}(vii) holds automatically and
convergence of $\tau^{-1}\sum_{t=\min\left(  1,\tau_{0}\right)  +\left[  \tau
s\right]  +1}^{\tau_{0}+\left[  \tau r\right]  }\eta_{t}^{2}\rightarrow\left(
r-s\right)  \sigma^{2}$ follows from an argument given in the proofs rather
than being assumed. On the other hand, Condition \ref{Diag_Unit Root CLT_Cond}%
(vii) is somewhat more restrictive than Condition \ref{Omega_z} in the sense
that it limits heteroskedasticity to be of a form that does not affect the
limiting distribution. In other words, we essentially assume $\tau^{-1}%
\sum_{t=\min\left(  1,\tau_{0}\right)  +\left[  \tau s\right]  +1}^{\tau
_{0}+\left[  \tau r\right]  }\eta_{t}^{2}$ to be proportional to $r-s$
asymptotically. This assumption is stronger than needed but helps to compare
the results with the existing unit root literature.

For Condition \ref{Diag_Unit Root CLT_Cond}(viii) we note that typically
$\Omega_{ty}\left(  \phi\right)  =E\left[  f_{\theta,it}f_{\theta,it}^{\prime
}\right]  $ and $\Omega_{ty}=\Omega_{ty}\left(  \phi_{0}\right)  $ where
$\phi_{0}=\left(  \beta_{0}^{\prime},V_{0}\left(  0\right)  ,\rho_{\tau
0}\right)  $. Thus, even if $\Omega_{ty}\left(  .\right)  $ is non-stochastic,
it follows that $\Omega_{ty}$ is random and measurable with respect to
$\mathcal{C}$ because it depends on $V\left(  0\right)  $ which is a random
variable measurable with respect to $\mathcal{C}$.

The following results are established by modifying arguments in Phillips
(1987) and Chan and Wei (1987) to account for $\mathcal{C}$-stable convergence
and by applying Theorem \ref{FCLT}.

\begin{theorem}
\label{FCLT_SI}Assume that Conditions \ref{Diag_Unit Root CLT_Cond} hold. With
$\tau_{0}=0$ and as $\tau,n\rightarrow\infty$ and $T$ fixed with $\tau=\kappa
n$ for some $\kappa\in\left(  0,\infty\right)  $ it follows that%
\[
\left(  V_{\tau n}\left(  r\right)  ,s_{ML}^{y},\int_{0}^{s}V_{\tau n}\left(
u\right)  dW_{\tau n}\left(  u\right)  \right)  \Rightarrow\left(
V_{\gamma,V\left(  0\right)  }\left(  r\right)  ,\Omega_{y}^{1/2}W_{y}\left(
1\right)  ,\int_{0}^{s}\sigma V_{\gamma,V\left(  0\right)  }\left(  u\right)
dW_{\nu}\left(  u\right)  \right)  \text{ (}\mathcal{C}\text{-stably)}%
\]
in the Skorohod topology on $D_{R^{d}}\left[  0,1\right]  .$
\end{theorem}

\begin{proof}
In Appendix \ref{proof-section-jointCLT}.
\end{proof}

We now employ Theorem \ref{FCLT_SI} to analyze the limiting behavior of
$\hat{\theta}$ when the common factors are generated from a linear unit root
process. To derive a limiting distribution for $\hat{\phi}$ we impose the
following additional assumption.

\begin{condition}
\label{Unit Root Consistency}Let $\hat{\theta}=\arg\max\sum_{t=1}^{T}%
\sum_{i=1}^{n}f\left(  y_{it}|\theta,\hat{\rho}\right)  $. Assume that
$\left(  \hat{\theta}-\theta_{0}\right)  =O_{p}\left(  n^{-1/2}\right)  $.
\end{condition}

\begin{condition}
\label{Unit Root Uniform Hessian Y}Let $\kappa=\lim n/\tau^{2}$. Let
$A_{y,\theta}\left(  \phi\right)  =\sum_{t=1}^{T}E\left[  \partial
f_{\theta,it}/\partial\theta^{\prime}\right]  $, $A_{y,\rho}\left(
\phi\right)  =\sum_{t=1}^{T}E\left[  \partial f_{\theta,it}/\partial
\rho\right]  $, and define $A^{y}\left(  \phi\right)  =\left[
\begin{array}
[c]{cc}%
A_{y,\theta}\left(  \phi\right)  & \sqrt{\kappa}A_{y,\rho}\left(  \phi\right)
\end{array}
\right]  $ where $A\left(  \phi\right)  $ is a $k_{\theta}\times k_{\phi}$
dimensional matrix of non-random functions $\phi\rightarrow\mathbb{R}$. Assume
that $A_{y,\theta}\left(  \phi_{0}\right)  $ is full rank almost surely.
Assume that for some $\varepsilon>0,$%
\[
\sup_{\phi:\left\Vert \phi-\phi_{0}\right\Vert \leq\varepsilon}\left\Vert
\frac{\partial\tilde{s}^{y}\left(  \phi\right)  }{\partial\phi^{\prime}%
}D_{n\tau}-A^{y}\left(  \phi\right)  \right\Vert =o_{p}\left(  1\right)  .
\]

\end{condition}

We make the possibly simplifying assumption that $A\left(  \phi\right)  $ only
depends on the factors through the parameter $\theta.$

\begin{theorem}
\label{CLT_MLE_Unit Root}Assume that Conditions \ref{Diag_Unit Root CLT_Cond},
\ref{Unit Root Consistency} and \ref{Unit Root Uniform Hessian Y} hold. It
follows that
\[
\sqrt{n}\left(  \hat{\theta}-\theta_{0}\right)  \overset{d}{\rightarrow
}-A_{y,\theta}^{-1}\Omega_{y}^{1/2}W_{y}\left(  1\right)  -\sqrt{\kappa
}A_{y,\theta}^{-1}A_{y,\rho}\left(  \int_{0}^{1}V_{\gamma,V\left(  0\right)
}^{2}\left(  r\right)  dr\right)  ^{-1}\left(  \int_{0}^{1}\sigma
V_{\gamma,V\left(  0\right)  }\left(  r\right)  dW_{\nu}\left(  r\right)
\right)  \text{ (}\mathcal{C}\text{-stably).}%
\]

\end{theorem}

\begin{proof}
In Appendix \ref{proof-section-jointCLT}.
\end{proof}

Note that the term $\left(  \int_{0}^{1}V_{\gamma,V\left(  0\right)  }%
^{2}\left(  r\right)  dr\right)  ^{-1}$ corresponds to $A_{\nu,\rho}^{-1}$ in
the stationary case when the time series model does not depend on
cross-sectional parameters. The result in Theorem \ref{CLT_MLE_Unit Root} is
an example that shows how common factors affecting both time series and
cross-section data can lead to non-standard limiting distributions. In this
case, the initial condition of the unit root process in the time series
dimension causes dependence between the components of the asymptotic
distribution of $\hat{\theta}$ because both $\Omega_{y}$ and $V_{\gamma
,V\left(  0\right)  }$ in general depend on $V\left(  0\right)  $. Thus, the
situation encountered here is generally more difficult than the one considered
in Campbell and Yogo (2006) and Phillips (2014). In addition, because the
limiting distribution of $\hat{\theta}$ is not mixed asymptotically normal,
simple pivotal test statistics as in Andrews (2005) are not readily available
contrary to the stationary case.

\section{Summary}

We develop a new limit theory for combined cross-sectional and time-series
data sets. We focus on situations where the two data sets are interdependent
because of common factors that affect both. The concept of stable convergence
is used to handle this dependence when proving a joint Central Limit Theorem.
Our analysis is cast in a generic framework of cross-section and time-series
based criterion functions that jointly, but not individually, identify the
parameters. Within this framework, we show how our limit theory can be used to
derive asymptotic approximations to the sampling distribution of estimators
that are based on data from both samples. We explicitly consider the unit root
case as an example where particularly difficult to handle limiting expressions
arise. Our results are expected to be helpful for the econometric analysis of
rational expectation models involving individual decision making as well as
general equilibrium settings. We investigate these topics, and related
implementation issues, in a companion paper HKM20. The question of efficient
inference in the context of our model is an interesting topic for future
research, but beyond the scope of the current paper.

\appendix{}

\begin{center}
{\Large Appendix}
\end{center}

\section{Details for Section \ref{OP}\label{OP_detail}}

We consider a simplified version of Olley and Pakes' (1996) model without any
aggregate shock. Assume that $\omega_{j,t}$ follows an AR(1)
process\footnote{Olley and Pakes (1996) adopted a non-parametric specification
for the dynamics of $\omega_{j,t}$, but Ackerberg, Caves and Frazer (2015)
adopted a parametric specification. The parametric specification makes it
easier to recognize the source of complication in the presence of aggregate
shocks.}%
\[
\omega_{j,t}=\alpha\omega_{j,t-1}+e_{j,t},
\]
where we assume that the intercept term is zero for notational simplicity, and
that our parameters of interests are $\left(  \beta_{k},\alpha\right)  $. This
means that we can write the conditional expectation of $\mathfrak{y}%
_{j,t+1}^{\ast}$ given the information $I_{j,t}$ available at time $t$ as%
\begin{align*}
E\left[  \left.  \mathfrak{y}_{j,t+1}^{\ast}\right\vert I_{j,t}\right]   &
=\beta_{k}k_{j,t+1}+\alpha\omega_{j,t}\\
&  =\beta_{k}k_{j,t+1}+\alpha\left(  \phi_{t}\left(  i_{j,t},k_{j,t}\right)
-\beta_{k}k_{j,t}\right) \\
&  =\beta_{k}k_{j,t+1}+\alpha\left(  \phi_{t}\left(  i_{j,t},k_{j,t}\right)
-\beta_{k}k_{j,t}\right)  ,
\end{align*}
and that $\left(  \beta_{k},\alpha\right)  $ can be identified by the
conditional moment restriction
\[
0=E\left[  \left.  \mathfrak{y}_{j,t+1}^{\ast}-\left(  \beta_{k}%
k_{j,t+1}+\alpha\left(  \phi_{t}\left(  i_{j,t},k_{j,t}\right)  -\beta
_{k}k_{j,t}\right)  \right)  \right\vert I_{j,t}\right]
\]
using cross sectional variation. This gives the basic intuition underlying
(\ref{C-GMM}).

From
\[
\mathfrak{y}_{j,t}^{\ast}-\beta_{k}k_{j,t}=\omega_{j,t}+\eta_{j,t}=\nu
_{t}+\varepsilon_{j,t}+\eta_{j,t},
\]
we can infer that%
\begin{align}
\operatorname*{plim}_{n\rightarrow\infty}n^{-1}\sum_{j=1}^{n}\left(
\mathfrak{y}_{j,t}^{\ast}-\nu_{t}-\beta_{k}k_{j,t}\right)  \left(
\mathfrak{y}_{j,t-1}^{\ast}-\nu_{t-1}-\beta_{k}k_{j,t-1}\right)   &
=\alpha^{\left(  C\right)  }\sigma_{\varepsilon}^{2}%
\label{Cross-Section_Moments}\\
\operatorname*{plim}_{n\rightarrow\infty}n^{-1}\sum_{j=1}^{n}\left(
\mathfrak{y}_{j,t}^{\ast}-\nu_{t}-\beta_{k}k_{j,t}\right)  \left(
\mathfrak{y}_{j,t-2}^{\ast}-\nu_{t-2}-\beta_{k}k_{j,t-2}\right)   &  =\left(
\alpha^{\left(  C\right)  }\right)  ^{2}\sigma_{\varepsilon}^{2}\nonumber\\
&  \vdots\nonumber
\end{align}
If the panel consists of $T$ observations, we can identify $T+3$ parameters,
including $\nu_{1},\ldots,\nu_{T},$ $\beta_{k},$ $\alpha^{\left(  C\right)
},$ $\sigma_{\varepsilon}^{2}$, using the $T\left(  T-1\right)  /2$ moments
based on all available pairs of time periods.\footnote{It may be tempting to
use the Yule-Walker equation
\begin{align*}
E\left[  \omega_{j,t}\omega_{j,t}\right]   &  =\sigma_{\nu}^{2}+\sigma
_{\varepsilon}^{2}\\
E\left[  \omega_{j,t}\omega_{j,t-1}\right]   &  =\rho^{\left(  A\right)
}\sigma_{\nu}^{2}+\rho^{\left(  C\right)  }\sigma_{\varepsilon}^{2}%
\end{align*}
but the expectation operator on the LHS refers to the joint distribution
involving both time series and cross sectional variations, and as such, is not
implementable in cross sectional data.}

While the moment conditions in (\ref{Cross-Section_Moments}) demonstrate
identification in the cross-section of certain parameters, these moments are
not suitable for our limit theory which requires estimating functions to have
a martingale difference sequence property. To address this issue we propose
the following moment conditions, which is similar to Ackerberg, Caves and
Frazer's (2015) parametric rendition of Olley and Pakes' (1996) moment
condition. Recall our assumption that $\beta_{l}$ and $\phi_{t}\left(
i_{j,t},k_{j,t}\right)  $ are known. From
\[
\mathfrak{y}_{j,t}=\beta_{l}l_{j,t}+\phi_{t}\left(  i_{j,t},k_{j,t}\right)
+\eta_{j,t},
\]
we obtain%
\[
\mathfrak{y}_{j,t}^{\ast}=\mathfrak{y}_{j,t}-\beta_{l}l_{j,t}=\phi_{t}\left(
i_{j,t},k_{j,t}\right)  +\eta_{j,t}%
\]
from which we further obtain
\begin{equation}
\eta_{j,t}=\mathfrak{y}_{j,t}^{\ast}-\phi_{t}\left(  i_{j,t},k_{j,t}\right)  .
\label{eta-alt}%
\end{equation}
We also have%
\[
\mathfrak{y}_{j,t}^{\ast}-\beta_{k}k_{j,t}=\nu_{t}+\varepsilon_{j,t}%
+\eta_{j,t}%
\]
so%
\begin{equation}
\varepsilon_{j,t}=\mathfrak{y}_{j,t}^{\ast}-\beta_{k}k_{j,t}-\nu_{t}%
-\eta_{j,t}. \label{epsilon-alt}%
\end{equation}
Combining (\ref{eta-alt}) and (\ref{epsilon-alt}), we obtain%
\begin{align*}
\varepsilon_{j,t}  &  =\mathfrak{y}_{j,t}^{\ast}-\beta_{k}k_{j,t}-\nu
_{t}-\left(  \mathfrak{y}_{j,t}^{\ast}-\phi_{t}\left(  i_{j,t},k_{j,t}\right)
\right) \\
&  =\phi_{t}\left(  i_{j,t},k_{j,t}\right)  -\nu_{t}-\beta_{k}k_{j,t},
\end{align*}
which can be combined with
\begin{align*}
\mathfrak{y}_{j,t+1}^{\ast}-\beta_{k}k_{j,t+1}  &  =\nu_{t+1}+\varepsilon
_{j,t+1}+\eta_{j,t+1}\\
&  =\nu_{t+1}+\left(  \alpha^{\left(  C\right)  }\varepsilon_{j,t}%
+e_{j,t+1}^{\left(  C\right)  }\right)  +\eta_{j,t+1}%
\end{align*}
to yield%
\[
\mathfrak{y}_{j,t+1}^{\ast}-\beta_{k}k_{j,t+1}=\nu_{t+1}+\alpha^{\left(
C\right)  }\left(  \phi_{t}\left(  i_{j,t},k_{j,t}\right)  -\nu_{t}-\beta
_{k}k_{j,t}\right)  +e_{j,t+1}^{\left(  C\right)  }+\eta_{j,t+1}.
\]
After some straightforward algebra, we obtain%
\[
\mathfrak{y}_{j,t+1}^{\ast}=\beta_{0,t+1}^{\ast}+\beta_{k}k_{j,t+1}%
+\alpha^{\left(  C\right)  }\left(  \phi_{t}\left(  i_{j,t},k_{j,t}\right)
-\beta_{k}k_{j,t}\right)  +\left(  e_{j,t+1}^{\left(  C\right)  }+\eta
_{j,t+1}\right)  ,
\]
where
\[
\beta_{0,t+1}^{\ast}\equiv\nu_{t+1}-\alpha^{\left(  C\right)  }\nu_{t}%
\]
and the error $e_{j,t+1}^{\left(  C\right)  }+\eta_{j,t+1}$ is orthogonal to
the past variables such as $k_{j,t},\mathfrak{y}_{j,t}^{\ast}$. Therefore, if
the $z_{j,t}$ is an instrument uncorrelated with the error $e_{j,t+1}^{\left(
C\right)  }+\eta_{j,t+1}$, we can use the moment (\ref{C-GMM}) as a basis of
estimating the parameters $\left(  \beta_{0,t+1}^{\ast},\beta_{k}%
,\alpha^{\left(  C\right)  }\right)  $.

\section{Detailed Calculations for Example \ref{Example_AR_Mixing}%
\label{Example_AR_Mxing_Calc}}

First note that it follows that
\[
E\left[  z_{s}u_{s+1}|\mathcal{G}_{\tau n,\left(  s-\min\left(  1,\tau
_{0}\right)  \right)  n+i}\right]  =z_{s}E\left[  u_{s+1}|\mathcal{C}\right]
,
\]
where the equality uses the fact that $z_{s}$ is measurable with respect to
$\mathcal{G}_{\tau n,\left(  s-\min\left(  1,\tau_{0}\right)  \right)  n+i}$
and that $u_{s+1}$ is independent of $\sigma\left(  \left\{  z_{s}%
,z_{s-1},\ldots,z_{\tau_{0}}\right\}  \right)  .$ To evaluate $E\left[
u_{s+1}|\mathcal{C}\right]  $ consider the joint distribution of $u_{s+1}$ for
$s<0$ and $\nu_{1}=z_{1},$ which is Gaussian,
\[
N\left(  \left[
\begin{array}
[c]{c}%
0\\
0
\end{array}
\right]  ,\left[
\begin{array}
[c]{cc}%
1 & \rho^{\left\vert s\right\vert }\\
\rho^{\left\vert s\right\vert } & \frac{1}{1-\rho^{2}}%
\end{array}
\right]  \right)  .
\]
This implies that $E\left[  u_{s+1}|\mathcal{C}\right]  =\rho^{\left\vert
s\right\vert }\left(  1-\rho^{2}\right)  z_{1}$. Evaluating the $L_{2}$ norm
of Condition \ref{Diag_CLT_Cond}(vii)\ leads to
\begin{align*}
&  \left\Vert E\left[  \left.  \psi_{\tau,s}^{\nu}\right\vert \mathcal{G}%
_{\tau n,\left(  s-\min\left(  1,\tau_{0}\right)  -1\right)  n+i}\right]
\right\Vert _{2}^{2}\\
&  \leq\left\vert \rho\right\vert ^{\left\vert s\right\vert }\left(
1-\rho^{2}\right)  \left(  E\left[  \left(  z_{s}z_{1}-E\left[  z_{s}%
z_{1}\right]  \right)  ^{2}\right]  +\left(  E\left[  z_{s}z_{1}\right]
\right)  ^{2}\right) \\
&  =\left\vert \rho\right\vert ^{\left\vert s\right\vert }\left(  1-\rho
^{2}\right)  E\left[  \left(  z_{s}z_{1}-E\left[  z_{s}z_{1}\right]  \right)
^{2}\right]  +\left\vert \rho\right\vert ^{\left\vert s\right\vert }\left(
1-\rho^{2}\right)  \left(  \frac{\rho^{1-s}}{1-\rho^{2}}\right)  ^{2}\\
&  =\left\vert \rho\right\vert ^{\left\vert s\right\vert }\left(  1-\rho
^{2}\right)  \frac{\rho^{2-2s}+1}{\left(  1-\rho^{2}\right)  ^{2}}+\left\vert
\rho\right\vert ^{\left\vert s\right\vert }\frac{\rho^{2-2s}}{1-\rho^{2}}\\
&  =\left\vert \rho\right\vert ^{\left\vert s\right\vert }\frac{\rho
^{2+2\left\vert s\right\vert }+1}{1-\rho^{2}}+\left\vert \rho\right\vert
^{\left\vert s\right\vert }\frac{\rho^{2+2\left\vert s\right\vert }}%
{1-\rho^{2}}\\
&  =O\left(  \left\vert \rho\right\vert ^{\left\vert s\right\vert }\right)
=o\left(  \left\vert s\right\vert ^{-\left(  1+\delta\right)  }\right)
\end{align*}
such that \textbf{ }%
\[
\left\Vert E\left[  \left.  \psi_{\tau,s}^{\nu}\right\vert \mathcal{G}_{\tau
n,\left(  s-\min\left(  1,\tau_{0}\right)  -1\right)  n+i}\right]  \right\Vert
_{2}=O\left(  \left\vert \rho\right\vert ^{\left\vert s\right\vert /2}\right)
=o\left(  \left\vert s\right\vert ^{-\left(  1+\delta\right)  /2}\right)  .
\]

\section{Proofs for Section \ref{Section_JointCLT}%
\label{proof-section-jointCLT}}

The proof of the functional central limit theorem is given in Section
\ref{Main_CLT_Proof} below. Our proof is self-contained but follows the
strategy of Billingsley (1968) for the case of conventional weak convergence
adapted to our setting of stable convergence. The proof consists of three
steps: (a) establishing finite dimensional convergence through a stable
central limit theorem, (b) establishing tightness of the empirical process and
(c) providing a stochastic process representation of the limiting
distribution. Tightness is established by extending techniques developed by
Billingsley (1968)\ to our setting. Finally, the characterization of the
limiting distribution is based on a proof strategy in Rootzen (1983), which we
again adapt to our setting.

\subsection{Auxiliary Results\label{Aux_Results}}

For ease of reference we present two results that are used in the proof of
Theorem \ref{FCLT}. The first result is Theorem 1 of Kuersteiner and Prucha (2013).

\begin{theorem}
[Theorem 1, Kuersteiner and Prucha (2013)]Let $\left\{  S_{nq,}\mathcal{F}%
_{nq},1\leq q\leq k_{n},n\geq1\right\}  $ be a zero mean, square integrable
martingale array with differences $X_{ni}.$ Let $\mathcal{F}_{0}=\cap
_{n=1}^{\infty}\mathcal{F}_{n0}$ with $\mathcal{F}_{n0}\subseteq
\mathcal{F}_{n1}$ for each $n$ and $E\left[  X_{n1}|\mathcal{F}_{n0}\right]
=0$ and let $\eta^{2}$ be an a.s. finite random variable measurable w.r.t.
$\mathcal{F}_{0}$. If $\max_{q}\left\vert X_{nq}\right\vert
\overset{p}{\rightarrow}0,$ $\sum_{q=1}^{k_{n}}X_{nq}^{2}%
\overset{p}{\rightarrow}\eta^{2}$ and $E\left(  \max_{q}X_{nq}^{2}\right)  $is
bounded in $n$, then
\[
S_{nk_{n}}=\sum_{v=1}^{k_{n}}X_{nv}\overset{d}{\rightarrow}Z\text{
(}\mathcal{F}_{0}\text{-stably)}%
\]
where the random variable $Z$ has characteristic function $E\left[
\exp\left(  -\frac{1}{2}\eta^{2}t^{2}\right)  \right]  $. In particular,
$S_{nk_{n}}\overset{d}{\rightarrow}\eta\xi$ ($\mathcal{F}_{0}$-stably) where
$\xi$ $\sim N(0,1)$ is independent of $\eta$ (possibly after redefining all
variables on an extended probability space).
\end{theorem}

The second result is Theorem 8.3 of Billingsley (1968). To state the theorem
the following notation is needed, see Billingsley (1968, p.19-20 and p.55).
Let $C$ be the space of continuous functions on $\left[  0,1\right]  $ with
the uniform metric. Let $\mathcal{C}$ be the class of Borel sets in $C$. Let
$P_{n}$ be a sequence of probability measures on $\left(  C,\mathcal{C}%
\right)  $.

\begin{theorem}
[Theorem 8.3, Billingsley (1968)]The sequence $\left\{  P_{n}\right\}  $ is
tight if these two conditions are satisfied:\newline(i) For each positive $c$
there exists an $a$ such that $P_{n}\left(  x:\left\vert x\left(  0\right)
\right\vert >a\right)  \leq c,$ $n\geq1$.\newline(ii) For each positive $c$
and $\varepsilon$ there is a $\delta,$with $0<\delta<1$ and an integer $n_{0}$
such that
\[
\frac{1}{\delta}P_{n}\left(  x:\sup_{t\leq s\leq t+\delta}\left\vert x\left(
s\right)  -x\left(  t\right)  \right\vert \geq c\right)  \leq\varepsilon
,\text{ }n\geq n_{0}%
\]
for all $t$.
\end{theorem}

\subsection{Proof of Theorem \label{Main_CLT_Proof}\ref{FCLT}}

We first establish that the stable functional central limit theorem follows
from establishing finite dimensional convergence and tightness. To see this
note that JS (p.512, Definition 5.28) define stable convergence for sequences
$Z^{n}$ defined on a Polish space as in
(\ref{Def Functional C Stable Convergence}). We adopt the definition in JS to
our setting, noting that by JS\ (p.328, Theorem 1.14), $D_{\mathbb{R}%
^{k_{\theta}}\times\mathbb{R}^{k_{\rho}}}\left[  0,1\right]  $\textbf{
}equipped with the Skorohod topology is a Polish space.\textbf{ }Following
Billingsley (1968, p. 120) let $r_{1}<r_{2}<\cdots<r_{k}$ be an arbitrary
finite partition of $\left[  0,1\right]  $ \textbf{ }and \textbf{ }$\pi
_{r_{1},...,r_{k}}Z^{n}=\left(  Z_{r_{1}}^{n},...,Z_{r_{k}}^{n}\right)  $ be
the coordinate projections of $Z^{n}.$ By JS Proposition VIII5.33(iv)\textbf{
}$\mathcal{C}$-stable convergence of $Z^{n}$ to $Z$ is equivalent to the
following two statements: (1) $Z^{n}$ is tight\ and (2) for any bounded
continuous function $H$ with domain $D_{\mathbb{R}^{k_{\theta}}\times
\mathbb{R}^{k_{\rho}}}\left[  0,1\right]  $ and for all $A\mathfrak{\in
}\mathcal{C}$ and letting $1_{A}$ be the indicator function of the set $A$ it
follows that \textbf{ }$E\left[  1_{A}H\left(  Z^{n}\right)  \right]
\ $converges. Part (1) is established by noting that by Billingsley (1968),
Theorem 15.5, convergence under the uniform metric implies tightness for
partial sum processes in $D_{\mathbb{R}^{k_{\theta}}\times\mathbb{R}^{k_{\rho
}}}\left[  0,1\right]  $ that have stochastically bounded initial conditions -
see Billingsley (1968, Condition 15.17) which is satisfied in our case. Part
(2) is established as follows. First, prove stable convergence of the finite
dimensional vector of random variables $Z_{r_{1}}^{n},...,Z_{r_{k}}^{n}$
defined on $\mathbb{R}^{k}$ using a multivariate stable central limit theorem.
This is shown in Step (a) of the proof below. Second, use an argument based on
on the proof of JS Theorem VIII5.7 and VIII5.14, JS (p.509) as follows: For
$P$ defined in (\ref{P-Measure}), define the new probability measure
$\tilde{P}\left(  d\omega\right)  =P\left(  d\omega\right)  1_{A}\left(
\omega\right)  $ such that $E\left[  1_{A}H\left(  Z^{n}\right)  \right]
=\tilde{E}\left[  H\left(  Z^{n}\right)  \right]  $ where $\tilde{E}$ is the
expectation with respect to $\tilde{P}.$ Then, by the stable coordinate-wise
CLT, the finite dimensional distributions of $Z^{n}$ converge under $\tilde
{P}.$ Further, for any compact subset $K$ of $D_{\mathbb{R}^{k_{\theta}}%
\times\mathbb{R}^{k_{\rho}}}\left[  0,1\right]  $ it follows that $\tilde
{P}\left(  Z^{n}\in K\right)  \leq P\left(  Z^{n}\in K\right)  $ such that it
is sufficient to show that $Z^{n}$ is tight under the measure $P.$ This shows
that a tightness argument following Billingsley (1968, Theorem 8.3) is
sufficient to establish Part (2) above. This is done in Step (b) of the proof
below, where we show that Billingsley (1968, Theorem 8.3ii) holds.

We now proceed by first establishing finite dimensional stable convergence in
Step (a) and tightness in Step (b).\textbf{ }Finally, in Step (c) we give a
stochastic process representation for the limiting distribution.\textbf{ }

\subparagraph{(a) Finite Dimensional Convergence}

For finite dimensional convergence fix $r_{1}<r_{2}<\cdots<r_{k}\in\left[
0,1\right]  $. We use the notational convention $r_{0}=0$ below. Define the
increment
\begin{equation}
\Delta X_{n\tau}\left(  r_{i}\right)  =X_{n\tau}\left(  r_{i}\right)
-X_{n\tau}\left(  r_{i-1}\right)  . \label{D_FCLT1}%
\end{equation}
Since there is a one to one mapping between $X_{n\tau}\left(  r_{1}\right)
,...,X_{n\tau}\left(  r_{k}\right)  $ and $X_{n\tau}\left(  r_{1}\right)
,\Delta X_{n\tau}\left(  r_{2}\right)  ,...,\Delta X_{n\tau}\left(
r_{k}\right)  $ we establish joint convergence of the latter. The proof
proceeds by checking that the conditions of Theorem 1 in Kuersteiner and
Prucha (2013) hold. For the convenience of the reader Kuersteiner and Prucha
(2013, Theorem 1)\ is stated in Appendix \ref{Aux_Results}. Let $k_{n}%
=\max(T,\tau)n$ where both $n\rightarrow\infty$ and $\tau\rightarrow\infty$
such that clearly $k_{n}\rightarrow\infty$ (this is a diagonal limit in the
terminology of Phillips and Moon, 1999). Let $d=k_{\phi}=k_{\theta}+k_{\rho}%
$.\textbf{ }To handle the fact that $X_{n\tau}\in\mathbb{R}^{d}$ we use Lemmas
A.1 - A.3 in Phillips and Durlauf (1986). Define $\lambda_{j}=\left(
\lambda_{j,y}^{\prime},\lambda_{j,\nu}^{\prime}\right)  ^{\prime}$and let
$\lambda=\left(  \lambda_{1},\ldots,\lambda_{k}\right)  \in\mathbb{R}^{dk}$
with $\left\Vert \lambda\right\Vert =1.$ Define $t^{\ast}=t-\min\left(
1,\tau_{0}\right)  $.

For each $n$ and $\tau_{0}$ define the mapping $q\left(  t,i\right)
:\mathbb{N}_{+}^{2}\rightarrow\mathbb{N}_{+}$ as $q\left(  i,t\right)  \equiv
t^{\ast}n+i$ \textbf{ }and note that $q\left(  i,t\right)  $ is invertible, in
particular for each $q\in\left\{  1,...,k_{n}\right\}  $ there is a unique
pair $t,i$ such that $q\left(  i,t\right)  =q.$ We often use shorthand
notation $q$ for $q\left(  i,t\right)  $. Let
\begin{equation}
\ddot{\psi}_{q\left(  i,t\right)  }\equiv\sum_{j=1}^{k}\lambda_{j}^{\prime
}\left(  \Delta\tilde{\psi}_{it}\left(  r_{j}\right)  -E\left[  \Delta
\tilde{\psi}_{it}\left(  r_{j}\right)  |\mathcal{G}_{\tau n,t^{\ast}%
n+i-1}\right]  \right)  \label{D_FCLT2}%
\end{equation}
where
\begin{equation}
\Delta\tilde{\psi}_{it}\left(  r_{j}\right)  =\tilde{\psi}_{it}\left(
r_{j}\right)  -\tilde{\psi}_{it}\left(  r_{j-1}\right)  ;\text{ }\Delta
\tilde{\psi}_{it}\left(  r_{1}\right)  =\tilde{\psi}_{it}\left(  r_{1}\right)
. \label{D_FCLT3}%
\end{equation}
Note that $\Delta\tilde{\psi}_{it}\left(  r_{j}\right)  =\left(  \Delta
\tilde{\psi}_{it}^{y}\left(  r_{j}\right)  ,\Delta\tilde{\psi}_{t}^{\nu
}\left(  r_{j}\right)  \right)  ^{\prime}$with
\begin{equation}
\Delta\tilde{\psi}_{it}^{y}\left(  r_{j}\right)  =\left\{
\begin{array}
[c]{cc}%
\tilde{\psi}_{it}^{y} & \text{for }j=1\\
0 & \text{otherwise}%
\end{array}
\right.  \label{D_FCLT4}%
\end{equation}
and
\begin{equation}
\Delta\tilde{\psi}_{it}^{\nu}\left(  r_{j}\right)  =\left\{
\begin{array}
[c]{cc}%
\tilde{\psi}_{\tau,t}^{\nu}\left(  r_{j}\right)  & \text{if }\left[  \tau
r_{j-1}\right]  <t-\tau_{0}\leq\left[  \tau r_{j}\right]  \text{ and }i=1\\
0 & \text{otherwise}%
\end{array}
\right.  . \label{D_FCLT5}%
\end{equation}
Using this notation and noting that $\sum_{t=\min(1,\tau_{0}+1)}^{\max
(T,\tau_{0}+\tau)}\sum_{i=1}^{n}\ddot{\psi}_{q\left(  i,t\right)  }=\sum
_{q=1}^{k_{n}}\ddot{\psi}_{q}$, we write%
\begin{align}
&  \lambda_{1}^{\prime}X_{n\tau}\left(  r_{1}\right)  +\sum_{j=2}^{k}%
\lambda_{j}^{\prime}\Delta X_{n\tau}\left(  r_{j}\right) \nonumber\\
&  =\sum_{q=1}^{k_{n}}\ddot{\psi}_{q}+\sum_{t=\min(1,\tau_{0}+1)}^{\max
(T,\tau_{0}+\tau)}\sum_{i=1}^{n}\sum_{j=1}^{k}\lambda_{j}^{\prime}E\left[
\left.  \Delta\tilde{\psi}_{it}\left(  r_{j}\right)  \right\vert
\mathcal{G}_{\tau n,t^{\ast}n+i-1}\right]  \label{D_FCLT5_1}%
\end{align}
First analyze the term $\sum_{q=1}^{k_{n}}\ddot{\psi}_{q}$. Note that
$\psi_{n,it}^{y}$ is measurable with respect to $\mathcal{G}_{\tau n,t^{\ast
}n+i}$ by construction. Note that by (\ref{D_FCLT2}), (\ref{D_FCLT4}) and
(\ref{D_FCLT5}) the individual components of $\ddot{\psi}_{q}$ are either $0$
or equal to $\tilde{\psi}_{it}\left(  r_{j}\right)  -E\left[  \tilde{\psi
}_{it}\left(  r_{j}\right)  |\mathcal{G}_{\tau n,t^{\ast}n+i-1}\right]  $
respectively. This implies that $\ddot{\psi}_{q}$ is measurable with respect
to $\mathcal{G}_{\tau n,q},$ noting in particular that $E\left[  \tilde{\psi
}_{it}\left(  r_{j}\right)  |\mathcal{G}_{\tau n,t^{\ast}n+i-1}\right]  $ is
measurable w.r.t $\mathcal{G}_{\tau n,t^{\ast}n+i-1}$ by the properties of
conditional expectations and $\mathcal{G}_{\tau n,t^{\ast}n+i-1}%
\subset\mathcal{G}_{\tau n,q}$. By construction, $E\left[  \ddot{\psi}%
_{q}|\mathcal{G}_{\tau n,q-1}\right]  =0$. This establishes that for
$S_{nq}=\sum_{s=1}^{q}\ddot{\psi}_{s}$,
\[
\left\{  S_{nq},\mathcal{G}_{\tau n,q},1\leq q\leq k_{n},n\geq1\right\}
\]
is a mean zero martingale array with differences $\ddot{\psi}_{q}$.

To establish finite dimensional convergence we follow Kuersteiner and Prucha
(2013) in the proof of their Theorem 2. To establish the limiting distribution
of $\sum_{q=1}^{k_{n}}\ddot{\psi}_{q}$ we check that%
\begin{equation}
\sum_{q=1}^{k_{n}}E\left[  \left\vert \ddot{\psi}_{q}\right\vert ^{2+\delta
}\right]  \rightarrow0, \label{KP1}%
\end{equation}%
\begin{equation}
\sum_{q=1}^{k_{n}}\ddot{\psi}_{q}^{2}\overset{p}{\rightarrow}\sum
_{t\in\left\{  1,..,T\right\}  }\lambda_{1,y}^{\prime}\Omega_{yt}\lambda
_{1,y}+\sum_{j=1}^{k}\lambda_{j,\nu}^{\prime}\Omega_{\nu}\left(  r_{j}%
-r_{j-1}\right)  \lambda_{j,\nu}, \label{KP2}%
\end{equation}
and
\begin{equation}
\sup_{n}E\left[  \left(  \sum_{q=1}^{k_{n}}E\left[  \left.  \ddot{\psi}%
_{q}^{2}\right\vert \mathcal{G}_{\tau n,q-1}\right]  \right)  ^{1+\delta
/2}\right]  <\infty, \label{KP3}%
\end{equation}
which are adapted to the current setting from Conditions (A.26), (A.27) and
(A.28) in Kuersteiner and Prucha (2013). (These conditions in turn are related
to conditions of Hall and Heyde (1980) and are shown by Kuersteiner and Prucha
(2013) to be sufficient for their Theorem 1.) We check these conditions in
Sections \ref{sec-proof-of-KP1}, \ref{sec-proof-of-KP2}, and
\ref{sec-proof-of-KP3} later, which establishes that (\ref{KP1}), (\ref{KP2})
and (\ref{KP3}) hold and thus establishes the CLT for $\sum_{q=1}^{k_{n}}%
\ddot{\psi}_{q}$. In Section \ref{D_FCLT5_1-second}, we also show that the
second term in (\ref{D_FCLT5_1}) can be neglected. We therefore have
\begin{equation}
\lambda_{1}^{\prime}X_{n\tau}\left(  r_{1}\right)  +\sum_{j=2}^{k}\lambda
_{j}^{\prime}\Delta X_{n\tau}\left(  r_{j}\right)  =\sum_{q=1}^{k_{n}}%
\ddot{\psi}_{q}+o_{p}\left(  1\right)  . \label{DCLT_D25}%
\end{equation}

We have shown that the conditions of Theorem 1 of Kuersteiner and Prucha
(2013) hold by establishing (\ref{KP1}), (\ref{KP2}), (\ref{KP3}) and
(\ref{DCLT_D25}). Applying the Cramer-Wold theorem to the vector
\[
Y_{nt}=\left(  X_{n\tau}\left(  r_{1}\right)  ^{\prime},\Delta X_{n\tau
}\left(  r_{2}\right)  ^{\prime}...,\Delta X_{n\tau}\left(  r_{k}\right)
^{\prime}\right)  ^{\prime},
\]
and Theorem 1 in Kuersteiner and Prucha (2013), we obtain that for all fixed
$r_{1},..,r_{k}$ and using the convention that $r_{0}=0,$
\begin{equation}
E\left[  \exp\left(  i\lambda^{\prime}Y_{nt}\right)  \right]  \rightarrow
E\left[  \exp\left(  -\frac{1}{2}\left(  \sum_{t\in\left\{  1,..,T\right\}
}\lambda_{1,y}^{\prime}\Omega_{yt}\lambda_{1,y}+\sum_{j=1}^{k}\lambda_{j,\nu
}^{\prime}\left(  \Omega_{\nu}\left(  r_{j}\right)  -\Omega_{\nu}\left(
r_{j-1}\right)  \right)  \lambda_{j,\nu}\right)  \right)  \right]  .
\label{Fidi_CharFun_Convergence}%
\end{equation}
When $\Omega_{\nu}\left(  r\right)  =r\Omega_{\nu}$ for all $r\in\left[
0,1\right]  $ and some $\Omega_{\nu}$ positive definite and measurable w.r.t
$\mathcal{C}$ this result simplifies to
\[
\sum_{j=1}^{k}\lambda_{j,\nu}^{\prime}\left(  \Omega_{\nu}\left(
r_{j}\right)  -\Omega_{\nu}\left(  r_{j-1}\right)  \right)  \lambda_{j,\nu
}=\sum_{j=1}^{k}\lambda_{j,\nu}^{\prime}\Omega_{\nu}\lambda_{j,\nu}\left(
r_{j}-r_{j-1}\right)  .
\]

\subparagraph{(b) Tightness}

The second step in establishing the functional CLT involves proving tightness
of the sequence $\lambda^{\prime}X_{n\tau}\left(  r\right)  .$ By Lemma A.3 of
Phillips and Durlauf (1986) and Proposition 4.1 of Wooldridge and White
(1988), see also Billingsley (1968, p.41), it is enough to establish tightness
componentwise. This is implied by establishing tightness for $\lambda^{\prime
}X_{n\tau}\left(  r\right)  $ for all $\lambda\in\mathbb{R}^{d}$ such that
$\lambda^{\prime}\lambda=1$. In the following we make use of Theorem 8.3 in
Billingsley (1968), see Appendix \ref{Aux_Results} for a statement of Theorem
8.3. The fact that tightness in our case can be established using Criterion
(8.5) in Billingsley (1968, Theorem 8.3) follows from Billingsley (1968,
Theorem 15.5) and the proof of Billingsley (1968, Theorem 8.3).

Recall the definition\textbf{ }%
\[
X_{n\tau,y}\left(  r\right)  =\frac{1}{\sqrt{n}}\sum_{t=1}^{T}\sum_{i=1}%
^{n}\psi_{n,it}^{y},\text{ }X_{n\tau,\nu}\left(  r\right)  =\frac{1}%
{\sqrt{\tau}}\sum_{t=\tau_{0}+1}^{\tau_{0}+\left[  \tau r\right]  }\psi
_{\tau,t}^{\nu}.
\]
and note that\textbf{ }%
\[
\left\vert \lambda^{\prime}\left(  X_{n\tau}\left(  s\right)  -X_{n\tau
}\left(  t\right)  \right)  \right\vert \leq\left\vert \lambda_{y}^{\prime
}\left(  X_{n\tau,y}\left(  s\right)  -X_{n\tau,y}\left(  t\right)  \right)
\right\vert +\left\vert \lambda_{\nu}^{\prime}\left(  X_{n\tau,\nu}\left(
s\right)  -X_{n\tau,\nu}\left(  t\right)  \right)  \right\vert
\]
where $\left\vert \lambda_{y}^{\prime}\left(  X_{n\tau,y}\left(  s\right)
-X_{n\tau,y}\left(  t\right)  \right)  \right\vert =0$ uniformly in
$t,s\in\left[  0,1\right]  $ because of the initial condition $X_{n\tau
}\left(  0\right)  $ given in (\ref{CLT_Process}) and the fact that
$X_{n\tau,y}\left(  t\right)  $ is constant as a function of $t.$ Thus, to
show tightness, we only need to consider $\lambda_{\nu}^{\prime}X_{n\tau,\nu
}\left(  t\right)  $.

Billingsley (1968, p. 58-59) constructs a continuous approximation to
$\lambda^{\prime}X_{n\tau}\left(  r\right)  .$ We denote it as $\lambda_{\nu
}^{\prime}X_{n\tau,\nu}^{c}\left(  r\right)  $ and define it analogously to
Billingsley (1968, eq 8.15) as\textbf{ }%
\[
\lambda_{\nu}^{\prime}X_{n\tau,\nu}^{c}\left(  r\right)  =\frac{1}{\sqrt{\tau
}}\sum_{t=\tau_{0}+1}^{\tau_{0}+\left[  \tau r\right]  }\lambda_{\nu}^{\prime
}\psi_{\tau,t}^{\nu}+\frac{\left(  \tau r-\left[  \tau r\right]  \right)
}{\sqrt{\tau}}\lambda_{\nu}^{\prime}\psi_{\tau,\left[  \tau r\right]  +1}%
^{\nu}.
\]
First, note that for $\varepsilon>0,$ $\sup_{\tau,r}\left\vert \tau r-\left[
\tau r\right]  \right\vert \leq1$ such that%
\[
P\left(  \max_{r\in\left[  0,1\right]  }\left\vert \frac{\left(  \tau
r-\left[  \tau r\right]  \right)  }{\sqrt{\tau}}\lambda_{\nu}^{\prime}%
\psi_{\tau,\left[  \tau r\right]  +1}^{\nu}\right\vert >\varepsilon\right)
\leq\frac{\sup_{t}E\left[  \left\Vert \psi_{\tau,t}^{\nu}\right\Vert \right]
}{\varepsilon\sqrt{\tau}}\rightarrow0
\]
and consequently, $X_{n\tau,\nu}^{c}\left(  r\right)  =X_{n\tau,\nu}\left(
r\right)  +o_{p}\left(  1\right)  $ uniformly in $r\in\left[  0,1\right]  $.
To establish tightness we need to establish that the `modulus of
continuity'\textbf{ }%
\begin{equation}
\omega\left(  X_{n\tau}^{c},\delta\right)  =\sup_{\left\vert t-s\right\vert
<\delta}\left\vert \lambda^{\prime}\left(  X_{n\tau}^{c}\left(  s\right)
-X_{n\tau}^{c}\left(  t\right)  \right)  \right\vert \label{DCLT30}%
\end{equation}
where $t,s\in\left[  0,1\right]  $ satisfies
\[
\lim_{\delta\rightarrow0}\underset{n,\tau}{\lim\sup}P\left(  \omega\left(
X_{n\tau}^{c},\delta\right)  \geq\varepsilon\right)  =0.
\]
Let $S_{\tau,k}=\frac{1}{\sqrt{\tau}}\sum_{t=\tau_{0}+1}^{\tau_{0}+k}%
\lambda_{\nu}^{\prime}\psi_{\tau,t}^{\nu}.$ By the inequalities in Billingsley
(1968, p.59)\ it follows that for $k$ such that $k/\tau<t<\left(  k+1\right)
/\tau$\textbf{ }%
\[
\sup_{t\leq s\leq t+\delta/2}\left\vert \lambda^{\prime}\left(  X_{n\tau}%
^{c}\left(  s\right)  -X_{n\tau}^{c}\left(  t\right)  \right)  \right\vert
\leq2\max_{0\leq i\leq\tau}\left\vert S_{\tau,k+i}-S_{\tau,k}\right\vert .
\]
By Billingsley (1968, Theorem 8.4) and the comments in Billingsley (1968, p.
59) to establish tightness it is enough to show that, translated to our
notation, for each positive $\epsilon$ there exists a positive $c>1$ and an
integer $\tau_{0}$ such that for $\tau\geq\tau_{0}$ and for all $k$ it follows
that%
\begin{equation}
P\left(  \max_{s\leq\tau}\left\vert S_{\tau,k+s}-S_{\tau,k}\right\vert
>c\right)  \leq\frac{\epsilon}{c^{2}}. \label{DCLT40-a}%
\end{equation}
We note that we normalized the scaling factor $\sigma=1$ relative to the
expression in Billingsley (1968), see also Billingsley (1968, p.58). Using
properties of $\lim\sup$ and $\lim,$ Condition \ref{DCLT40-a} is implied by
Condition \ref{DCLT40} which states that \textbf{ }%
\begin{equation}
\lim_{c\rightarrow\infty}\limsup_{\tau\rightarrow\infty}c^{2}P\left(
\max_{s\leq\tau}\left\vert S_{\tau,k+s}-S_{\tau,k}\right\vert >c\right)  =0
\label{DCLT40}%
\end{equation}
holds for all $k\in\mathbb{N}$. A proof of (\ref{DCLT40}) is given in Section
\ref{DCLT40-proof}.

\subparagraph{(c) Characterization of the limit distribution}

We now identify the limiting distribution using the technique of Rootzen
(1983). Tightness together with finite dimensional convergence in distribution
in (\ref{Fidi_CharFun_Convergence}), Condition \ref{Omega_z} and the fact that
the partition $r_{1},...,r_{k}$ is arbitrary implies that for $\lambda
\in\mathbb{R}^{d}$ with $\lambda=\left(  \lambda_{y}^{\prime},\lambda_{\nu
}^{\prime}\right)  ^{\prime}$
\begin{equation}
E\left[  \exp\left(  i\lambda^{\prime}X_{n\tau}\left(  r\right)  \right)
\right]  \rightarrow E\left[  \exp\left(  -\frac{1}{2}\left(  \lambda
_{y}^{\prime}\Omega_{y}\lambda_{y}+\lambda_{\nu}^{\prime}\Omega_{\nu}\left(
r\right)  \lambda_{\nu}\right)  \right)  \right]  \label{DCLT41}%
\end{equation}
with $\Omega_{y}=\sum_{t\in\left\{  1,..,T\right\}  }\Omega_{yt}.$ The final
step of the argument consists in representing the limiting process in terms of
stochastic integrals over isonormal Gaussian processes.\footnote{We are
grateful to an anonymous referee for suggesting a simplified method of proof
for this step.} By the law of iterated expectations and the fact that by
Assumptions (\ref{Omega_z}) and (\ref{Omega_y}) the matrices $\Omega_{y}$ and
$\Omega_{\nu}\left(  r\right)  $ are $\mathcal{C}$-measurable it follows that
\begin{align}
&  E\left[  \exp\left(  -\frac{1}{2}\left(  \lambda_{y}^{\prime}\Omega
_{y}\lambda_{y}+\lambda_{\nu}^{\prime}\Omega_{\nu}\left(  r\right)
\lambda_{\nu}\right)  \right)  \right] \label{DCLT41a}\\
&  =E\left[  E\left[  \left.  \exp\left(  -\frac{1}{2}\lambda_{y}^{\prime
}\Omega_{y}\lambda_{y}\right)  \right\vert \mathcal{C}\right]  E\left[
\left.  \exp\left(  -\frac{1}{2}\lambda_{\nu}^{\prime}\Omega_{\nu}\left(
r\right)  \lambda_{\nu}\right)  \right\vert \mathcal{C}\right]  \right]
.\nonumber
\end{align}
Let $W\left(  r\right)  =\left(  W_{y}\left(  r\right)  ,W_{\nu}\left(
r\right)  \right)  $ be a vector of mutually independent standard Brownian
motion processes in $\mathbb{R}^{d},$ independent of any $\mathcal{C}%
$-measurable random variable. We note that the first term on the RHS of
(\ref{DCLT41a}) satisfies%
\begin{equation}
E\left[  \left.  \exp\left(  -\frac{1}{2}\lambda_{y}^{\prime}\Omega_{y}%
\lambda_{y}\right)  \right\vert \mathcal{C}\right]  =E\left[  \left.
\exp\left(  i\lambda_{y}^{\prime}\Omega_{y}^{1/2}W_{y}\left(  1\right)
\right)  \right\vert \mathcal{C}\right]  \label{DCLT41b}%
\end{equation}
by the properties of the standard Gaussian characteristic function. To analyze
the second conditional expectation $E\left[  \exp\left(  -\frac{1}{2}%
\lambda_{\nu}^{\prime}\Omega_{\nu}\left(  r\right)  \lambda_{\nu}\right)
|\mathcal{C}\right]  $ note that by Kallenberg (1997, p.210) it follows from
the isometry of the stochastic integral that there exists a standard Brownian
process $W_{\nu}\left(  r\right)  $ such that
\[
E\left[  \left.  \left(  \int_{0}^{r}\lambda_{\nu}^{\prime}\left(  \dot
{\Omega}_{\nu}\left(  t\right)  \right)  ^{1/2}dW_{\nu}\left(  t\right)
\right)  ^{2}\right\vert \mathcal{C}\right]  =\int_{0}^{r}\lambda_{\nu
}^{\prime}\left(  \dot{\Omega}_{\nu}\left(  t\right)  \right)  \lambda_{\nu
}dt=\lambda_{\nu}^{\prime}\Omega_{\nu}\left(  r\right)  \lambda_{\nu}.
\]
By linearity of the stochastic integral, conditional on $\mathcal{C}$,
$\int_{0}^{r}\lambda_{\nu}^{\prime}\left(  \dot{\Omega}_{\nu}\left(  t\right)
\right)  ^{1/2}dW_{\nu}\left(  t\right)  $ is a centered Gaussian process with
conditional (on $\mathcal{C}$) characteristic function
\begin{equation}
E\left[  \left.  \exp\left(  -\frac{1}{2}\lambda_{\nu}^{\prime}\Omega_{\nu
}\left(  r\right)  \lambda_{\nu}\right)  \right\vert \mathcal{C}\right]
=E\left[  \left.  \exp\left(  i\int_{0}^{r}\lambda_{\nu}^{\prime}\left(
\dot{\Omega}_{\nu}\left(  t\right)  \right)  ^{1/2}dW_{\nu}\left(  t\right)
\right)  \right\vert \mathcal{C}\right]  . \label{DCLT41c}%
\end{equation}
\textbf{ }Combining (\ref{DCLT41a}), (\ref{DCLT41b}) and (\ref{DCLT41c})
gives
\begin{align}
&  E\left[  \exp\left(  -\frac{1}{2}\left(  \lambda_{y}^{\prime}\Omega
_{y}\lambda_{y}+\lambda_{\nu}^{\prime}\Omega_{\nu}\left(  r\right)
\lambda_{\nu}\right)  \right)  \right] \label{DCLT41d}\\
&  =E\left[  E\left[  \left.  \exp\left(  i\lambda_{y}^{\prime}\Omega
_{y}^{1/2}W_{y}\left(  1\right)  \right)  \right\vert \mathcal{C}\right]
E\left[  \left.  \exp\left(  i\int_{0}^{r}\lambda_{\nu}^{\prime}\left(
\dot{\Omega}_{\nu}\left(  t\right)  \right)  ^{1/2}dW_{\nu}\left(  t\right)
\right)  \right\vert \mathcal{C}\right]  \right] \nonumber
\end{align}
Since by construction, $\left(  W_{y}\left(  r\right)  ,W_{\nu}\left(
r\right)  \right)  $ are mutually independent conditional on $\mathcal{C}$ it
follows that the RHS of (\ref{DCLT41d}) can be written as
\begin{align*}
&  E\left[  E\left[  \left.  \exp\left(  i\lambda_{y}^{\prime}\Omega_{y}%
^{1/2}W_{y}\left(  1\right)  \right)  \right\vert \mathcal{C}\right]  E\left[
\left.  \exp\left(  i\int_{0}^{r}\lambda_{\nu}^{\prime}\left(  \dot{\Omega
}_{\nu}\left(  t\right)  \right)  ^{1/2}dW_{\nu}\left(  t\right)  \right)
\right\vert \mathcal{C}\right]  \right] \\
&  =E\left[  E\left[  \left.  \exp\left(  i\lambda_{y}^{\prime}\Omega
_{y}^{1/2}W_{y}\left(  1\right)  +i\int_{0}^{r}\lambda_{\nu}^{\prime}\left(
\dot{\Omega}_{\nu}\left(  t\right)  \right)  ^{1/2}dW_{\nu}\left(  t\right)
\right)  \right\vert \mathcal{C}\right]  \right] \\
&  =E\left[  \exp\left(  i\lambda_{y}^{\prime}\Omega_{y}^{1/2}W_{y}\left(
1\right)  +i\int_{0}^{r}\lambda_{\nu}^{\prime}\left(  \dot{\Omega}_{\nu
}\left(  t\right)  \right)  ^{1/2}dW_{\nu}\left(  t\right)  \right)  \right]
,
\end{align*}
where the last equality follows from the law of iterated expectations.

\subsection{Proof of (\ref{KP1})\label{sec-proof-of-KP1}}

In this section we show that $\sum_{q=1}^{k_{n}}E\left[  \left\vert \ddot
{\psi}_{q}\right\vert ^{2+\delta}\right]  \rightarrow0$ where $\ddot{\psi}%
_{q}$ is defined in (\ref{D_FCLT2}). Note that, for any fixed $n$ and given
$q$, and thus for a corresponding unique vector $\left(  t,i\right)  $, there
exists a unique $j\in\left\{  1,\ldots,k\right\}  $ such that $\tau
_{0}+\left[  \tau r_{j-1}\right]  <t\leq\tau_{0}+\left[  \tau r_{j}\right]  .$
Then,
\begin{align*}
\ddot{\psi}_{q\left(  i,t\right)  }  &  =\sum_{l=1}^{k}\lambda_{l}^{\prime
}\left(  \Delta\tilde{\psi}_{it}\left(  r_{l}\right)  -E\left[  \left.
\Delta\tilde{\psi}_{it}\left(  r_{l}\right)  \right\vert \mathcal{G}_{\tau
n,t^{\ast}n+i-1}\right]  \right) \\
&  =\lambda_{1,y}^{\prime}\left(  \tilde{\psi}_{it}^{y}-E\left[  \left.
\tilde{\psi}_{it}^{y}\right\vert \mathcal{G}_{\tau n,t^{\ast}n+i-1}\right]
\right)  1\left\{  j=1\right\} \\
&  +\lambda_{j,\nu}^{\prime}\left(  \tilde{\psi}_{\tau,t}^{\nu}\left(
r_{j}\right)  -E\left[  \left.  \tilde{\psi}_{\tau,t}^{\nu}\left(
r_{j}\right)  \right\vert \mathcal{G}_{\tau n,t^{\ast}n+i-1}\right]  \right)
1\left\{  \left[  \tau r_{j-1}\right]  <t-\tau_{0}\leq\left[  \tau
r_{j}\right]  \right\}  1\left\{  \text{ }i=1\right\}  ,
\end{align*}
where all remaining terms in the sum are zero because of by (\ref{D_FCLT2}),
(\ref{D_FCLT4}) and (\ref{D_FCLT5}). For the subsequent inequalities, fix
$q\in\left\{  1,...,k_{n}\right\}  $ (and the corresponding $\left(
t,i\right)  $ and $j$) arbitrarily. Introduce the shorthand notation
$1_{j}=1\left\{  j=1\right\}  $ and $1_{ij}=1\left\{  \left[  \tau
r_{j-1}\right]  <t-\tau_{0}\leq\left[  \tau r_{j}\right]  \right\}  1\left\{
\text{ }i=1\right\}  $.

First, note that for $\delta\geq0$, and by Jensen's inequality applied to the
empirical measure $\frac{1}{4}\sum_{i=1}^{4}x_{i}$ we have that%
\begin{align*}
&  \left\vert \ddot{\psi}_{q}\right\vert ^{2+\delta}\\
&  =4^{2+\delta}\left\vert \frac{1}{4}\lambda_{1,y}^{\prime}\left(
\tilde{\psi}_{it}^{y}-E\left[  \tilde{\psi}_{it}^{y}|\mathcal{G}_{\tau
n,t^{\ast}n+i-1}\right]  \right)  1_{j}+\frac{1}{4}\lambda_{j,\nu}^{\prime
}\left(  \tilde{\psi}_{\tau,t}^{\nu}\left(  r_{j}\right)  -E\left[
\tilde{\psi}_{\tau,t}^{\nu}\left(  r_{j}\right)  |\mathcal{G}_{\tau n,t^{\ast
}n+i-1}\right]  \right)  1_{ij}\right\vert ^{2+\delta}\\
&  \leq4^{2+\delta}\left(  \frac{1}{4}\left\Vert \lambda_{1,y}\right\Vert
^{2+\delta}\left\Vert \tilde{\psi}_{it}^{y}\right\Vert ^{2+\delta}+\frac{1}%
{4}\left\Vert \lambda_{1,y}\right\Vert ^{2+\delta}\left\Vert E\left[
\tilde{\psi}_{it}^{y}|\mathcal{G}_{\tau n,t^{\ast}n+i-1}\right]  \right\Vert
^{2+\delta}\right)  1_{j}\\
&  +4^{2+\delta}\left(  \frac{1}{4}\left\Vert \lambda_{j,\nu}\right\Vert
^{2+\delta}\left\Vert \tilde{\psi}_{\tau,t}^{\nu}\left(  r_{j}\right)
\right\Vert ^{2+\delta}+\frac{1}{4}\left\Vert \lambda_{j,\nu}\right\Vert
^{2+\delta}\left\Vert E\left[  \tilde{\psi}_{\tau,t}^{\nu}\left(
r_{j}\right)  |\mathcal{G}_{\tau n,t^{\ast}n+i-1}\right]  \right\Vert
^{2+\delta}\right)  1_{ij}\\
&  =2^{2+2\delta}\left(  \left\Vert \lambda_{1,y}\right\Vert ^{2+\delta
}\left\Vert \tilde{\psi}_{it}^{y}\right\Vert ^{2+\delta}+\left\Vert
\lambda_{1,y}\right\Vert ^{2+\delta}\left\Vert E\left[  \tilde{\psi}_{it}%
^{y}|\mathcal{G}_{\tau n,t^{\ast}n+i-1}\right]  \right\Vert ^{2+\delta
}\right)  1_{j}\\
&  +2^{2+2\delta}\left(  \left\Vert \lambda_{j,\nu}\right\Vert ^{2+\delta
}\left\Vert \tilde{\psi}_{\tau,t}^{\nu}\left(  r_{j}\right)  \right\Vert
^{2+\delta}+\left\Vert \lambda_{j,\nu}\right\Vert ^{2+\delta}\left\Vert
E\left[  \tilde{\psi}_{\tau,t}^{\nu}\left(  r_{j}\right)  |\mathcal{G}_{\tau
n,t^{\ast}n+i-1}\right]  \right\Vert ^{2+\delta}\right)  1_{ij}.
\end{align*}
We further use the definitions in (\ref{psi_tilde_z}) such that by Jensen's
inequality and for $i=1$ and $t\in\left[  \tau_{0}+1,\tau_{0}+\tau\right]  $
\begin{align*}
&  \left\Vert \tilde{\psi}_{\tau,t}^{\nu}\left(  r_{j}\right)  \right\Vert
^{2+\delta}+\left\Vert E\left[  \tilde{\psi}_{\tau,t}^{\nu}\left(
r_{j}\right)  |\mathcal{G}_{\tau n,t^{\ast}n+i-1}\right]  \right\Vert
^{2+\delta}\\
&  \leq\frac{1}{\tau^{1+\delta/2}}\left(  \left\Vert \psi_{\tau,t}^{\nu
}\right\Vert ^{2+\delta}+\left(  E\left[  \left\Vert \psi_{\tau,t}^{\nu
}\right\Vert |\mathcal{G}_{\tau n,t^{\ast}n+i-1}\right]  \right)  ^{2+\delta
}\right) \\
&  \leq\frac{1}{\tau^{1+\delta/2}}\left(  \left\Vert \psi_{\tau,t}^{\nu
}\right\Vert ^{2+\delta}+E\left[  \left\Vert \psi_{\tau,t}^{\nu}\right\Vert
^{2+\delta}|\mathcal{G}_{\tau n,t^{\ast}n+i-1}\right]  \right)
\end{align*}
while for $i>1$ or $t\notin\left[  \tau_{0}+1,\tau_{0}+\tau\right]  ,$
\[
\left\Vert \tilde{\psi}_{it}^{\nu}\right\Vert =0.
\]
Similarly, for $t\in\left[  1,...,T\right]  $
\begin{align*}
&  \left\Vert \tilde{\psi}_{it}^{y}\right\Vert ^{2+\delta}+\left\Vert E\left[
\tilde{\psi}_{it}^{y}|\mathcal{G}_{\tau n,t^{\ast}n+i-1}\right]  \right\Vert
^{2+\delta}\\
&  \leq\frac{1}{n^{1+\delta/2}}\left(  \left\Vert \psi_{it}^{y}\right\Vert
^{2+\delta}+E\left[  \left\Vert \psi_{it}^{y}\right\Vert ^{2+\delta
}|\mathcal{G}_{\tau n,t^{\ast}n+i-1}\right]  \right)
\end{align*}
while for $t\notin\left[  1,...,T\right]  $
\[
\left\Vert \tilde{\psi}_{it}^{y}\right\Vert =0.
\]
Noting that $\left\Vert \lambda_{j,y}\right\Vert \leq1$ and $\left\Vert
\lambda_{j,\nu}\right\Vert <1$,%
\begin{align}
E\left[  \left.  \left\vert \ddot{\psi}_{q}\right\vert ^{2+\delta}\right\vert
\mathcal{G}_{\tau n,q-1}\right]   &  \leq\frac{2^{3+2\delta}1\left\{
i=1,t\in\left[  \tau_{0}+1,\tau_{0}+\tau\right]  \right\}  }{\tau^{1+\delta
/2}}E\left[  \left.  \left\Vert \psi_{\tau,t}^{\nu}\right\Vert ^{2+\delta
}\right\vert \mathcal{G}_{\tau n,t^{\ast}n+i-1}\right] \nonumber\\
&  +\frac{2^{3+2\delta}1\left\{  t\in\left[  1,...,T\right]  \right\}
}{n^{1+\delta/2}}E\left[  \left.  \left\Vert \psi_{it}^{y}\right\Vert
^{2+\delta}\right\vert \mathcal{G}_{\tau n,t^{\ast}n+i-1}\right]  ,
\label{DCLT_D11}%
\end{align}
where the inequality in (\ref{DCLT_D11}) holds for $\delta\geq0$.

To show that (\ref{KP1}) holds note that from (\ref{DCLT_D11}), the law of
iterated expectations and Condition \ref{Diag_CLT_Cond} it follows that for
some constant $C<\infty,$
\begin{align*}
\sum_{q=1}^{k_{n}}E\left[  \left\vert \ddot{\psi}_{q}\right\vert ^{2+\delta
}\right]   &  =\sum_{q=1}^{k_{n}}E\left[  E\left[  \left\vert \ddot{\psi}%
_{q}\right\vert ^{2+\delta}|\mathcal{G}_{\tau n,q-1}\right]  \right] \\
&  \leq\frac{2^{3+2\delta}}{\tau^{1+\delta/2}}\sum_{t=\tau_{0}+1}^{\tau
_{0}+\tau}E\left[  \left\Vert \psi_{\tau,t}^{\nu}\right\Vert ^{2+\delta
}\right] \\
&  +\frac{2^{3+2\delta}}{n^{1+\delta/2}}\sum_{t=1}^{T}\sum_{i=1}^{n}E\left[
\left\Vert \psi_{it}^{y}\right\Vert ^{2+\delta}\right] \\
&  \leq\frac{2^{3+2\delta}\tau C}{\tau^{1+\delta/2}}+\frac{2^{3+2\delta}%
nTC}{n^{1+\delta/2}}=\frac{2^{3+2\delta}C}{\tau^{\delta/2}}+\frac
{2^{3+2\delta}TC}{n^{\delta/2}}\rightarrow0
\end{align*}
because $2^{3+2\delta}C$ and $T$ are fixed as $\tau,n\rightarrow\infty$.

\subsection{Proof of (\ref{KP2})\label{sec-proof-of-KP2}}

Consider the probability limit of $\sum_{q=1}^{k_{n}}\ddot{\psi}_{q}^{2}.$ We
have
\begin{align}
&  \sum_{q=1}^{k_{n}}\ddot{\psi}_{q}^{2}\nonumber\\
&  =\frac{1}{\tau}\sum_{j=1}^{k}\sum_{t=\tau_{0}+1}^{\tau_{0}+\tau}\left(
\lambda_{j,\nu}^{\prime}\left(  \psi_{\tau,t}^{\nu}-E\left[  \psi_{\tau
,t}^{\nu}|\mathcal{G}_{\tau n,t^{\ast}n}\right]  \right)  \right)  ^{2}%
1_{ij}\label{DCLT_D13a}\\
&  +\frac{2}{\sqrt{\tau n}}\sum_{\substack{t\in\left(  \tau_{0}+1,...,\tau
_{0}+\tau\right)  \\\cap\left\{  1,..,T\right\}  }}\sum_{j=1}^{k}%
\lambda_{j,\nu}^{\prime}\left(  \psi_{\tau,t}^{\nu}-E\left[  \psi_{\tau
,t}^{\nu}|\mathcal{G}_{\tau n,s^{\ast}n}\right]  \right)  \left(  \psi
_{1t}^{y}-E\left[  \psi_{1t}^{y}|\mathcal{G}_{\tau n,t^{\ast}n}\right]
\right)  ^{\prime}\lambda_{j,y}1_{1j}1_{j}\label{DCLT_D13b}\\
&  +\frac{1}{n}\left(  \sum_{t\in\left\{  1,..,T\right\}  }\sum_{i=1}%
^{n}\left(  \lambda_{1,y}^{\prime}\left(  \psi_{n,it}^{y}-E\left[  \psi
_{n,it}^{y}|\mathcal{G}_{\tau n,t^{\ast}n+i-1}\right]  \right)  \right)
\right)  ^{2}. \label{DCLT_D13c}%
\end{align}
Note that\textbf{ }%
\begin{align*}
\left(  \lambda_{j,\nu}^{\prime}\left(  \psi_{\tau,t}^{\nu}-E\left[
\psi_{\tau,t}^{\nu}|\mathcal{G}_{\tau n,t^{\ast}n}\right]  \right)  \right)
^{2}  &  \leq\left(  \lambda_{j,\nu}^{\prime}\psi_{\tau,t}^{\nu}\right)
^{2}+2\left\Vert \psi_{\tau,t}^{\nu}\right\Vert \left\Vert \lambda_{j,\nu
}\right\Vert ^{2}\left\Vert E\left[  \psi_{\tau,t}^{\nu}|\mathcal{G}_{\tau
n,t^{\ast}n}\right]  \right\Vert \\
&  +\left\Vert \lambda_{j,\nu}\right\Vert ^{2}\left\Vert E\left[  \psi
_{\tau,t}^{\nu}|\mathcal{G}_{\tau n,t^{\ast}n}\right]  \right\Vert .
\end{align*}
Also note that $E\left[  \psi_{\tau,t}^{\nu}|\mathcal{G}_{\tau n,t^{\ast}%
n}\right]  =0$ when $t>T$ implies that
\begin{align}
&  \frac{1}{\tau}\sum_{j=1}^{k}\sum_{t=\tau_{0}+1}^{\tau_{0}+\tau}\left\vert
\left(  \left(  \lambda_{j,\nu}^{\prime}\left(  \psi_{\tau,t}^{\nu}-E\left[
\psi_{\tau,t}^{\nu}|\mathcal{G}_{\tau n,t^{\ast}n}\right]  \right)  \right)
^{2}-\left(  \lambda_{j,\nu}^{\prime}\psi_{\tau,t}^{\nu}\right)  ^{2}\right)
1_{ij}\right\vert \label{TS_Variance}\\
&  \leq\frac{1}{\tau}\sum_{j=1}^{k}\sum_{t\in\left\{  \tau_{0}+1,...,\tau
_{0}+\tau\right\}  \cap\left\{  1,..,T\right\}  }\left(  2\left\Vert
\psi_{\tau,t}^{\nu}\right\Vert \left\Vert \lambda_{j,\nu}\right\Vert
^{2}\left\Vert E\left[  \psi_{\tau,t}^{\nu}|\mathcal{G}_{\tau n,t^{\ast}%
n}\right]  \right\Vert +\left\Vert \lambda_{j,\nu}\right\Vert ^{2}\left\Vert
E\left[  \psi_{\tau,t}^{\nu}|\mathcal{G}_{\tau n,t^{\ast}n}\right]
\right\Vert ^{2}\right)  .\nonumber
\end{align}
Because Condition \ref{Omega_z} implies that%
\[
\frac{1}{\tau}\sum_{j=1}^{k}\sum_{t=t=\tau_{0}+1}^{\tau_{0}+\tau}\left(
\lambda_{j,\nu}^{\prime}\psi_{\tau,t}^{\nu}\right)  ^{2}1\left\{  \tau
_{0}+\left[  \tau r_{j-1}\right]  <t\leq\tau_{0}+\left[  \tau r_{j}\right]
\right\}  \overset{p}{\rightarrow}\sum_{j=1}^{k}\lambda_{j,\nu}^{\prime
}\left(  \Omega_{\nu}\left(  r_{j}\right)  -\Omega_{\nu}\left(  r_{j-1}%
\right)  \right)  \lambda_{j,\nu},
\]
the term (\ref{DCLT_D13a}) is equal to
\[
\sum_{j=1}^{k}\lambda_{j,\nu}^{\prime}\left(  \Omega_{\nu}\left(
r_{j}\right)  -\Omega_{\nu}\left(  r_{j-1}\right)  \right)  \lambda_{j,\nu
}+o_{p}\left(  1\right)
\]
if the RHS of (\ref{TS_Variance}) is $o_{p}\left(  1\right)  $. To show that
it is indeed the case, note that by the Markov inequality and the
Cauchy-Schwarz inequality it is enough to show that
\begin{align}
&  \frac{1}{\tau}\sum_{j=1}^{k}\sum_{\substack{t\in\left\{  \tau
_{0}+1,...,\tau_{0}+\tau\right\}  \\\cap\left\{  1,..,T\right\}  }}\left\Vert
\lambda_{j,\nu}\right\Vert ^{2}\left(  2\sqrt{E\left[  \left\Vert \psi
_{\tau,t}^{\nu}\right\Vert ^{2}\right]  E\left[  \left\Vert E\left[
\psi_{\tau,t}^{\nu}|\mathcal{G}_{\tau n,t^{\ast}n}\right]  \right\Vert
^{2}\right]  }+E\left[  \left\Vert E\left[  \psi_{\tau,t}^{\nu}|\mathcal{G}%
_{\tau n,t^{\ast}n}\right]  \right\Vert ^{2}\right]  \right)  \label{DCLT_D14}%
\\
\leq &  \frac{1}{\tau}\sum_{j=1}^{k}\left\Vert \lambda_{j,\nu}\right\Vert
^{2}\left(  \sum_{t=0}^{T}3E\left[  \left\Vert \psi_{\tau,t}^{\nu}\right\Vert
^{2}\right]  +\sum_{t=\tau_{0}}^{-1}\left(  2\sqrt{E\left[  \left\Vert
\psi_{\tau,t}^{\nu}\right\Vert ^{2}\right]  }\vartheta_{t}+\vartheta_{t}%
^{2}\right)  \right) \nonumber\\
\leq &  O\left(  \tau^{-1}\right)  +\frac{C}{\tau}\sum_{t=\tau_{0}}%
^{-1}\left(  2\sqrt{E\left[  \left\Vert \psi_{\tau,t}^{\nu}\right\Vert
^{2}\right]  }\left(  \left\vert t\right\vert ^{1+\delta}\right)
^{-1/2}+C\left(  \left\vert t\right\vert ^{1+\delta}\right)  ^{-1}\right)
\nonumber\\
\leq &  \frac{2C\sup_{t}\sqrt{E\left[  \left\Vert \psi_{\tau,t}^{\nu
}\right\Vert ^{2}\right]  }}{\tau}\sum_{t=\tau_{0}}^{-1}\left(  \left\vert
t\right\vert ^{1+\delta}\right)  ^{-1/2}\nonumber\\
&  +o\left(  \tau^{-1}\sum_{t=\tau_{0}}^{-1}\left(  \left\vert t\right\vert
^{1+\delta}\right)  ^{-1/2}\right)  +O\left(  \tau^{-1}\right)  \rightarrow
0,\nonumber
\end{align}
where the first inequality follows from Condition \ref{Diag_CLT_Cond}(vii).
The second inequality follows from the fact that $T$ is fixed and bounded and
Condition \ref{Diag_CLT_Cond}(vii). The final result uses that $\sup
_{t}E\left[  \left\Vert \psi_{\tau,t}^{\nu}\right\Vert ^{2+\delta}\right]
\leq C<\infty$ by Condition \ref{Diag_CLT_Cond}(iv) and the fact that
$\left\vert \tau_{0}\right\vert \leq\tau,$ $t/\tau\leq1$ for\textbf{ }%
$t\in\left[  1,...,\tau\right]  $\textbf{,} such that
\begin{align*}
\tau^{-1}\sum_{t=\tau_{0}}^{-1}\left(  \left\vert t\right\vert ^{1+\delta
}\right)  ^{-1/2}  &  =\tau^{-1}\sum_{t=1}^{\left\vert \tau_{0}\right\vert
}\left(  t^{1+\delta}\right)  ^{-1/2}\\
&  \leq\tau^{-1}\sum_{t=1}^{\left\vert \tau_{0}\right\vert }\left(  \frac
{t}{\tau}\right)  ^{-1/2}t^{-\left(  1/2+\delta/2\right)  }\\
&  \leq\tau^{-1/2}\sum_{t=1}^{\infty}t^{-\left(  1+\delta/2\right)  }=O\left(
\tau^{-1/2}\right)  .
\end{align*}
The last equality above uses the fact $\sum_{t=1}^{\infty}t^{-\left(
1+\delta/2\right)  }<\infty$ for any $\delta>0$.

Next we show that (\ref{DCLT_D13b}) is $o_{p}\left(  1\right)  $. For this
purpose, we note
\begin{align}
&  E\left[  \left\vert \frac{2}{\sqrt{\tau n}}\sum_{\substack{t\in\left(
\tau_{0}+1,...,\tau_{0}+\tau\right)  \\\cap\left\{  1,..,T\right\}  }%
}\sum_{j=1}^{k}\lambda_{j,\nu}^{\prime}\left(  \psi_{\tau,t}^{\nu}-E\left[
\psi_{\tau,t}^{\nu}|\mathcal{G}_{\tau n,s^{\ast}n}\right]  \right)  \left(
\psi_{1t}^{y}-E\left[  \psi_{1t}^{y}|\mathcal{G}_{\tau n,t^{\ast}n}\right]
\right)  ^{\prime}\lambda_{1,y}1_{1j}1_{1}\right\vert \right] \nonumber\\
&  \leq\frac{2}{\sqrt{\tau n}}\sum_{\substack{t\in\left\{  \tau_{0}%
+1,...,\tau_{0}+\tau\right\}  \\\cap\left\{  1,..,T\right\}  }}\left\{
\left(  E\left[  \left\vert \sum_{j=1}^{k}\lambda_{j\nu}^{\prime}\left(
\psi_{\tau,t}^{\nu}-E\left[  \psi_{\tau,t}^{\nu}|\mathcal{G}_{\tau n,t^{\ast
}n}\right]  \right)  \right\vert ^{2}\right]  \right)  ^{1/2}\right.
\nonumber\\
&  \left.  \times\left(  E\left[  \left\vert \left(  \psi_{1t}^{y}-E\left[
\psi_{1t}^{y}|\mathcal{G}_{\tau n,t^{\ast}n}\right]  \right)  ^{\prime}%
\lambda_{1y}\right\vert ^{2}\right]  \right)  ^{1/2}\right\}
\label{DCLT_D16a}\\
&  \leq\frac{2^{2}\sqrt{k}}{\sqrt{\tau n}}\sup_{t}\left(  E\left[  \left\Vert
\psi_{\tau,t}^{\nu}\right\Vert \right]  \right)  ^{1/2}\sum_{\substack{t\in
\left\{  \tau_{0},...,\tau_{0}+\tau\right\}  \\\cap\left\{  1,..,T\right\}
}}\left(  E\left[  \left\vert \left(  \psi_{1t}^{y}-E\left[  \psi_{1t}%
^{y}|\mathcal{G}_{\tau n,t^{\ast}n}\right]  \right)  ^{\prime}\lambda
_{1,y}\right\vert ^{2}\right]  \right)  ^{1/2}\label{DCLT_D16b}\\
&  \leq\frac{2^{3}\sqrt{k}T}{\sqrt{\tau n}}\sup_{t}\left(  E\left[  \left\Vert
\psi_{\tau,t}^{\nu}\right\Vert \right]  \right)  ^{1/2}\left(  \sup
_{i,t}E\left[  \left\Vert \psi_{n,it}^{y}\right\Vert ^{2}\right]  \right)
^{1/2}\rightarrow0 \label{DCLT_D16d}%
\end{align}
where the first inequality in (\ref{DCLT_D16a}) follows from the
Cauchy-Schwarz inequalities. Then we have in (\ref{DCLT_D16b}), by Condition
\ref{Diag_CLT_Cond}(iii) and the H\"{o}lder inequality that
\[
E\left[  \left\vert \left(  \psi_{1t}^{y}-E\left[  \psi_{1t}^{y}%
|\mathcal{G}_{\tau n,t^{\ast}n}\right]  \right)  ^{\prime}\lambda
_{y}\right\vert ^{2}\right]  \leq2E\left[  \left\Vert \psi_{1t}^{y}\right\Vert
^{2}\right]
\]
such that (\ref{DCLT_D16d}) follows. We note that (\ref{DCLT_D16d}) goes to
zero because of \textbf{ }Condition \ref{Diag_CLT_Cond}(iv)\textbf{ }as long
as $T/\sqrt{\tau n}\rightarrow0.$ Clearly, this condition holds as long as $T$
is held fixed, but holds under weaker conditions as well.

Next the limit of (\ref{DCLT_D13c}) is, by Condition \ref{Diag_CLT_Cond}(v)
and Condition \ref{Omega_y},%
\[
\frac{1}{n}\sum_{t\in\left\{  1,..,T\right\}  }\sum_{i=1}^{n}\left(
\lambda_{1,y}^{\prime}\left(  \psi_{n,it}^{y}-E\left[  \psi_{n,it}%
^{y}|\mathcal{G}_{\tau n,t^{\ast}n+i-1}\right]  \right)  \right)
^{2}\overset{p}{\rightarrow}\sum_{t\in\left\{  1,..,T\right\}  }\lambda
_{1,y}^{\prime}\Omega_{ty}\lambda_{1,y}.
\]
This verifies (\ref{KP2}).

\subsection{Proof of (\ref{KP3})\label{sec-proof-of-KP3}}

For (\ref{KP3}) we check that
\begin{equation}
\sup_{n}E\left[  \left(  \sum_{q=1}^{k_{n}}E\left[  \left.  \left\vert
\ddot{\psi}_{q}\right\vert ^{2}\right\vert \mathcal{G}_{\tau n,q-1}\right]
\right)  ^{1+\delta/2}\right]  <\infty. \label{DCLT_D20}%
\end{equation}
First, use (\ref{DCLT_D11}) with $\delta=0$ to obtain
\begin{align}
\sum_{q=1}^{k_{n}}E\left[  \left.  \left\vert \ddot{\psi}_{q}\right\vert
^{2}\right\vert \mathcal{G}_{\tau n,q-1}\right]   &  \leq\frac{2^{3}}{\tau
}\sum_{t=\tau_{0}}^{\tau_{0}+\tau}E\left[  \left.  \left\Vert \psi_{\tau
,t}^{\nu}\right\Vert ^{2}\right\vert \mathcal{G}_{\tau n,t^{\ast}n}\right]
\nonumber\\
&  +\frac{2^{3}}{n}\sum_{t\in\left\{  1,..,T\right\}  }\sum_{i=1}^{n}E\left[
\left.  \left\Vert \psi_{n,it}^{y}\right\Vert ^{2}\right\vert \mathcal{G}%
_{\tau n,t^{\ast}n+i-1}\right]  . \label{DCLT_D22}%
\end{align}
Applying (\ref{DCLT_D22}) to (\ref{DCLT_D20}) and using the H\"{o}lder
inequality\textbf{ }implies
\begin{align*}
&  E\left[  \left(  \sum_{q=1}^{k_{n}}E\left[  \left.  \left\vert \ddot{\psi
}_{q}\right\vert ^{2}\right\vert \mathcal{G}_{\tau n,q-1}\right]  \right)
^{1+\delta/2}\right] \\
&  \leq2^{\delta/2}E\left[  \left(  \frac{2^{3}}{\tau}\sum_{t=\tau_{0}%
+1}^{\tau_{0}+\tau}E\left[  \left.  \left\Vert \psi_{\tau,t}^{\nu}\right\Vert
^{2}\right\vert \mathcal{G}_{\tau n,t^{\ast}n+i-1}\right]  \right)
^{1+\delta/2}\right] \\
&  +2^{\delta/2}E\left[  \left(  \frac{2^{3}}{n}\sum_{t\in\left\{  \tau
_{0},...,\tau_{0}+\tau\right\}  \cap\left\{  1,..,T\right\}  }\sum_{i=1}%
^{n}E\left[  \left.  \left\Vert \psi_{n,it}^{y}\right\Vert ^{2}\right\vert
\mathcal{G}_{\tau n,t^{\ast}n+i-1}\right]  \right)  ^{1+\delta/2}\right]  .
\end{align*}
By Jensen's inequality, we have%
\[
\left(  \frac{1}{\tau}\sum_{t=\tau_{0}+1}^{\tau_{0}+\tau}E\left[  \left.
\left\Vert \psi_{\tau,t}^{\nu}\right\Vert ^{2}\right\vert \mathcal{G}_{\tau
n,t^{\ast}n+i-1}\right]  \right)  ^{1+\delta/2}\leq\frac{1}{\tau}\sum
_{t=\tau_{0}+1}^{\tau_{0}+\tau}E\left[  \left.  \left\Vert \psi_{\tau,t}^{\nu
}\right\Vert ^{2}\right\vert \mathcal{G}_{\tau n,t^{\ast}n+i-1}\right]
^{1+\delta/2}%
\]
and%
\[
E\left[  \left.  \left\Vert \psi_{\tau,t}^{\nu}\right\Vert ^{2}\right\vert
\mathcal{G}_{\tau n,t^{\ast}n+i-1}\right]  ^{1+\delta/2}\leq E\left[  \left.
\left\Vert \psi_{\tau,t}^{\nu}\right\Vert ^{2+\delta}\right\vert
\mathcal{G}_{\tau n,t^{\ast}n+i-1}\right]
\]
so that%
\begin{align}
E\left[  \left(  \frac{2^{3}}{\tau}\sum_{t=\tau_{0}+1}^{\tau_{0}+\tau}E\left[
\left.  \left\Vert \psi_{\tau,t}^{\nu}\right\Vert ^{2}\right\vert
\mathcal{G}_{\tau n,t^{\ast}n+i-1}\right]  \right)  ^{1+\delta/2}\right]   &
\leq\frac{2^{3+3\delta/2}}{\tau}\sum_{t=\tau_{0}+1}^{\tau_{0}+\tau}E\left[
E\left[  \left.  \left\Vert \psi_{\tau,t}^{\nu}\right\Vert ^{2+\delta
}\right\vert \mathcal{G}_{\tau n,t^{\ast}n+i-1}\right]  \right] \nonumber\\
&  \leq2^{3+3\delta/2}\sup_{t}E\left[  \left\Vert \psi_{\tau,t}^{\nu
}\right\Vert ^{2+\delta}\right]  <\infty. \label{DCLT_D23}%
\end{align}
and similarly, \textbf{ }
\begin{align}
&  E\left[  \left(  \frac{2^{3}}{n}\sum_{t\in\left\{  1,..,T\right\}  }%
\sum_{i=1}^{n}E\left[  \left\Vert \psi_{n,it}^{y}\right\Vert ^{2}%
|\mathcal{G}_{\tau n,t^{\ast}n+i-1}\right]  \right)  ^{1+\delta/2}\right]
\label{DCLT_D24}\\
&  \leq\frac{2^{3+3\delta/2}\left(  Tn\right)  ^{\delta/2}}{n^{1+\delta/2}%
}\sum_{t=1}^{T}\sum_{i=1}^{n}E\left[  \left\Vert \psi_{n,it}^{y}\right\Vert
^{2+\delta}\right] \nonumber\\
&  \leq2^{3+3\delta/2}T^{1+\delta/2}\sup_{i,t}E\left[  \left\Vert \psi
_{n,it}^{y}\right\Vert ^{2+\delta}\right]  <\infty\nonumber
\end{align}
By combining (\ref{DCLT_D23}) and (\ref{DCLT_D24}) we obtain the following
bound for (\ref{DCLT_D20}),%
\begin{align*}
&  E\left[  \left(  \sum_{q=1}^{k_{n}}E\left[  \left\vert \ddot{\psi}%
_{q}\right\vert ^{2+\delta}|\mathcal{G}_{\tau n,q-1}\right]  \right)
^{1+\delta/2}\right] \\
&  \leq2^{3+3\delta/2}\sup_{t}E\left[  \left\Vert \psi_{\tau,t}^{\nu
}\right\Vert ^{2+\delta}\right]  +2^{3+3\delta/2}T^{1+\delta/2}\sup
_{i,t}E\left[  \left\Vert \psi_{n,it}^{y}\right\Vert ^{2+\delta}\right]
<\infty.
\end{align*}

\subsection{Second Term in (\ref{D_FCLT5_1})\label{D_FCLT5_1-second}}

Consider
\begin{align*}
&  \sum_{t=\min(1,\tau_{0}+1)}^{\max(T,\tau_{0}+\tau)}\sum_{i=1}^{n}\sum
_{j=1}^{k}\lambda_{j}^{\prime}E\left[  \Delta\tilde{\psi}_{it}\left(
r_{j}\right)  |\mathcal{G}_{\tau n,t^{\ast}n+i-1}\right] \\
&  =\tau^{-1/2}\sum_{t=\tau_{0}+1}^{\tau_{0}+\tau}\sum_{j=1}^{k}\lambda
_{j,\nu}^{\prime}E\left[  \psi_{\tau,t}^{\nu}|\mathcal{G}_{\tau n,t^{\ast
}n+i-1}\right]  1\left\{  \tau_{0}+\left[  \tau r_{j-1}\right]  <t\leq\tau
_{0}+\left[  \tau r_{j}\right]  \right\} \\
&  +n^{-1/2}\sum_{t=1}^{T}\sum_{i=1}^{n}\lambda_{1,y}^{\prime}E\left[
\psi_{n,it}^{y}|\mathcal{G}_{\tau n,t^{\ast}n+i-1}\right]  ,
\end{align*}
where we defined\textbf{ }$r_{0}=0$. Note that
\[
E\left[  \psi_{\tau,t}^{\nu}|\mathcal{G}_{\tau n,t^{\ast}n+i-1}\right]
=0\text{ for }t>T
\]
and
\[
E\left[  \psi_{n,it}^{y}|\mathcal{G}_{\tau n,t^{\ast}n+i-1}\right]  =0.
\]
This implies, using the convention that a term is zero if it is a sum over
indices from $a$ to $b$ with $a>b$, as well as the fact that\textbf{ }%
$T\leq\tau_{0}+\tau,$ that
\begin{align*}
&  \tau^{-1/2}\sum_{t=\tau_{0}+1}^{\tau_{0}+\tau}\lambda_{\nu}^{\prime
}E\left[  \psi_{\tau,t}^{\nu}|\mathcal{G}_{\tau n,t^{\ast}n+i-1}\right] \\
&  =\tau^{-1/2}\sum_{t=\tau_{0}+1}^{T}\sum_{j=1}^{k}\lambda_{j,\nu}^{\prime
}E\left[  \psi_{\tau,t}^{\nu}|\mathcal{G}_{\tau n,t^{\ast}n+i-1}\right]
1\left\{  \tau_{0}+\left[  \tau r_{j-1}\right]  <t\leq\tau_{0}+\left[  \tau
r_{j}\right]  \right\}  .
\end{align*}
By a similar argument used to show that (\ref{DCLT_D14}) vanishes, and noting
that $T$ is fixed while $\tau\rightarrow\infty,$ it follows that, as long as
$\tau_{0}\leq T,$
\begin{align*}
&  E\left[  \left\vert \tau^{-1/2}\sum_{t=\tau_{0}+1}^{T}\lambda_{j,\nu
}^{\prime}E\left[  \psi_{\tau,t}^{\nu}|\mathcal{G}_{\tau n,t^{\ast}n}\right]
\right\vert ^{1+\delta/2}\right] \\
&  \leq\left(  \frac{1}{\tau}\right)  ^{1/2+\delta/4}\left(  T-\tau
_{0}\right)  ^{\delta/2}E\left[  \sum_{t=\tau_{0}+1}^{T}\left\vert
\lambda_{j,\nu}^{\prime}E\left[  \psi_{\tau,t}^{\nu}|\mathcal{G}_{\tau
n,t^{\ast}n}\right]  \right\vert ^{1+\delta/2}\right] \\
&  \leq\left(  \frac{1}{\tau}\right)  ^{1/2+\delta/4}\left(  T-\tau
_{0}\right)  ^{\delta/2}\sum_{t=\tau_{0}+1}^{T}E\left[  \left\Vert E\left[
\psi_{\tau,t}^{\nu}|\mathcal{G}_{\tau n,t^{\ast}n}\right]  \right\Vert
^{1+\delta/2}\right] \\
&  \leq\left(  \frac{1}{\tau}\right)  ^{1/2+\delta/4}\left(  T-\tau
_{0}\right)  ^{\delta/2}\sum_{t=\tau_{0}+1}^{T}\left(  E\left[  \left\Vert
E\left[  \psi_{\tau,t}^{\nu}|\mathcal{G}_{\tau n,t^{\ast}n}\right]
\right\Vert ^{2}\right]  \right)  ^{1/2+\delta/4},
\end{align*}
where the first inequality uses the triangular inequality and the second and
third inequalities are based on versions of Jensen's inequality. We continue
to use the convention that a term is zero if it is a sum over indices from $a$
to $b$ with $a>b$. Now consider two cases. When $\tau_{0}\rightarrow-\infty$
use Condition \ref{Diag_CLT_Cond}(vii) and the fact that $\left(  T-\tau
_{0}\right)  /\tau\leq1$ as well as $\left(  T-\tau_{0}\right)  ^{-\delta
/2}\leq\left\vert \tau_{0}\right\vert ^{-\delta/2}$
\begin{align*}
&  \left(  \frac{1}{\tau}\right)  ^{1/2+\delta/4}\left(  T-\tau_{0}\right)
^{\delta/2}\sum_{t=\tau_{0}}^{T}\left(  E\left[  \left\Vert E\left[
\psi_{\tau,t}^{\nu}|\mathcal{G}_{\tau n,t^{\ast}n}\right]  \right\Vert
^{2}\right]  \right)  ^{1/2+\delta/4}\\
&  \leq\tau^{-\left(  1/2+\delta/4\right)  }\sum_{t=0}^{T}\left(  E\left[
\left\Vert \psi_{\tau,t}^{\nu}\right\Vert ^{2}\right]  \right)  ^{1/2+\delta
/4}+\tau^{-\left(  1/2+\delta/4\right)  }C\sum_{t=\tau_{0}}^{-1}\left\vert
t\right\vert ^{-\left(  \frac{1}{2}+\frac{3\delta+\delta^{2}}{4}\right)  }\\
&  \leq O\left(  \tau^{-\left(  1/2+\delta/4\right)  }\right)  +\tau^{-\left(
1/2+\delta/4\right)  }\left\vert \tau_{0}\right\vert ^{1/2-\left(
\delta+\delta^{2}\right)  /4}C\sum_{t=\tau_{0}}^{-1}\left\vert t\right\vert
^{-\left(  1+\delta/2\right)  }\\
&  \leq O\left(  \tau^{-\left(  1/2+\delta/4\right)  }\right)  +\tau
^{-\delta/2}\left(  \frac{\left\vert \tau_{0}\right\vert }{\tau}\right)
^{1/2-\delta/4}C\sum_{t=1}^{\infty}t^{-\left(  1+\delta/2\right)  }\\
&  =O\left(  \tau^{-\delta/2}\right)  \rightarrow0
\end{align*}
since $\frac{\left\vert \tau_{0}\right\vert }{\tau}\rightarrow\upsilon$ as
$\tau\rightarrow\infty$ and $\sum_{t=1}^{\infty}t^{-1}(\log\left(  t+1\right)
)^{-\left(  1+\delta\right)  }<\infty$.\textbf{ }The second case arises when
$\tau_{0}$ is fixed. Then,
\begin{align*}
&  \left(  \frac{1}{\tau}\right)  ^{1/2+\delta/4}\left(  T-\tau_{0}\right)
^{\delta/2}\sum_{t=\tau_{0}+1}^{T}\left(  E\left[  \left\Vert E\left[
\psi_{\tau,t}^{\nu}|\mathcal{G}_{\tau n,t^{\ast}n}\right]  \right\Vert
^{2}\right]  \right)  ^{1/2+\delta/4}\\
&  \leq\tau^{-1/2-\delta/4}\sup_{t}\left(  E\left[  \left\Vert \psi_{\tau
,t}^{\nu}\right\Vert ^{2}\right]  \right)  ^{1/2+\delta/4}\left(  T+\left\vert
\tau_{0}\right\vert \right)  ^{1+\delta/2}\rightarrow0
\end{align*}
as $\tau\rightarrow\infty.$ In both cases the Markov inequality then implies
that
\[
\tau^{-1/2}\sum_{t=\tau_{0}+1}^{\tau_{0}+\tau}\sum_{i=1}^{n}\sum_{j=1}%
^{k}\lambda_{j}^{\prime}E\left[  \Delta\tilde{\psi}_{it}\left(  r_{j}\right)
|\mathcal{G}_{\tau n,t^{\ast}n+i-1}\right]  =o_{p}\left(  1\right)  ,
\]
and the conclusion follows consequently.

\subsection{Proof of (\ref{DCLT40})\label{DCLT40-proof}}

Define
\[
\ddot{\psi}_{\tau,t}^{\nu}=\frac{1}{\sqrt{\tau}}\left(  \psi_{\tau,t}^{\nu
}-E\left[  \psi_{\tau,t}^{\nu}|\mathcal{G}_{\tau n,\left(  t-\min\left(
1,\tau_{0}\right)  n+1\right)  }\right]  \right)
\]
and%
\[
\ddot{X}_{n\tau,\nu}\left(  r\right)  =\sum_{t=\tau_{0}+1}^{\tau_{0}+\left[
\tau r\right]  }\ddot{\psi}_{\tau,t}^{\nu}%
\]
such that\textbf{ }%
\[
X_{n\tau,\nu}\left(  r\right)  =\ddot{X}_{n\tau,\nu}\left(  r\right)
+\frac{1}{\sqrt{\tau}}\sum_{t=\tau_{0}+1}^{\tau_{0}+\left[  \tau r\right]
}E\left[  \psi_{\tau,t}^{\nu}|\mathcal{G}_{\tau n,\left(  t-\min\left(
1,\tau_{0}\right)  n+1\right)  }\right]  .
\]
Let $\ddot{S}_{\tau,s}=%
{\textstyle\sum\nolimits_{t=\tau_{0}+1}^{\tau_{0}+s}}
\lambda_{\nu}^{\prime}\ddot{\psi}_{\tau,t}^{\nu}$ and
\[
S_{\tau,s}=%
{\textstyle\sum\nolimits_{t=\tau_{0}+1}^{\tau_{0}+s}}
\lambda_{\nu}^{\prime}\left(  \ddot{\psi}_{\tau,t}^{\nu}+\frac{1}{\sqrt{\tau}%
}E\left[  \left.  \psi_{\tau,t}^{\nu}\right\vert \mathcal{G}_{\tau n,\left(
t-\min\left(  1,\tau_{0}\right)  n+1\right)  }\right]  \right)
\]
as before. Since
\begin{align}
&  P\left(  \max_{s\leq\tau}\left\vert S_{\tau,k+s}-S_{\tau,k}\right\vert
>c\right) \nonumber\\
&  \leq P\left(  \max_{s\leq\tau}\left\vert \ddot{S}_{\tau,k+s}-\ddot{S}%
_{\tau,k}\right\vert +\max_{s\leq\tau}\left\vert \frac{1}{\sqrt{\tau}}%
{\textstyle\sum\nolimits_{t=\tau_{0}+k+1}^{\tau_{0}+s}}
\lambda_{\nu}^{\prime}E\left[  \psi_{\tau,t}^{\nu}|\mathcal{G}_{\tau n,\left(
t-\min\left(  1,\tau_{0}\right)  n+1\right)  }\right]  \right\vert >c\right)
\nonumber\\
&  \leq P\left(  \max_{s\leq\tau}\left\vert \ddot{S}_{\tau,k+s}-\ddot{S}%
_{\tau,k}\right\vert >\frac{c}{2}\right) \label{DCLT_D26}\\
&  +P\left(  \max_{s\leq\tau}\left\vert \frac{1}{\sqrt{\tau}}%
{\textstyle\sum\nolimits_{t=\tau_{0}+k+1}^{\tau_{0}+s}}
\lambda_{\nu}^{\prime}E\left[  \psi_{\tau,t}^{\nu}|\mathcal{G}_{\tau n,\left(
t-\min\left(  1,\tau_{0}\right)  n+1\right)  }\right]  \right\vert >\frac
{c}{2}\right)  \label{DCLT_D27}%
\end{align}
Note that for each $k$ and $\tau$ fixed, $M_{s}=\ddot{S}_{\tau,s+k}-\ddot
{S}_{\tau,k}$ and $\mathcal{F}_{\tau,s}=\sigma\left(  z_{\tau_{0}}%
,...,z_{\tau_{0}+s+k}\right)  ,$ $\left\{  M_{\tau,s},\mathcal{F}_{\tau
,s}\right\}  $ is a martingale. Note that the filtration $\mathcal{F}_{\tau
,s}$ does not depend on $\tau$ when $\tau_{0}$ is held fixed. We prove an
extension\textbf{ }of Hall and Heyde (1980, Theorems 2.1 and 2.2) to
triangular martingale arrays in Lemma \ref{HH_Th2.1} and Corollary
\ref{HH_T2.2} in Appendix \ref{Max_Ineq_Triang}.\footnote{There is only a
limited literature on laws of large numbers for triangular arrays of
martingales. Andrews (1988) or Kanaya (2017) prove weak laws, de Jong (1996)
proves a strong law but without proving a maximal inequality. Atchade (2009)
and Hill (2010) allow for trinagular arrays but only with respect to a fixed
filteration that does not depend on samples size.}

We first consider the term (\ref{DCLT_D27}) and show that%
\begin{equation}
\lim_{c\rightarrow\infty}\limsup_{\tau\rightarrow\infty}c^{2}P\left(
\max_{s\leq\tau}\left\vert \frac{1}{\sqrt{\tau}}%
{\textstyle\sum\nolimits_{t=\tau_{0}+k+1}^{\tau_{0}+s}}
\lambda_{\nu}^{\prime}E\left[  \psi_{\tau,t}^{\nu}|\mathcal{G}_{\tau n,\left(
t-\min\left(  1,\tau_{0}\right)  n+1\right)  }\right]  \right\vert >\frac
{c}{2}\right)  =0. \label{DCLT_D27-negligible}%
\end{equation}
Using the convention that a term is zero if it is a sum over indices from $a$
to $b$ with $a>b$, note that (\ref{DCLT_D27}) is bounded by Markov inequality
by\textbf{ }%
\begin{align*}
&  P\left(  \max_{s\leq\tau}\left\vert
{\textstyle\sum\nolimits_{t=\tau_{0}+k+1}^{\tau_{0}+s}}
\lambda_{\nu}^{\prime}E\left[  \psi_{\tau,t}^{\nu}|\mathcal{G}_{\tau n,\left(
t-\min\left(  1,\tau_{0}\right)  n+1\right)  }\right]  \right\vert >\frac
{c}{2}\sqrt{\tau}\right) \\
&  \leq\frac{4}{c^{2}\tau}E\left[  \max_{s\leq\tau}\left\vert
{\textstyle\sum\nolimits_{t=\tau_{0}+k+1}^{\tau_{0}+s}}
\lambda_{\nu}^{\prime}E\left[  \psi_{\tau,t}^{\nu}|\mathcal{G}_{\tau n,\left(
t-\min\left(  1,\tau_{0}\right)  n+1\right)  }\right]  \right\vert ^{2}\right]
\\
&  \leq\frac{4}{c^{2}\tau}%
{\textstyle\sum\nolimits_{t=\tau_{0}+k+1}^{\tau_{0}+\tau}}
E\left[  \left\Vert E\left[  \psi_{\tau,t}^{\nu}|\mathcal{G}_{\tau n,\left(
t-\min\left(  1,\tau_{0}\right)  n+1\right)  }\right]  \right\Vert
^{2}\right]  ,
\end{align*}
so%
\begin{align}
&  c^{2}P\left(  \max_{s\leq\tau}\left\vert
{\textstyle\sum\nolimits_{t=\tau_{0}+k+1}^{\tau_{0}+s}}
\lambda_{\nu}^{\prime}E\left[  \psi_{\tau,t}^{\nu}|\mathcal{G}_{\tau n,\left(
t-\min\left(  1,\tau_{0}\right)  n+1\right)  }\right]  \right\vert >\frac
{c}{2}\sqrt{\tau}\right) \nonumber\\
&  \leq\frac{4}{\tau}%
{\textstyle\sum\nolimits_{t=T+1}^{\tau_{0}+\tau}}
E\left[  \left\Vert E\left[  \psi_{\tau,t}^{\nu}|\mathcal{G}_{\tau n,\left(
t-\min\left(  1,\tau_{0}\right)  n+1\right)  }\right]  \right\Vert ^{2}\right]
\label{DCLT_D28}\\
&  +\frac{4}{\tau}%
{\textstyle\sum\nolimits_{t=0}^{T}}
E\left[  \left\Vert \psi_{\tau,t}^{\nu}\right\Vert ^{2}\right]  +\frac{4}%
{\tau}%
{\textstyle\sum\nolimits_{t=\tau_{0}+k+1}^{-1}}
\vartheta_{t},\nonumber
\end{align}
For the first term in (\ref{DCLT_D28}), we have
\[
\frac{4}{\tau}%
{\textstyle\sum\nolimits_{t=T+1}^{\tau_{0}+\tau}}
E\left[  \left\Vert E\left[  \psi_{\tau,t}^{\nu}|\mathcal{G}_{\tau n,\left(
t-\min\left(  1,\tau_{0}\right)  n+1\right)  }\right]  \right\Vert
^{2}\right]  =0
\]
because $E\left[  \psi_{\tau,t}^{\nu}|\mathcal{G}_{\tau n,\left(
t-\min\left(  1,\tau_{0}\right)  n+1\right)  }\right]  =0$ for $t>T$. The
second term $\frac{4}{\varepsilon\tau}%
{\textstyle\sum\nolimits_{t=0}^{T}}
E\left[  \left\Vert \psi_{\tau,t}^{\nu}\right\Vert ^{2}\right]  \rightarrow0$
as $\tau\rightarrow\infty$ because $T$ is finite. For the third term in
(\ref{DCLT_D28}) note that
\begin{align*}
\frac{4}{\sqrt{\tau}}%
{\textstyle\sum\nolimits_{t=\tau_{0}+k+1}^{-1}}
\vartheta_{t}  &  \leq\frac{4K}{\sqrt{\tau}}%
{\textstyle\sum\nolimits_{t=\tau_{0}+k+1}^{-1}}
\left(  \left\vert t\right\vert ^{1+\delta}\right)  ^{-1/2}\\
&  \leq\frac{4K\tau^{1/2-\delta/4}}{\sqrt{\tau}}%
{\textstyle\sum\nolimits_{t=1}^{\infty}}
t^{-\left(  1+\delta/4\right)  }=O\left(  \tau^{-\delta/4}\right)
\rightarrow0.
\end{align*}
These considerations imply the desired result (\ref{DCLT_D27-negligible}).

With (\ref{DCLT_D27-negligible}), in order to establish (\ref{DCLT40}), it
suffices to consider (\ref{DCLT_D26}) and show that
\[
\lim_{c\rightarrow\infty}\limsup_{\tau\rightarrow\infty}c^{2}P\left(
\max_{s\leq\tau}\left\vert \ddot{S}_{\tau,k+s}-\ddot{S}_{\tau,k}\right\vert
>\frac{c}{2}\right)  =0
\]
Because%
\begin{align*}
P\left(  \max_{s\leq\tau}\left\vert \ddot{S}_{\tau,k+s}-\ddot{S}_{\tau
,k}\right\vert >\frac{c}{2}\right)   &  \leq\frac{4}{c^{2}\varepsilon}E\left[
\max_{s\leq\tau}\left\vert \ddot{S}_{\tau,k+s}-\ddot{S}_{\tau,k}\right\vert
^{2}\cdot1\left(  \max_{s\leq\tau}\left\vert \ddot{S}_{\tau,k+s}-\ddot
{S}_{\tau,k}\right\vert \geq\frac{c}{2}\right)  \right] \\
&  =\frac{4}{c^{2}\varepsilon}E\left[  \max_{s\leq\tau}\left\vert
{\textstyle\sum\nolimits_{t=k+1}^{k+s}}
\lambda_{\nu}^{\prime}\ddot{\psi}_{\tau,t}^{\nu}\right\vert ^{2}\cdot1\left(
\max_{s\leq\tau}\left\vert
{\textstyle\sum\nolimits_{t=k+1}^{k+s}}
\lambda_{\nu}^{\prime}\ddot{\psi}_{\tau,t}^{\nu}\right\vert \geq\frac{c}%
{2}\right)  \right]
\end{align*}
it suffices to prove that%
\begin{equation}
\lim_{c\rightarrow\infty}\limsup_{\tau\rightarrow\infty}E\left[  \max
_{s\leq\tau}\left\vert
{\textstyle\sum\nolimits_{t=k+1}^{k+s}}
\lambda_{\nu}^{\prime}\ddot{\psi}_{\tau,t}^{\nu}\right\vert ^{2}\cdot1\left(
\max_{s\leq\tau}\left\vert
{\textstyle\sum\nolimits_{t=k+1}^{k+s}}
\lambda_{\nu}^{\prime}\ddot{\psi}_{\tau,t}^{\nu}\right\vert \geq c\right)
\right]  =0. \label{tightness-objective}%
\end{equation}
We show in Appendix \ref{Proof_tight} that (\ref{tightness-objective}) holds
as long as $\sup_{t}E\left[  \left\vert \sqrt{\tau}\lambda_{\nu}^{\prime}%
\ddot{\psi}_{\tau,t}^{\nu}\right\vert ^{2+\delta}\right]  <\infty$.\textbf{
}The latter is satisfied by Condition \ref{Diag_CLT_Cond}(iv).

\subsection{Proof of Corollary \ref{Diag_CLT}}

We note that finite dimensional convergence established in the proof of
Theorem \ref{FCLT} implies that
\[
E\left[  \exp\left(  i\lambda^{\prime}X_{n\tau}\left(  1\right)  \right)
\right]  \rightarrow E\left[  \exp\left(  -\frac{1}{2}\left(  \lambda
_{y}^{\prime}\Omega_{y}\lambda_{y}+\lambda_{\nu}^{\prime}\Omega_{\nu}\left(
1\right)  \lambda_{\nu}\right)  \right)  \right]  .
\]
We also note that because of (\ref{DCLT41c}) it follows that
\[
E\left[  \exp\left(  i\int_{0}^{1}\lambda_{\nu}^{\prime}\left(  \dot{\Omega
}_{\nu}\left(  t\right)  \right)  ^{1/2}dW_{\nu}\left(  t\right)  \right)
\right]  =E\left[  \exp\left(  -\frac{1}{2}\lambda_{\nu}^{\prime}\Omega_{\nu
}\left(  1\right)  \lambda_{\nu}\right)  \right]
\]
which shows that $\int_{0}^{1}\left(  \dot{\Omega}_{\nu}\left(  t\right)
\right)  ^{1/2}dW_{\nu}\left(  t\right)  $ has the same distribution as
$\Omega_{\nu}\left(  1\right)  ^{1/2}W_{\nu}\left(  1\right)  $.

\subsection{Proof of Theorem \ref{CLT_MLE}}

Let $s_{it}^{y}\left(  \theta,\rho\right)  =f_{\theta,it}\left(  \theta
,\rho\right)  $ and $s_{t}^{\nu}\left(  \rho,\beta\right)  =g_{\rho,t}\left(
\rho,\beta\right)  $ in the case of maximum likelihood estimation and
$s_{it}^{y}\left(  \theta,\rho\right)  =f_{it}\left(  \theta,\rho\right)  $
and $s_{t}^{\nu}\left(  \rho,\beta\right)  =g_{t}\left(  \rho,\beta\right)  $
in the case of moment based estimation. Using the notation developed before we
define
\[
\tilde{s}_{it}^{y}\left(  \theta,\rho\right)  =\left\{
\begin{array}
[c]{cc}%
\frac{s_{it}^{y}\left(  \theta,\rho\right)  }{\sqrt{n}} & \text{if }%
t\in\left\{  1,...,T\right\} \\
0 & \text{otherwise}%
\end{array}
\right.
\]
analogously to (\ref{psi_tilde_y}) and
\[
\tilde{s}_{it}^{\nu}\left(  \beta,\rho\right)  =\left\{
\begin{array}
[c]{cc}%
\frac{s_{t}^{\nu}\left(  \beta,\rho\right)  }{\sqrt{\tau}} & \text{if }%
t\in\left\{  \tau_{0}+1,...,\tau_{0}+\tau\right\}  \text{ and }i=1\\
0 & \text{otherwise}%
\end{array}
\right.
\]
analogously to (\ref{psi_tilde_z}). Stack the moment vectors in
\begin{equation}
\tilde{s}_{it}\left(  \phi\right)  \equiv\tilde{s}_{it}\left(  \theta
,\rho\right)  =\left(  \tilde{s}_{it}^{y}\left(  \theta,\rho\right)  ^{\prime
},\tilde{s}_{it}^{\nu}\left(  \beta,\rho\right)  ^{\prime}\right)  ^{\prime}
\label{s_tilde_phi}%
\end{equation}
and define the scaling matrix $D_{n\tau}=\operatorname*{diag}\left(
n^{-1/2}I_{y},\tau^{-1/2}I_{\nu}\right)  $ where $I_{y}$ is an identity matrix
of dimension $k_{\theta}$ and $I_{\nu}$ is an identity matrix of dimension
$k_{\rho}.$ For the maximum likelihood estimator, the moment conditions
(\ref{Moment Cond z}) and (\ref{Moment Cond y}) can be directly written as
\[
\sum_{t=\min(1,\tau_{0}+1)}^{\max(T,\tau_{0}+\tau)}\sum_{i=1}^{n}\tilde
{s}_{it}\left(  \hat{\theta},\hat{\rho}\right)  =0.
\]
For moment based estimators we have by Conditions \ref{Hessian}(i) and (ii)
that
\[
\sup_{\left\Vert \phi-\phi_{0}\right\Vert \leq\varepsilon}\left\Vert \left(
s_{M}^{y}\left(  \theta,\rho\right)  ^{\prime},s_{M}^{\nu}\left(  \beta
,\rho\right)  ^{\prime}\right)  ^{\prime}-\sum_{t=\min(1,\tau_{0}+1)}%
^{\max(T,\tau_{0}+\tau)}\sum_{i=1}^{n}\tilde{s}_{it}\left(  \theta
,\rho\right)  \right\Vert =o_{p}\left(  1\right)  .
\]
It then follows that for the moment based estimators
\[
0=s\left(  \hat{\phi}\right)  =\sum_{t=\min(1,\tau_{0}+1)}^{\max(T,\tau
_{0}+\tau)}\sum_{i=1}^{n}\tilde{s}_{it}\left(  \hat{\theta},\hat{\rho}\right)
+o_{p}\left(  1\right)  .
\]
A first order mean value expansion around $\phi_{0}$ where $\phi=\left(
\theta^{\prime},\rho^{\prime}\right)  ^{\prime}$ and $\hat{\phi}=\left(
\hat{\theta}^{\prime},\hat{\rho}^{\prime}\right)  ^{\prime}$ leads to
\[
o_{p}\left(  1\right)  =\sum_{t=\min(1,\tau_{0}+1)}^{\max(T,\tau_{0}+\tau
)}\sum_{i=1}^{n}\tilde{s}_{it}\left(  \phi_{0}\right)  +\left(  \sum
_{t=\min(1,\tau_{0}+1)}^{\max(T,\tau_{0}+\tau)}\sum_{i=1}^{n}\frac
{\partial\tilde{s}_{it}\left(  \bar{\phi}\right)  }{\partial\phi^{\prime}%
}D_{n\tau}\right)  D_{n\tau}^{-1}\left(  \hat{\phi}-\phi_{0}\right)
\]
or
\[
D_{n\tau}^{-1}\left(  \hat{\phi}-\phi_{0}\right)  =-\left(  \sum
_{t=\min(1,\tau_{0}+1)}^{\max(T,\tau_{0}+\tau)}\sum_{i=1}^{n}\frac
{\partial\tilde{s}_{it}\left(  \bar{\phi}\right)  }{\partial\phi^{\prime}%
}D_{n\tau}\right)  ^{-1}\sum_{t=\min(1,\tau_{0}+1)}^{\max(T,\tau_{0}+\tau
)}\sum_{i=1}^{n}\tilde{s}_{it}\left(  \phi_{0}\right)  +o_{p}\left(  1\right)
\]
where $\bar{\phi}$ satisfies $\left\Vert \bar{\phi}-\phi_{0}\right\Vert
\leq\left\Vert \hat{\phi}-\phi_{0}\right\Vert $ and we note that with some
abuse of notation we implicitly allow for $\bar{\phi}$ to differ across rows
of $\partial\tilde{s}_{it}\left(  \bar{\phi}\right)  /\partial\phi^{\prime}$.
Note that
\[
\frac{\partial\tilde{s}_{it}\left(  \bar{\phi}\right)  }{\partial\phi^{\prime
}}=\left[
\begin{array}
[c]{cc}%
\partial\tilde{s}_{it}^{y}\left(  \theta,\rho\right)  /\partial\theta^{\prime}
& \partial\tilde{s}_{it}^{y}\left(  \theta,\rho\right)  /\partial\rho^{\prime
}\\
\partial\tilde{s}_{it,\rho}^{\nu}\left(  \beta,\rho\right)  /\partial
\theta^{\prime} & \partial\tilde{s}_{it,\rho}^{\nu}\left(  \beta,\rho\right)
/\partial\rho^{\prime}%
\end{array}
\right]
\]
where $\tilde{s}_{it,\rho}^{\nu}$ denotes moment conditions associated with
$\rho.$ From Condition \ref{Hessian}(iii) and Theorem \ref{Diag_CLT} it
follows that (note that we make use of the continuous mapping theorem which is
applicable because Theorem \ref{Diag_CLT} establishes stable and thus joint
convergence)
\[
D_{n\tau}^{-1}\left(  \hat{\phi}-\phi_{0}\right)  =-A\left(  \phi_{0}\right)
^{-1}\sum_{t=\min(1,\tau_{0}+1)}^{\max(T,\tau_{0}+\tau)}\sum_{i=1}^{n}%
\tilde{s}_{it}\left(  \phi_{0}\right)  +o_{p}\left(  1\right)
\]
It now follows from the continuous mapping theorem and joint convergence in
Corollary \ref{Diag_CLT} that
\[
D_{n\tau}^{-1}\left(  \hat{\phi}-\phi_{0}\right)  \overset{d}{\rightarrow
}-A\left(  \phi_{0}\right)  ^{-1}\Omega^{1/2}W\text{ (}\mathcal{C}%
\text{-stably)}%
\]

\subsection{Proof of Corollary \ref{Corollary_CLT_MLE}}

Partition%
\[
A\left(  \phi_{0}\right)  =\left[
\begin{array}
[c]{cc}%
A_{y,\theta} & \sqrt{\kappa}A_{y,\rho}\\
\frac{1}{\sqrt{\kappa}}A_{\nu,\theta} & A_{\nu,\rho}%
\end{array}
\right]
\]
with inverse%
\begin{align*}
A\left(  \phi_{0}\right)  ^{-1}  &  =\left[
\begin{array}
[c]{cc}%
A_{y,\theta}^{-1}+A_{y,\theta}^{-1}A_{y,\rho}\left(  A_{\nu,\rho}%
-A_{\nu,\theta}A_{y,\theta}^{-1}A_{y,\rho}\right)  ^{-1}A_{\nu,\theta
}A_{y,\theta}^{-1} & -\sqrt{\kappa}A_{y,\theta}^{-1}A_{y,\rho}\left(
A_{\nu,\rho}-A_{\nu,\theta}A_{y,\theta}^{-1}A_{y,\rho}\right)  ^{-1}\\
-\frac{1}{\sqrt{\kappa}}\left(  A_{\nu,\rho}-A_{\nu,\theta}A_{y,\theta}%
^{-1}A_{y,\rho}\right)  ^{-1}A_{\nu,\theta}A_{y,\theta}^{-1} & \left(
A_{\nu,\rho}-A_{\nu,\theta}A_{y,\theta}^{-1}A_{y,\rho}\right)  ^{-1}%
\end{array}
\right] \\
&  =\left[
\begin{array}
[c]{cc}%
A^{y,\theta} & \sqrt{\kappa}A^{y,\rho}\\
\frac{1}{\sqrt{\kappa}}A^{\nu,\theta} & A^{\nu,\rho}%
\end{array}
\right]  .
\end{align*}
It now follows from the continuous mapping theorem and joint convergence in
Corollary \ref{Diag_CLT} that
\[
D_{n\tau}^{-1}\left(  \hat{\phi}-\phi_{0}\right)  \overset{d}{\rightarrow
}-A\left(  \phi_{0}\right)  ^{-1}\Omega^{1/2}W\text{ (}\mathcal{C}%
\text{-stably)}%
\]
where the right hand side has a mixed normal distribution,
\[
A\left(  \phi_{0}\right)  ^{-1}\Omega^{1/2}W\sim MN\left(  0,A\left(  \phi
_{0}\right)  ^{-1}\Omega A\left(  \phi_{0}\right)  ^{\prime-1}\right)
\]
and
\[
A\left(  \phi_{0}\right)  ^{-1}\Omega A\left(  \phi_{0}\right)  ^{\prime
-1}=\left[
\begin{array}
[c]{cc}%
A^{y,\theta}\Omega_{y}A^{y,\theta\prime}+\kappa A^{y,\rho}\Omega_{\nu}\left(
1\right)  A^{y,\rho\prime} & \frac{1}{\sqrt{\kappa}}A^{y,\theta}\Omega
_{y}A^{\nu,\theta\prime}+\sqrt{\kappa}A^{y,\rho}\Omega_{\nu}\left(  1\right)
A^{\nu,\rho\prime}\\
\frac{1}{\sqrt{\kappa}}A^{\nu,\theta}\Omega_{y}A^{y,\theta\prime}+\sqrt
{\kappa}A^{\nu,\rho}\Omega_{\nu}\left(  1\right)  A^{y,\rho\prime} & \frac
{1}{\kappa}A^{\nu,\theta}\Omega_{y}A^{\nu,\theta\prime}+A^{\nu,\rho}%
\Omega_{\nu}\left(  1\right)  A^{\nu,\rho\prime}%
\end{array}
\right]
\]
The form of the matrices $\Omega_{y}$ and $\Omega_{\nu}$ follow from Condition
\ref{s-var-ML} in the case of the maximum likelihood estimator. For the moment
based estimator, $\Omega_{y}$ and $\Omega_{\nu}$ follow from Condition
\ref{s-var-GMM}, the definition of $s_{M}^{y}\left(  \theta,\rho\right)  $ and
$s_{M}^{\nu}\left(  \beta,\rho\right)  $ and Conditions \ref{Hessian}(i) and (ii).

\subsection{Proof of Theorem \ref{FCLT_SI}}

We first establish the joint stable convergence of $\left(  V_{\tau n}\left(
r\right)  ,s_{ML}^{y}\right)  .$ Recall that\textbf{ }%
\[
\tau^{-1/2}\nu_{\tau,t}=\exp\left(  \left(  t-\min\left(  1,\tau_{0}\right)
\right)  \gamma/\tau\right)  V\left(  0\right)  +\tau^{-1/2}\sum
_{s=\min\left(  1,\tau_{0}\right)  +1}^{t}\exp\left(  \left(  t-s\right)
\gamma/\tau\right)  \eta_{s}%
\]
and $V_{\tau n}\left(  r\right)  =\tau^{-1/2}\nu_{\tau,\tau_{0}+\left[  \tau
r\right]  }$. Define $\tilde{V}_{\tau n}\left(  r\right)  =\tau^{-1/2}%
\sum_{s=\min\left(  1,\tau_{0}\right)  +1}^{\tau_{0}+\left[  \tau r\right]
}\exp\left(  -s\gamma/\tau\right)  \eta_{s}.$ It follows that
\[
\tau^{-1/2}\nu_{\tau,\tau_{0}+\left[  \tau r\right]  }=\exp\left(  \left(
t-\min\left(  1,\tau_{0}\right)  \right)  \gamma/\tau\right)  V\left(
0\right)  +\exp\left(  \left[  \tau r\right]  \gamma/\tau\right)  \tilde
{V}_{\tau n}\left(  r\right)  .
\]
We establish joint stable convergence of $\left(  \tilde{V}_{\tau n}\left(
r\right)  ,s_{ML}^{y}\right)  $ and use the continuous mapping theorem to deal
with the first term in $\tau^{-1/2}\nu_{\tau\left[  \tau r\right]  }.$ By the
continuous mapping theorem (see Billingsley (1968, p.30)), the
characterization of stable convergence on $D\left[  0,1\right]  $ (as given in
JS, Theorem VIII 5.33(ii)) and an argument used in Kuersteiner and Prucha
(2013, p.119), stable convergence of $\left(  \tilde{V}_{\tau n}\left(
r\right)  ,s_{ML}^{y}\right)  $ implies that%
\[
\left(  \exp\left(  \left[  \tau r\right]  \gamma/\tau\right)  \tilde{V}_{\tau
n}\left(  r\right)  ,s_{ML}^{y}\right)
\]
also converges jointly and $\mathcal{C}$-stably. Subsequently, this argument
will simply be referred to as the `continuous mapping theorem'. In addition
$\exp\left(  \left(  \left[  \tau r\right]  -\min\left(  1,\tau_{0}\right)
\right)  \gamma/\tau\right)  V\left(  0\right)  \rightarrow^{p}\exp\left(
r\gamma\right)  V\left(  0\right)  $ which is measurable with respect to
$\mathcal{C}$. Together these results imply joint stable convergence of
$\left(  V_{\tau n}\left(  r\right)  ,s_{ML}^{y}\right)  $. We thus turn to
$\left(  \tilde{V}_{\tau n}\left(  r\right)  ,s_{ML}^{y}\right)  $. To apply
Theorem \ref{FCLT} we need to show that $\psi_{\tau,s}=\exp\left(
-s\gamma/\tau\right)  \eta_{s}$ satisfies Conditions \ref{Diag_CLT_Cond} iv)
and \ref{Omega_z}. Since
\begin{equation}
\left\vert \exp\left(  -s\gamma/\tau\right)  \eta_{s}\right\vert ^{2+\delta
}=\left\vert \exp\left(  -s/\tau\right)  \right\vert ^{\gamma\left(
2+\delta\right)  }\left\vert \eta_{s}\right\vert ^{2+\delta}=e^{-\gamma
s\left(  2+\delta\right)  /\tau}\left\vert \eta_{s}\right\vert ^{2+\delta}%
\leq\left\vert \eta_{s}\right\vert ^{2+\delta} \label{UI}%
\end{equation}
such that
\[
E\left[  \left\vert \exp\left(  -s\gamma/\tau\right)  \eta_{s}\right\vert
^{2+\delta}\right]  \leq C
\]
and Condition \ref{Diag_CLT_Cond} iv) holds. Note that $E\left[  \left\vert
\eta_{t}\right\vert ^{2+\delta}\right]  \leq C$ holds since we impose
Condition \ref{Diag_Unit Root CLT_Cond}. Next, note that $E\left[  \exp\left(
-2s\gamma/\tau\right)  \eta_{s}^{2}\right]  =\sigma^{2}\exp\left(
-2s\gamma/\tau\right)  $. Then, it follows from the proof of Chan and Wei
(1987, Equation 2.3)\footnote{See Appendix \ref{supplemental-appendix} for
details.} that
\begin{equation}
\tau^{-1}\sum_{t=\tau_{0}+\left[  \tau s\right]  +1}^{\tau_{0}+\left[  \tau
r\right]  }\left(  \psi_{\tau,s}\right)  ^{2}=\tau^{-1}\sum_{t=\tau
_{0}+\left[  \tau s\right]  +1}^{\tau_{0}+\left[  \tau r\right]  }\exp\left(
-2\gamma t/\tau\right)  \eta_{t}^{2}\rightarrow^{p}\sigma^{2}\int_{s}^{r}%
\exp\left(  -2\gamma t\right)  dt. \label{CW2.3}%
\end{equation}
In this case, $\Omega_{\nu}\left(  r\right)  =\sigma^{2}\left(  1-\exp\left(
-2r\gamma\right)  \right)  /2\gamma$ and $\left(  \dot{\Omega}_{\nu}\left(
r\right)  \right)  ^{1/2}=\sigma\exp\left(  -\gamma r\right)  $. By the
relationship in (\ref{DCLT41d}) and Theorem \ref{FCLT} we have that
\[
\left(  \tilde{V}_{\tau n}\left(  r\right)  ,s_{ML}^{y}\right)  \Rightarrow
\left(  \sigma\int_{0}^{r}e^{-s\gamma}dW_{\nu}\left(  s\right)  ,\Omega
_{y}^{1/2}W_{y}\left(  1\right)  \right)  \text{ }\mathcal{C}\text{-stably}%
\]
which implies, by the continuous mapping theorem and $\mathcal{C}$-stable
convergence, that
\begin{equation}
\left(  V_{\tau n}\left(  r\right)  ,s_{ML}^{y}\right)  \Rightarrow\left(
\exp\left(  r\gamma\right)  V\left(  0\right)  +\sigma\int_{0}^{r}e^{\left(
r-s\right)  \gamma}dW_{\nu}\left(  s\right)  ,\Omega_{y}^{1/2}W_{y}\left(
1\right)  \right)  \text{ }\mathcal{C}\text{-stably.} \label{FCLT_SI_5}%
\end{equation}
Note that $\sigma\int_{0}^{r}e^{\left(  r-s\right)  \gamma}dW_{\nu}\left(
s\right)  $ is the same term as in Phillips (1987) while the limit given in
(\ref{FCLT_SI_5}) is the same as in Kurtz and Protter (1991,p.1043).

We now square (\ref{Unit Root Gen Mech}) and sum both sides as in Chan and Wei
(1987, Equation (2.8) or Phillips, (1987) to write
\begin{equation}
\tau^{-1}\sum_{s=\tau_{0}+1}^{\tau+\tau_{0}}\nu_{\tau s-1}\eta_{s}%
=\frac{e^{-\gamma/\tau}}{2}\tau^{-1}\left(  \nu_{\tau,\tau+\tau_{0}}^{2}%
-\nu_{\tau,\tau_{0}}^{2}\right)  +\frac{\tau e^{-\gamma/\tau}}{2}\left(
1-e^{2\gamma/\tau}\right)  \tau^{-2}\sum_{s=\tau_{0}+1}^{\tau+\tau_{0}}%
\nu_{\tau s-1}^{2}-\frac{e^{-\gamma/\tau}}{2}\tau^{-1}\sum_{s=\tau_{0}%
+1}^{\tau+\tau_{0}}\eta_{s}^{2}. \label{FCLT_SI_6}%
\end{equation}
We note that $e^{-\gamma/\tau}\rightarrow1,$ $\tau e^{-\gamma/\tau}\left(
1-e^{2\gamma/\tau}\right)  \rightarrow-2\gamma$. Furthermore, note that for
all $\alpha,\varepsilon>0$ it follows by the Markov and triangular
inequalities and Condition \ref{Diag_Unit Root CLT_Cond}iv) that
\begin{align*}
&  P\left(  \left\vert \tau^{-1}\sum_{t=\tau_{0}+1}^{\tau+\tau_{0}}E\left[
\eta_{s}^{2}1\left\{  \left\vert \eta_{t}\right\vert >\tau^{1/2}%
\alpha\right\}  |\mathcal{G}_{n,t^{\ast}n}\right]  \right\vert >\varepsilon
\right) \\
&  \leq\frac{1}{\tau\varepsilon}\sum_{t=\tau_{0}+1}^{\tau+\tau_{0}}E\left[
\eta_{s}^{2}1\left\{  \left\vert \eta_{t}\right\vert >\tau^{1/2}%
\alpha\right\}  \right]  \leq\frac{\sup_{t}E\left[  \left\vert \eta
_{t}\right\vert ^{2+\delta}\right]  }{\alpha^{\delta}\tau^{\delta/2}%
}\rightarrow0\text{ as }\tau\rightarrow\infty.
\end{align*}
such that Condition 1.3 of Chan and Wei (1987) holds. Let $U_{\tau,k}^{2}%
=\tau^{-1}\sum_{t=\tau_{0}+1}^{k+\tau_{0}}E\left[  \eta_{s}^{2}|\mathcal{G}%
_{n,t^{\ast}n}\right]  .$ Then, by Holder's and Jensen's inequality
\begin{equation}
E\left[  \left\vert U_{\tau,\tau}\right\vert ^{2+\delta}\right]  \leq\tau
^{-1}\sum_{t=\tau_{0}+1}^{\tau+\tau_{0}}E\left[  \left\vert E\left[  \eta
_{s}^{2}|\mathcal{G}_{n,t^{\ast}n}\right]  \right\vert ^{1+\delta/2}\right]
\leq\sup_{t}E\left[  \left\vert \eta_{t}\right\vert ^{2+\delta}\right]
<\infty\label{FCLT_SI_6_1}%
\end{equation}
such that $U_{\tau,\tau}^{2}$ is uniformly integrable. The bound in
(\ref{FCLT_SI_6_1}) also means that by Theorem 2.23 of Hall and Heyde it
follows that $E\left[  \left\vert U_{\tau,\tau}^{2}-\tau^{-1}\sum_{s=\tau
_{0}+1}^{\tau+\tau_{0}}\eta_{t}^{2}\right\vert \right]  \rightarrow0$ and thus
by Condition \ref{Diag_Unit Root CLT_Cond} vii) and by Markov's inequality
\[
\tau^{-1}\sum_{s=\tau_{0}+1}^{\tau+\tau_{0}}\eta_{t}^{2}%
\overset{p}{\rightarrow}\sigma^{2}.
\]
We also have
\begin{equation}
\tau^{-1}\nu_{\tau,\tau+\tau_{0}}^{2}=V_{\tau n}\left(  1\right)  ^{2},
\label{FCLT_SI_7}%
\end{equation}%
\[
\tau^{-1}\nu_{\tau,\tau_{0}}^{2}\overset{p}{\rightarrow}V\left(  0\right)
^{2}%
\]
and%
\[
\tau^{-2}\sum_{s=\tau_{0}+1}^{\tau+\tau_{0}}\nu_{\tau s-1}^{2}=\tau^{-1}%
\sum_{s=1}^{\tau}V_{\tau n}^{2}\left(  \frac{s}{\tau}\right)  =\int_{0}%
^{1}V_{\tau n}^{2}\left(  r\right)  dr
\]
such that by the continuous mapping theorem and (\ref{FCLT_SI_5}) it follows
that
\begin{equation}
\tau^{-1}\sum_{s=\tau_{0}+1}^{\tau+\tau_{0}}\nu_{\tau s-1}\eta_{s}%
\Rightarrow\frac{1}{2}\left(  V_{\gamma,V\left(  0\right)  }\left(  1\right)
^{2}-V\left(  0\right)  ^{2}\right)  -\gamma\int_{0}^{1}V_{\gamma,V\left(
0\right)  }\left(  r\right)  ^{2}dr-\frac{\sigma^{2}}{2}. \label{FCLT_SI_9}%
\end{equation}
An application of Ito's calculus to $V_{\gamma,V\left(  0\right)  }\left(
r\right)  ^{2}/2$ shows that the RHS of (\ref{FCLT_SI_9}) is equal to
$\sigma\int_{0}^{1}V_{\gamma,V\left(  0\right)  }dW_{\nu}$ which also appears
in Kurtz and Protter (1991, Equation 3.10). However, note that the results in
Kurtz and Protter (1991) do not establish stable convergence and thus don't
directly apply here. When $V\left(  0\right)  =0$ these expressions are the
same as in Phillips (1987, Equation 8). It then is a further consequence of
the continuous mapping theorem that
\[
\left(  V_{\tau n}\left(  r\right)  ,s_{ML}^{y},\tau^{-1}\sum_{s=\tau_{0}%
+1}^{\tau+\tau_{0}}\nu_{\tau s-1}\eta_{s}\right)  \Rightarrow\left(
V_{\gamma,V\left(  0\right)  }\left(  r\right)  ,\Omega_{y}^{1/2}W_{y}\left(
1\right)  ,\sigma\int_{0}^{1}V_{\gamma,V\left(  0\right)  }\left(  r\right)
dW_{\nu}\left(  r\right)  \right)  \text{ (}\mathcal{C}\text{-stably).}%
\]

\subsection{Proof of Theorem \ref{CLT_MLE_Unit Root}}

For $\tilde{s}_{it}\left(  \phi\right)  =\left(  \tilde{s}_{it}^{y}\left(
\theta,\rho\right)  ^{\prime},\tilde{s}_{it,\rho}^{\nu}\left(  \rho\right)
\right)  ^{\prime}$ we note that in the case of the unit root model
\[
\frac{\partial\tilde{s}_{it}\left(  \phi\right)  }{\partial\phi^{\prime}%
}=\left[
\begin{array}
[c]{cc}%
\partial\tilde{s}_{it}^{y}\left(  \theta,\rho\right)  /\partial\theta^{\prime}
& \partial\tilde{s}_{it}^{y}\left(  \theta,\rho\right)  /\partial\rho^{\prime
}\\
0 & \partial\tilde{s}_{it,\rho}^{\nu}\left(  \rho\right)  /\partial
\rho^{\prime}%
\end{array}
\right]  .
\]
Defining
\[
A_{\tau n}^{y}\left(  \phi\right)  =\left(  \sum_{t=\min(1,\tau_{0}+1)}%
^{\max(T,\tau_{0}+\tau)}\sum_{i=1}^{n}\frac{\partial\tilde{s}_{it}^{y}\left(
\phi\right)  }{\partial\phi^{\prime}}D_{n\tau}\right)
\]
and partitioning $A_{\tau n}^{y}\left(  \phi\right)  =\left(  A_{\tau
n}^{y,\theta}\left(  \phi\right)  ,A_{\tau n}^{y,\rho}\left(  \phi\right)
\right)  $ where $A_{\tau n}^{y,\theta}\left(  \phi\right)  $ and $A_{\tau
n}^{y,\rho}\left(  \phi\right)  $ contain the partial derivatives with respect
to $\theta$ and $\rho,$ we have as before for some $\left\Vert \tilde{\phi
}-\phi\right\Vert \leq\left\Vert \hat{\phi}-\phi\right\Vert $ that for
\[
A_{\tau n}\left(  \phi\right)  =\left[
\begin{array}
[c]{cc}%
A_{\tau n}^{y,\theta}\left(  \phi\right)  & A_{\tau n}^{y,\rho}\left(
\phi\right) \\
0 & -\tau^{-2}\sum_{t=\tau_{0}}^{\tau_{0}+\tau}\nu_{\tau,t}^{2}%
\end{array}
\right]  ,
\]
we have%
\[
D_{n\tau}^{-1}\left(  \hat{\phi}-\phi_{0}\right)  =-A_{\tau n}\left(
\tilde{\phi}\right)  ^{-1}\sum_{t=\min(1,\tau_{0}+1)}^{\max(T,\tau_{0}+\tau
)}\sum_{i=1}^{n}\tilde{s}_{it}\left(  \phi_{0}\right)
\]
Using the representation
\[
\tau^{-2}\sum_{t=\tau_{0}}^{\tau_{0}+\tau}\nu_{\tau,t}^{2}=\int_{0}^{1}V_{\tau
n}\left(  r\right)  ^{2}dr,
\]
it follows from the continuous mapping theorem and Theorem \ref{FCLT_SI} that
\begin{align}
&  \left(  V_{\tau n}\left(  r\right)  ,s_{ML}^{y},A_{\tau n}^{y}\left(
\phi_{0}\right)  ,\int_{0}^{1}V_{\tau n}\left(  r\right)  ^{2}dr,\tau^{-1}%
\sum_{s=\tau_{0}+1}^{\tau+\tau_{0}}\nu_{\tau s-1}\eta\right)
\label{Joint Convergence}\\
&  \Rightarrow\left(  V\left(  r\right)  ,\Omega_{y}^{1/2}W_{y}\left(
1\right)  ,A^{y}\left(  \phi_{0}\right)  ,\int_{0}^{1}V_{\gamma,V\left(
0\right)  }\left(  r\right)  ^{2}dr,\int_{0}^{s}\sigma V_{\gamma,V\left(
0\right)  }dW_{\nu}\right)  \text{ (}\mathcal{C}\text{-stably).}\nonumber
\end{align}
The partitioned inverse formula implies that
\begin{equation}
A\left(  \phi_{0}\right)  ^{-1}=\left[
\begin{array}
[c]{cc}%
A_{y,\theta}^{-1} & A_{y,\theta}^{-1}A_{y,\rho}\left(  \int_{0}^{1}%
V_{\gamma,V\left(  0\right)  }\left(  r\right)  ^{2}dr\right)  ^{-1}\\
0 & -\left(  \int_{0}^{1}V_{\gamma,V\left(  0\right)  }\left(  r\right)
^{2}dr\right)  ^{-1}%
\end{array}
\right]  \label{Partitioned Inverse}%
\end{equation}
By Condition \ref{Unit Root Uniform Hessian Y}, (\ref{Joint Convergence}) and
the continuous mapping theorem it follows that
\begin{equation}
D_{n\tau}^{-1}\left(  \hat{\phi}-\phi_{0}\right)  \Rightarrow-A\left(
\phi_{0}\right)  ^{-1}\left[
\begin{array}
[c]{c}%
\Omega_{y}^{1/2}W_{y}\left(  1\right) \\
\int_{0}^{s}\sigma V_{\gamma,V\left(  0\right)  }dW_{\nu}%
\end{array}
\right]  . \label{Joint Limit}%
\end{equation}
The result now follows immediately from (\ref{Partitioned Inverse}) and
(\ref{Joint Limit}).

\subsection{Proof of (\ref{tightness-objective})\label{Proof_tight}}

We will follow Billingsley (1968, p.208). Let%
\begin{align*}
\xi_{\tau,t}  &  =\lambda_{\nu}^{\prime}\left(  \psi_{\tau,t}^{\nu}-E\left[
\psi_{\tau,t}^{\nu}|\mathcal{G}_{\tau n,\left(  t-\min\left(  1,\tau
_{0}\right)  n+1\right)  }\right]  \right)  =\sqrt{\tau}\lambda_{\nu}^{\prime
}\ddot{\psi}_{\tau,t}^{\nu},\\
\xi_{\tau,t}^{u}  &  =\xi_{\tau,t}1\left(  \left\vert \xi_{\tau,t}\right\vert
\leq u\right)  ,\\
\eta_{\tau,t}^{u}  &  =\xi_{\tau,t}^{u}-E\left[  \left.  \xi_{\tau,t}%
^{u}\right\vert \mathcal{G}_{\tau n,\left(  t-\min\left(  1,\tau_{0}\right)
n+1\right)  }\right]  ,\\
\delta_{\tau,t}^{u}  &  =\xi_{\tau,t}-\eta_{\tau,t}^{u}=\xi_{\tau,t}-\xi
_{\tau,t}^{u}-E\left[  \left.  \xi_{\tau,t}-\xi_{\tau,t}^{u}\right\vert
\mathcal{G}_{\tau n,\left(  t-\min\left(  1,\tau_{0}\right)  n+1\right)
}\right]
\end{align*}
and note that the expectation in (\ref{tightness-objective}) can be written as%
\begin{align}
&  E\left[  \max_{s\leq\tau}\left\vert
{\textstyle\sum\nolimits_{t=k+1}^{k+s}}
\lambda_{\nu}^{\prime}\ddot{\psi}_{\tau,t}^{\nu}\right\vert ^{2}\cdot1\left(
\max_{s\leq\tau}\left\vert
{\textstyle\sum\nolimits_{t=k+1}^{k+s}}
\lambda_{\nu}^{\prime}\ddot{\psi}_{\tau,t}^{\nu}\right\vert \geq c\right)
\right] \nonumber\\
&  =E\left[  \max_{s\leq\tau}\left\vert \frac{1}{\sqrt{\tau}}%
{\textstyle\sum\nolimits_{t=k+1}^{k+s}}
\xi_{\tau,t}\right\vert ^{2}\cdot1\left(  \max_{s\leq\tau}\left\vert \frac
{1}{\sqrt{\tau}}%
{\textstyle\sum\nolimits_{t=k+1}^{k+s}}
\xi_{\tau,t}\right\vert \geq c\right)  \right]  . \label{tight-ele-0}%
\end{align}
We will use the fact that%
\begin{equation}
\max_{s\leq\tau}\left\vert \frac{1}{\sqrt{\tau}}%
{\textstyle\sum\nolimits_{t=k+1}^{k+s}}
\xi_{\tau,t}\right\vert ^{2}\leq2\max_{s\leq\tau}\left\vert \frac{1}%
{\sqrt{\tau}}%
{\textstyle\sum\nolimits_{t=k+1}^{k+s}}
\eta_{\tau,t}^{u}\right\vert ^{2}+2\max_{s\leq\tau}\left\vert \frac{1}%
{\sqrt{\tau}}%
{\textstyle\sum\nolimits_{t=k+1}^{k+s}}
\delta_{\tau,t}^{u}\right\vert ^{2}, \label{tight-ele-1}%
\end{equation}
which also implies that
\begin{align}
&  1\left(  \max_{s\leq\tau}\left\vert \frac{1}{\sqrt{\tau}}%
{\textstyle\sum\nolimits_{t=k+1}^{k+s}}
\xi_{\tau,t}\right\vert \geq c\right) \nonumber\\
&  =1\left(  \max_{s\leq\tau}\left\vert \frac{1}{\sqrt{\tau}}%
{\textstyle\sum\nolimits_{t=k+1}^{k+s}}
\xi_{\tau,t}\right\vert ^{2}\geq c^{2}\right) \nonumber\\
&  \leq1\left(  2\max_{s\leq\tau}\left\vert \frac{1}{\sqrt{\tau}}%
{\textstyle\sum\nolimits_{t=k+1}^{k+s}}
\eta_{\tau,t}^{u}\right\vert ^{2}\geq\frac{c^{2}}{2}\right)  +1\left(
2\max_{s\leq\tau}\left\vert \frac{1}{\sqrt{\tau}}%
{\textstyle\sum\nolimits_{t=k+1}^{k+s}}
\delta_{\tau,t}^{u}\right\vert ^{2}\geq\frac{c^{2}}{2}\right) \nonumber\\
&  =1\left(  \max_{s\leq\tau}\left\vert \frac{1}{\sqrt{\tau}}%
{\textstyle\sum\nolimits_{t=k+1}^{k+s}}
\eta_{\tau,t}^{u}\right\vert \geq\frac{c}{2}\right)  +1\left(  \max_{s\leq
\tau}\left\vert \frac{1}{\sqrt{\tau}}%
{\textstyle\sum\nolimits_{t=k+1}^{k+s}}
\delta_{\tau,t}^{u}\right\vert \geq\frac{c}{2}\right)  . \label{tight-ele-2}%
\end{align}
Combining (\ref{tight-ele-0}), (\ref{tight-ele-1}) and (\ref{tight-ele-2}), we
obtain%
\begin{align}
&  E\left[  \max_{s\leq\tau}\left\vert
{\textstyle\sum\nolimits_{t=k+1}^{k+s}}
\lambda_{\nu}^{\prime}\ddot{\psi}_{\tau,t}^{\nu}\right\vert ^{2}\cdot1\left(
\max_{s\leq\tau}\left\vert
{\textstyle\sum\nolimits_{t=k+1}^{k+s}}
\lambda_{\nu}^{\prime}\ddot{\psi}_{\tau,t}^{\nu}\right\vert \geq c\right)
\right] \nonumber\\
&  \leq2E\left[  \max_{s\leq\tau}\left\vert \frac{1}{\sqrt{\tau}}%
{\textstyle\sum\nolimits_{t=k+1}^{k+s}}
\eta_{\tau,t}^{u}\right\vert ^{2}\cdot1\left(  \max_{s\leq\tau}\left\vert
\frac{1}{\sqrt{\tau}}%
{\textstyle\sum\nolimits_{t=k+1}^{k+s}}
\eta_{\tau,t}^{u}\right\vert \geq\frac{c}{2}\right)  \right]
\label{tight-term-1}\\
&  +2E\left[  \max_{s\leq\tau}\left\vert \frac{1}{\sqrt{\tau}}%
{\textstyle\sum\nolimits_{t=k+1}^{k+s}}
\delta_{\tau,t}^{u}\right\vert ^{2}\cdot1\left(  \max_{s\leq\tau}\left\vert
\frac{1}{\sqrt{\tau}}%
{\textstyle\sum\nolimits_{t=k+1}^{k+s}}
\delta_{\tau,t}^{u}\right\vert \geq\frac{c}{2}\right)  \right]
\label{tight-term-2}\\
&  +2E\left[  \max_{s\leq\tau}\left\vert \frac{1}{\sqrt{\tau}}%
{\textstyle\sum\nolimits_{t=k+1}^{k+s}}
\eta_{\tau,t}^{u}\right\vert ^{2}\cdot1\left(  \max_{s\leq\tau}\left\vert
\frac{1}{\sqrt{\tau}}%
{\textstyle\sum\nolimits_{t=k+1}^{k+s}}
\delta_{\tau,t}^{u}\right\vert \geq\frac{c}{2}\right)  \right]
\label{tight-term-3}\\
&  +2E\left[  \max_{s\leq\tau}\left\vert \frac{1}{\sqrt{\tau}}%
{\textstyle\sum\nolimits_{t=k+1}^{k+s}}
\delta_{\tau,t}^{u}\right\vert ^{2}\cdot1\left(  \max_{s\leq\tau}\left\vert
\frac{1}{\sqrt{\tau}}%
{\textstyle\sum\nolimits_{t=k+1}^{k+s}}
\eta_{\tau,t}^{u}\right\vert \geq\frac{c}{2}\right)  \right]  .
\label{tight-term-4}%
\end{align}

We will bound each term (\ref{tight-term-1}) - (\ref{tight-term-4}). First, we
have
\begin{align}
\left(  \ref{tight-term-1}\right)   &  =2E\left[  \max_{s\leq\tau}\left\vert
\frac{1}{\sqrt{\tau}}%
{\textstyle\sum\nolimits_{t=k+1}^{k+s}}
\eta_{\tau,t}^{u}\right\vert ^{2}\cdot1\left(  \max_{s\leq\tau}\left\vert
\frac{1}{\sqrt{\tau}}%
{\textstyle\sum\nolimits_{t=k+1}^{k+s}}
\eta_{\tau,t}^{u}\right\vert \geq\frac{c}{2}\right)  \right] \nonumber\\
&  \leq2\cdot\frac{4}{c^{2}}E\left[  \max_{s\leq\tau}\left\vert \frac{1}%
{\sqrt{\tau}}%
{\textstyle\sum\nolimits_{t=k+1}^{k+s}}
\eta_{\tau,t}^{u}\right\vert ^{4}\right]  .\nonumber\\
&  \leq2\cdot\frac{4}{c^{2}}\left(  \frac{4}{3}\right)  ^{4}\frac{1}{\tau^{2}%
}E\left[  \left\vert
{\textstyle\sum\nolimits_{t=k+1}^{k+\tau}}
\eta_{\tau,t}^{u}\right\vert ^{4}\right]  . \label{tight-term-1-1}%
\end{align}
We can note that
\[
E\left[  \max_{s\leq\tau}\left\vert \frac{1}{\sqrt{\tau}}%
{\textstyle\sum\nolimits_{t=k+1}^{k+s}}
\eta_{\tau,t}^{u}\right\vert ^{4}\right]  \leq\left(  \frac{4}{3}\right)
^{4}\frac{1}{\tau^{2}}E\left[  \left\vert
{\textstyle\sum\nolimits_{t=k+1}^{k+\tau}}
\eta_{\tau,t}^{u}\right\vert ^{4}\right]
\]
by using Corollary \ref{HH_Th2.1}. We can then note the fact that by
construction, (i) $\eta_{\tau,t}^{u}$ is a martingale difference; and (ii) it
is bounded by $2u$, which allows us to follow Billingsley's (1968, p.207)
argument, leading to $\frac{1}{\tau^{2}}E\left[  \left\vert
{\textstyle\sum\nolimits_{t=k+1}^{k+s}}
\eta_{\tau,t}^{u}\right\vert ^{4}\right]  \leq6\left(  2u\right)  ^{4}$.
Therefore, we have
\begin{equation}
E\left[  \max_{s\leq\tau}\left\vert \frac{1}{\sqrt{\tau}}%
{\textstyle\sum\nolimits_{t=k+1}^{k+s}}
\eta_{\tau,t}^{u}\right\vert ^{4}\right]  \leq\left(  \frac{4}{3}\right)
^{4}\cdot6\left(  2u\right)  ^{4}. \label{tight-term-1-2}%
\end{equation}
Combining (\ref{tight-term-1-1}) and (\ref{tight-term-1-2}), we obtain
\begin{equation}
\left(  \ref{tight-term-1}\right)  \leq2\frac{4}{c^{2}}\left(  \frac{4}%
{3}\right)  ^{4}6\left(  2u\right)  ^{4}=\frac{Cu^{4}}{c^{2}},
\label{tight-term-1-3}%
\end{equation}
where $C$ denotes a generic finite constant.

Second, we have%
\begin{align}
\left(  \ref{tight-term-2}\right)   &  =2E\left[  \max_{s\leq\tau}\left\vert
\frac{1}{\sqrt{\tau}}%
{\textstyle\sum\nolimits_{t=k+1}^{k+s}}
\delta_{\tau,t}^{u}\right\vert ^{2}\cdot1\left(  \max_{s\leq\tau}\left\vert
\frac{1}{\sqrt{\tau}}%
{\textstyle\sum\nolimits_{t=k+1}^{k+s}}
\delta_{\tau,t}^{u}\right\vert \geq\frac{c}{2}\right)  \right] \nonumber\\
&  \leq2E\left[  \max_{s\leq\tau}\left\vert \frac{1}{\sqrt{\tau}}%
{\textstyle\sum\nolimits_{t=k+1}^{k+s}}
\delta_{\tau,t}^{u}\right\vert ^{2}\right]  . \label{tight-term-2-1}%
\end{align}
Using Lemma \ref{HH_Th2.1} in Appendix \ref{Max_Ineq_Triang}, and the fact
that by construction, $\delta_{\tau,t}^{u}$ is a martingale difference, we can
conclude that%
\begin{align*}
E\left[  \max_{s\leq\tau}\left\vert \frac{1}{\sqrt{\tau}}%
{\textstyle\sum\nolimits_{t=k+1}^{k+s}}
\delta_{\tau,t}^{u}\right\vert ^{2}\right]   &  =\frac{1}{\tau}E\left[
\max_{s\leq\tau}\left\vert
{\textstyle\sum\nolimits_{t=k+1}^{k+s}}
\delta_{\tau,t}^{u}\right\vert ^{2}\right] \\
&  \leq\frac{1}{\tau}\cdot4E\left[  \left(
{\textstyle\sum\nolimits_{t=k+1}^{k+\tau}}
\delta_{\tau,t}^{u}\right)  ^{2}\right] \\
&  =\frac{4}{\tau}%
{\textstyle\sum\nolimits_{t=k+1}^{k+\tau}}
E\left[  \left(  \delta_{\tau,t}^{u}\right)  ^{2}\right]  .
\end{align*}
Now, note that%
\begin{align*}
E\left[  \left(  \delta_{\tau,t}^{u}\right)  ^{2}\right]   &  =E\left[
\left(  \xi_{\tau,t}-\xi_{\tau,t}^{u}-E\left[  \left.  \xi_{\tau,t}-\xi
_{\tau,t}^{u}\right\vert \mathcal{G}_{\tau n,\left(  t-\min\left(  1,\tau
_{0}\right)  n+1\right)  }\right]  \right)  ^{2}\right] \\
&  \leq E\left[  \left(  \xi_{\tau,t}-\xi_{\tau,t}^{u}-E\left[  \xi_{\tau
,t}-\xi_{\tau,t}^{u}\right]  \right)  ^{2}\right] \\
&  \leq E\left[  \left(  \xi_{\tau,t}-\xi_{\tau,t}^{u}\right)  ^{2}\right] \\
&  =E\left[  \left(  \xi_{\tau,t}-\xi_{\tau,t}1\left(  \left\vert \xi_{\tau
,t}\right\vert \leq u\right)  \right)  ^{2}\right] \\
&  =E\left[  \left(  \xi_{\tau,t}1\left(  \left\vert \xi_{\tau,t}\right\vert
>u\right)  \right)  ^{2}\right] \\
&  =E\left[  \xi_{\tau,t}^{2}1\left(  \left\vert \xi_{\tau,t}\right\vert
>u\right)  \right] \\
&  \leq\frac{1}{u^{\delta}}E\left[  \xi_{\tau,t}^{2+\delta}\right]  \leq
\frac{1}{u^{\delta}}\sup_{t}E\left[  \xi_{\tau,t}^{2+\delta}\right]  ,
\end{align*}
where the first inequality is by Billingsley's (1968, p. 184) Lemma 1. It
follows that
\begin{equation}
E\left[  \max_{s\leq\tau}\left\vert \frac{1}{\sqrt{\tau}}%
{\textstyle\sum\nolimits_{t=k+1}^{k+s}}
\delta_{\tau,t}^{u}\right\vert ^{2}\right]  \leq\frac{C}{u^{\delta}}\sup
_{t}E\left[  \xi_{\tau,t}^{2+\delta}\right]  . \label{tight-term-2-2}%
\end{equation}
Combining (\ref{tight-term-2-1}) and (\ref{tight-term-2-2}), we obtain
\begin{equation}
\left(  \ref{tight-term-2}\right)  \leq\frac{C}{u^{\delta}}.\sup_{t}E\left[
\xi_{\tau,t}^{2+\delta}\right]  \label{tight-term-2-3}%
\end{equation}

Third, we have%
\begin{align}
\left(  \ref{tight-term-3}\right)   &  =2E\left[  \max_{s\leq\tau}\left\vert
\frac{1}{\sqrt{\tau}}%
{\textstyle\sum\nolimits_{t=k+1}^{k+s}}
\eta_{\tau,t}^{u}\right\vert ^{2}\cdot1\left(  \max_{s\leq\tau}\left\vert
\frac{1}{\sqrt{\tau}}%
{\textstyle\sum\nolimits_{t=k+1}^{k+s}}
\delta_{\tau,t}^{u}\right\vert \geq\frac{c}{2}\right)  \right] \nonumber\\
&  \leq2\left(  E\left[  \max_{s\leq\tau}\left\vert \frac{1}{\sqrt{\tau}}%
{\textstyle\sum\nolimits_{t=k+1}^{k+s}}
\eta_{\tau,t}^{u}\right\vert ^{4}\right]  \right)  ^{\frac{1}{2}}\left(
E\left[  1\left(  \max_{s\leq\tau}\left\vert \frac{1}{\sqrt{\tau}}%
{\textstyle\sum\nolimits_{t=k+1}^{k+s}}
\delta_{\tau,t}^{u}\right\vert \geq\frac{c}{2}\right)  ^{2}\right]  \right)
^{\frac{1}{2}}\nonumber\\
&  =2\cdot\left(  \left(  \frac{4}{3}\right)  ^{4}\cdot6\left(  2u\right)
^{4}\right)  ^{\frac{1}{2}}\cdot P\left[  \max_{s\leq\tau}\left\vert \frac
{1}{\sqrt{\tau}}%
{\textstyle\sum\nolimits_{t=k+1}^{k+s}}
\delta_{\tau,t}^{u}\right\vert \geq\frac{c}{2}\right]  ,
\label{tight-term-3-1}%
\end{align}
where the first inequality is by Cauchy-Schwarz and the last equality is by
(\ref{tight-term-1-2}). We further have%
\begin{align}
P\left[  \max_{s\leq\tau}\left\vert \frac{1}{\sqrt{\tau}}%
{\textstyle\sum\nolimits_{t=k+1}^{k+s}}
\delta_{\tau,t}^{u}\right\vert \geq\frac{c}{2}\right]   &  \leq\frac{4}{c^{2}%
}E\left[  \left(  \max_{s\leq\tau}\left\vert \frac{1}{\sqrt{\tau}}%
{\textstyle\sum\nolimits_{t=k+1}^{k+s}}
\delta_{\tau,t}^{u}\right\vert \right)  ^{2}\right] \nonumber\\
&  \leq\frac{4}{c^{2}}\frac{C}{u^{\delta}}\sup_{t}E\left[  \xi_{\tau
,t}^{2+\delta}\right] \nonumber\\
&  =\frac{C}{c^{2}u^{\delta}}\sup_{t}E\left[  \xi_{\tau,t}^{2+\delta}\right]
, \label{tight-term-3-2}%
\end{align}
where the first inequality is by Markov, and the second inequality is by
(\ref{tight-term-2-2}). Combining (\ref{tight-term-3-1}) and
(\ref{tight-term-3-2}), we obtain%
\begin{equation}
\left(  \ref{tight-term-3}\right)  \leq\frac{Cu^{2-\delta}}{c^{2}}\sup
_{t}E\left[  \xi_{\tau,t}^{2+\delta}\right]  . \label{tight-term-3-3}%
\end{equation}

Fourth, we have%
\begin{align}
\left(  \ref{tight-term-4}\right)   &  =2E\left[  \max_{s\leq\tau}\left\vert
\frac{1}{\sqrt{\tau}}%
{\textstyle\sum\nolimits_{t=k+1}^{k+s}}
\delta_{\tau,t}^{u}\right\vert ^{2}\cdot1\left(  \max_{s\leq\tau}\left\vert
\frac{1}{\sqrt{\tau}}%
{\textstyle\sum\nolimits_{t=k+1}^{k+s}}
\eta_{\tau,t}^{u}\right\vert \geq\frac{c}{2}\right)  \right] \nonumber\\
&  \leq2E\left[  \max_{s\leq\tau}\left\vert \frac{1}{\sqrt{\tau}}%
{\textstyle\sum\nolimits_{t=k+1}^{k+s}}
\delta_{\tau,t}^{u}\right\vert ^{2}\right] \nonumber\\
&  \leq\frac{C}{u^{\delta}}\sup_{t}E\left[  \xi_{\tau,t}^{2+\delta}\right]  ,
\label{tight-term-4-3}%
\end{align}
where the second inequality is by (\ref{tight-term-2-2}).

Combining (\ref{tight-term-1-3}), (\ref{tight-term-2-3}),
(\ref{tight-term-3-3}), (\ref{tight-term-4}), we obtain that%
\begin{align*}
&  E\left[  \max_{s\leq\tau}\left\vert
{\textstyle\sum\nolimits_{t=k+1}^{k+s}}
\lambda_{\nu}^{\prime}\ddot{\psi}_{\tau,t}^{\nu}\right\vert ^{2}\cdot1\left(
\max_{s\leq\tau}\left\vert
{\textstyle\sum\nolimits_{t=k+1}^{k+s}}
\lambda_{\nu}^{\prime}\ddot{\psi}_{\tau,t}^{\nu}\right\vert \geq c\right)
\right] \\
&  \leq\frac{Cu^{4}}{c^{2}}+\frac{C}{u^{\delta}}.\sup_{t}E\left[  \xi_{\tau
,t}^{2+\delta}\right]  +\frac{Cu^{2-\delta}}{c^{2}}\sup_{t}E\left[  \xi
_{\tau,t}^{2+\delta}\right]  +\frac{C}{u^{\delta}}\sup_{t}E\left[  \xi
_{\tau,t}^{2+\delta}\right]  .
\end{align*}
In view of (\ref{tightness-objective}), it suffices to prove that we can
choose $u\rightarrow\infty$ as a function of $c$ such that the terms above all
converge to zero as $c\rightarrow\infty$. This we can do by choosing
$u=c^{1/3}$, for example.

\section{A Maximal Inequality for Triangular Arrays\label{Max_Ineq_Triang}}

In this section we extend Hall and Heyde (1980) Theorem 2.1 and Corollary 2.1
to the case of triangular arrays of martingales. Let $\mathcal{F}_{\tau,s}$ be
a increasing filtration such that for each $\tau,$ $\mathcal{F}_{\tau
,s}\subset\mathcal{F}_{\tau,s+1}.$ Let $S_{\tau,s}$ be adapted to
$\mathcal{F}_{\tau,s}$ and assume that for all $\tau$ and $k>0,$ $E\left[
S_{\tau,s+k}|\mathcal{F}_{\tau,s}\right]  =S_{\tau,s}.$ Since for $p\geq1,$
$\left\vert .\right\vert ^{p}$ is convex, it follows by Jensen's inequality
for conditional expectations that $E\left[  \left\vert S_{\tau,s+k}\right\vert
^{p}|\mathcal{F}_{\tau,s}\right]  \geq\left\vert E\left[  S_{\tau
,s+k}|\mathcal{F}_{\tau,s}\right]  \right\vert ^{p}=\left\vert S_{\tau
,s}\right\vert ^{p}.$ Thus, for each $\tau,$ $\left\{  \left\vert S_{\tau
,s}\right\vert ^{p},\mathcal{F}_{\tau,s}\right\}  $ is a submartingale. We say
that $\left\{  \left\vert S_{\tau,s}\right\vert ^{p},\mathcal{F}_{\tau
,s}\right\}  $ is a triangular array of submartingales. If $\left\{
S_{\tau,s},\mathcal{F}_{\tau,s}\right\}  $ is a triangular array
submartingales then the same holds for $\left\{  \left\vert S_{\tau
,s}\right\vert ^{p},\mathcal{F}_{\tau,s}\right\}  .$ The following Lemma
extends Theorem 2.1 of Hall and Heyde (1980) to triangular arrays of submartingales.

\begin{lemma}
\label{HH_Th2.1}For each $\tau,$ let $\left\{  S_{\tau,s},\mathcal{F}_{\tau
,s}\right\}  $ be a submartingale $S_{\tau,s}$ with respect to an increasing
filtration $\mathcal{F}_{\tau,s}.$ Then for each real $\lambda,$ and each
$\tau$ it follows that
\[
\lambda P\left(  \max_{s\leq\tau}S_{\tau,s}>\lambda\right)  \leq E\left[
S_{\tau,\tau}1\left\{  \max_{s\leq\tau}S_{\tau,s}>\lambda\right\}  \right]  .
\]

\end{lemma}

\begin{proof}
The proof closely follows Hall and Heyde (1980, p.14), with the necessary
modifications. Define the event
\[
E_{\tau}=\left\{  \max_{s\leq\tau}S_{\tau,s}>\lambda\right\}  =\cup
_{i=1}^{\tau}\left\{  S_{\tau,i}>\lambda;\max_{1\leq j<i}S_{\tau,j}\leq
\lambda\right\}  =\cup_{i=1}^{\tau}E_{\tau,i}.
\]
These events are $\mathcal{F}_{\tau,i}$ measurable and disjoint. Then,
\begin{align*}
\lambda P\left(  E_{\tau}\right)   &  \leq\sum_{i=1}^{\tau}E\left[  S_{\tau
,i}1\left\{  E_{\tau,i}\right\}  \right] \\
&  \leq\sum_{i=1}^{\tau}E\left[  E\left[  S_{\tau,\tau}|\mathcal{F}_{\tau
,i}\right]  1\left\{  E_{\tau,i}\right\}  \right] \\
&  =\sum_{i=1}^{\tau}E\left[  E\left[  S_{\tau,\tau}1\left\{  E_{\tau
,i}\right\}  |\mathcal{F}_{\tau,i}\right]  \right] \\
&  =\sum_{i=1}^{\tau}E\left[  S_{\tau,\tau}1\left\{  E_{\tau,i}\right\}
\right]  =E\left[  S_{\tau,\tau}1\left\{  E_{\tau}\right\}  \right]  .
\end{align*}

\end{proof}

Corollary 2.1 in Hall and Heyde (1980) now follows in the same way. If
$\left\{  S_{\tau,s},\mathcal{F}_{\tau,s}\right\}  $ is a martingale
triangular array then $\left\{  \left\vert S_{\tau,s}\right\vert
^{p},\mathcal{F}_{\tau,s}\right\}  $ is submartingale for $p\geq1.$

\begin{corollary}
\label{HH_C2.1}For each $\tau,$ let $\left\{  S_{\tau,s},\mathcal{F}_{\tau
,s}\right\}  $ be a triangular array of a martingale $S_{\tau,s}$ with respect
to an increasing filtration $\mathcal{F}_{\tau,s}.$ Then for each real
$\lambda$ and for each $p\geq1$ and each $\tau$ it follows that
\[
\lambda^{p}P\left(  \max_{s\leq\tau}\left\vert S_{\tau,s}\right\vert
>\lambda\right)  \leq E\left[  \left\vert S_{\tau,\tau}\right\vert
^{p}\right]  .
\]

\end{corollary}

\begin{proof}
Note that $P\left(  \max_{s\leq\tau}\left\vert S_{\tau,s}\right\vert
>\lambda\right)  =P\left(  \max_{s\leq\tau}\left\vert S_{\tau,s}\right\vert
^{p}>\lambda^{p}\right)  .$ Then, using the fact that $\left\{  \left\vert
S_{\tau,s}\right\vert ^{p},\mathcal{F}_{\tau,s}\right\}  $ is submartingale
for $p\geq1,$ apply Lemma \ref{HH_Th2.1}.
\end{proof}

Below is a triangular array counterpart of Hall and Heyde's (1980) Theorem 2.2.

\begin{corollary}
\label{HH_T2.2}For each $\tau,$ let $\left\{  S_{\tau,s},\mathcal{F}_{\tau
,s}\right\}  $ be a triangular array of a martingale $S_{\tau,s}$ with respect
to an increasing filtration $\mathcal{F}_{\tau,s}.$ Then for each real
$\lambda$ and for each $p>1$ and each $\tau$ it follows that
\[
E\left[  \max_{s\leq\tau}\left\vert S_{\tau,s}\right\vert ^{p}\right]
\leq\left(  \frac{p}{p-1}\right)  ^{p}E\left[  \left\vert S_{\tau,\tau
}\right\vert ^{p}\right]
\]

\end{corollary}

\begin{proof}
This proof is based on Hall and Heyde (1980, proof of Theorem 2.2). Note that
by the layer-cake representation of an integral, we have%
\[
E\left[  \max_{s\leq\tau}\left\vert S_{\tau,s}\right\vert ^{p}\right]
=\int_{0}^{\infty}P\left(  \max_{s\leq\tau}\left\vert S_{\tau,s}\right\vert
^{p}>t\right)  dt=\int_{0}^{\infty}P\left(  \max_{s\leq\tau}\left\vert
S_{\tau,s}\right\vert >t^{\frac{1}{p}}\right)  dt
\]
With the change of variable $x=t^{\frac{1}{p}}$, we get
\[
E\left[  \max_{s\leq\tau}\left\vert S_{\tau,s}\right\vert ^{p}\right]
=p\int_{0}^{\infty}x^{p-1}P\left(  \max_{s\leq\tau}\left\vert S_{\tau
,s}\right\vert >x\right)  dt
\]
Because $\left\vert S_{\tau,s}\right\vert $ is a submartingale, we can apply
Lemma \ref{HH_Th2.1} and obtain%
\begin{align*}
E\left[  \max_{s\leq\tau}\left\vert S_{\tau,s}\right\vert ^{p}\right]   &
\leq p\int_{0}^{\infty}x^{p-2}E\left[  \left\vert S_{\tau,\tau}\right\vert
1\left\{  \max_{s\leq\tau}\left\vert S_{\tau,s}\right\vert >x\right\}  \right]
\\
&  =pE\left[  \left\vert S_{\tau,\tau}\right\vert \int_{0}^{\infty}%
x^{p-2}1\left\{  \max_{s\leq\tau}\left\vert S_{\tau,s}\right\vert >x\right\}
dx\right] \\
&  =pE\left[  \left\vert S_{\tau,\tau}\right\vert \int_{0}^{\max_{s\leq\tau
}\left\vert S_{\tau,s}\right\vert }x^{p-2}dx\right] \\
&  =\frac{p}{p-1}E\left[  \left\vert S_{\tau,\tau}\right\vert \max_{s\leq\tau
}\left\vert S_{\tau,s}\right\vert ^{p-1}\right] \\
&  \leq\frac{p}{p-1}\left(  E\left[  \left\vert S_{\tau,\tau}\right\vert
^{p}\right]  \right)  ^{\frac{1}{p}}\left(  E\left[  \left(  \max_{s\leq\tau
}\left\vert S_{\tau,s}\right\vert ^{p-1}\right)  ^{q}\right]  \right)
^{\frac{1}{q}},
\end{align*}
where the last inequality is an application of H\"{o}lder's inequality for
$q=\frac{1}{1-\frac{1}{p}}=\frac{p}{p-1}$. Dividing both sides by $\left(
E\left[  \left(  \max_{s\leq\tau}\left\vert S_{\tau,s}\right\vert
^{p-1}\right)  ^{q}\right]  \right)  ^{\frac{1}{q}}=\left(  E\left[
\max_{s\leq\tau}\left\vert S_{\tau,s}\right\vert ^{p}\right]  \right)
^{\frac{1}{q}}$, we get%
\[
\left(  E\left[  \max_{s\leq\tau}\left\vert S_{\tau,s}\right\vert ^{p}\right]
\right)  ^{1-\frac{1}{q}}\leq\frac{p}{p-1}\left(  E\left[  \left\vert
S_{\tau,\tau}\right\vert ^{p}\right]  \right)  ^{\frac{1}{p}}%
\]
or%
\[
E\left[  \max_{s\leq\tau}\left\vert S_{\tau,s}\right\vert ^{p}\right]
\leq\left(  \frac{p}{p-1}\right)  ^{p}E\left[  \left\vert S_{\tau,\tau
}\right\vert ^{p}\right]  .
\]

\end{proof}

\section{Proof of (\ref{CW2.3})\label{supplemental-appendix}}

\begin{lemma}
\label{HH_Th2.23}Assume that Conditions \ref{Diag_Unit Root CLT_Cond},
\ref{Unit Root Consistency} and \ref{Unit Root Uniform Hessian Y} hold. For
$r,s\in\left[  0,1\right]  $ fixed and as $\tau\rightarrow\infty$ it follows
that%
\[
\left\vert \tau^{-1}\sum_{t=\tau_{0}+\left[  \tau s\right]  +1}^{\tau
_{0}+\left[  \tau r\right]  }\left(  \left(  \psi_{\tau,s}\right)
^{2}-e^{\left(  -2\gamma t/\tau\right)  }E\left[  \eta_{t}^{2}|\mathcal{G}%
_{\tau n,\left(  t-\min\left(  1,\tau_{0}\right)  -1\right)  n}\right]
\right)  \right\vert \overset{p}{\rightarrow}0
\]

\end{lemma}

\begin{proof}
By Hall and Heyde (1980, Theorem 2.23) we need to show that for all
$\varepsilon>0$
\begin{equation}
\tau^{-1}\sum_{t=\tau_{0}+\left[  \tau s\right]  +1}^{\tau_{0}+\left[  \tau
r\right]  }e^{-2\gamma t/\tau}E\left[  \left.  \eta_{t}^{2}1\left\{
\left\vert \tau^{-1/2}e^{-\gamma t/\tau}\eta_{t}\right\vert >\varepsilon
\right\}  \right\vert \mathcal{G}_{\tau n,\left(  t-\min\left(  1,\tau
_{0}\right)  -1\right)  n}\right]  \overset{p}{\rightarrow}0.
\label{HH_Th2.23_D1}%
\end{equation}
By Condition \ref{Diag_Unit Root CLT_Cond}iv) it follows that for some
$\delta>0$
\begin{align*}
&  E\left[  \tau^{-1}\sum_{t=\tau_{0}+\left[  \tau s\right]  +1}^{\tau
_{0}+\left[  \tau r\right]  }e^{-2\gamma t/\tau}E\left[  \eta_{t}^{2}1\left\{
\left\vert \tau^{-1/2}e^{-\gamma t/\tau}\eta_{t}\right\vert >\varepsilon
\right\}  |\mathcal{G}_{\tau n,\left(  t-\min\left(  1,\tau_{0}\right)
-1\right)  n}\right]  \right] \\
&  \leq\tau^{-\left(  1+\delta/2\right)  }\sum_{t=\tau_{0}+\left[  \tau
s\right]  +1}^{\tau_{0}+\left[  \tau r\right]  }\frac{\left(  e^{-\gamma
t/\tau}\right)  ^{2+\delta}}{\varepsilon^{\delta}}E\left[  \left\vert \eta
_{t}\right\vert ^{2+\delta}\right] \\
&  \leq\sup_{t}E\left[  \left\vert \eta_{t}\right\vert ^{2+\delta}\right]
\frac{\left[  \tau r\right]  -\left[  \tau s\right]  }{\tau^{1+\delta
/2}\varepsilon^{\delta}}e^{\left(  2+\delta\right)  \left\vert \gamma
\right\vert }\rightarrow0.
\end{align*}
This establishes (\ref{HH_Th2.23_D1}) by the Markov inequality. Since
$\tau^{-1}\sum_{t=\tau_{0}+\left[  \tau s\right]  +1}^{\tau_{0}+\left[  \tau
r\right]  }e^{\left(  -2\gamma t/\tau\right)  }E\left[  \eta_{t}%
^{2}|\mathcal{G}_{\tau n,\left(  t-\min\left(  1,\tau_{0}\right)  -1\right)
n}\right]  $ is uniformly integrable by (\ref{UI}) and (\ref{FCLT_SI_6_1}) it
follows from Hall and Heyde (1980, Theorem 2.23, Eq 2.28) that
\[
E\left[  \left\vert \tau^{-1}\sum_{t=\tau_{0}+\left[  \tau s\right]  +1}%
^{\tau_{0}+\left[  \tau r\right]  }\left(  \left(  \psi_{\tau,s}\right)
^{2}-e^{\left(  -2\gamma t/\tau\right)  }E\left[  \eta_{t}^{2}|\mathcal{G}%
_{\tau n,\left(  t-\min\left(  1,\tau_{0}\right)  -1\right)  n}\right]
\right)  \right\vert \right]  \rightarrow0.
\]
The result now follows from the Markov inequality.
\end{proof}

\begin{lemma}
\label{Lemma CW2.5}Assume that Conditions \ref{Diag_Unit Root CLT_Cond},
\ref{Unit Root Consistency} and \ref{Unit Root Uniform Hessian Y} hold. For
$r,s\in\left[  0,1\right]  $ fixed and as $\tau\rightarrow\infty$ it follows
that
\[
\tau^{-1}\sum_{t=\tau_{0}+\left[  \tau s\right]  +1}^{\tau_{0}+\left[  \tau
r\right]  }e^{\left(  -2\gamma t/\tau\right)  }E\left[  \eta_{t}%
^{2}|\mathcal{G}_{\tau n,\left(  t-\min\left(  1,\tau_{0}\right)  -1\right)
n}\right]  \rightarrow^{p}\sigma^{2}\int_{s}^{r}\exp\left(  -2\gamma t\right)
dt.
\]

\end{lemma}

\begin{proof}
The proof closely follows Chan and Wei (1987, p. 1060-1062) with a few
necessary adjustments. Fix $\delta>0$ and choose $s=t_{0}\leq t_{1}\leq...\leq
t_{k}=r$ such that
\[
\max_{i\leq k}\left\vert e^{-2\gamma t_{i}}-e^{-2\gamma t_{i-1}}\right\vert
<\delta.
\]
This implies
\begin{equation}
\left\vert \int_{s}^{r}e^{-2\gamma t}dt-\sum_{i=1}^{k}e^{-2\gamma t_{i}%
}\left(  t_{i}-t_{i-1}\right)  \right\vert \leq\sum_{i=1}^{k}\int_{t_{i-1}%
}^{t_{i}}\left\vert e^{-2\gamma t}-e^{-2\gamma t_{i}}\right\vert dt\leq\delta.
\label{CW2.5_D1}%
\end{equation}
Let $I_{i}=\left\{  l:\left[  \tau t_{i-1}\right]  <l\leq\left[  \tau
t_{i}\right]  \right\}  .$ Then,
\begin{align*}
&  \tau^{-1}\sum_{t=\tau_{0}+\left[  \tau s\right]  +1}^{\tau_{0}+\left[  \tau
r\right]  }e^{-2\gamma t/\tau}E\left[  \eta_{t}^{2}|\mathcal{G}_{\tau
n,\left(  t-\min\left(  1,\tau_{0}\right)  -1\right)  n}\right]  -\sigma
^{2}\int_{s}^{r}e^{-2\gamma t}dt\\
&  =\tau^{-1}\sum_{i=1}^{k}\sum_{l\in I_{i}}e^{-2\gamma l/\tau}E\left[
\eta_{l}^{2}|\mathcal{G}_{\tau n,\left(  l-\min\left(  1,\tau_{0}\right)
-1\right)  n}\right]  -\sigma^{2}\int_{s}^{r}e^{-2\gamma t}dt\\
&  =\tau^{-1}\sum_{i=1}^{k}\sum_{l\in I_{i}}\left(  e^{-2\gamma l/\tau
}-e^{-2\gamma\left[  \tau t_{i-1}\right]  /\tau}\right)  E\left[  \eta_{l}%
^{2}|\mathcal{G}_{\tau n,\left(  l-\min\left(  1,\tau_{0}\right)  -1\right)
n}\right] \\
&  +\sum_{i=1}^{k}e^{-2\gamma\left[  \tau t_{i-1}\right]  /\tau}\left(
\tau^{-1}\sum_{l\in I_{i}}E\left[  \eta_{l}^{2}|\mathcal{G}_{\tau n,\left(
l-\min\left(  1,\tau_{0}\right)  -1\right)  n}\right]  -\sigma^{2}\left(
t_{i}-t_{i-1}\right)  \right) \\
&  +\sum_{i=1}^{k}e^{-2\gamma\left[  \tau t_{i-1}\right]  /\tau}\sigma
^{2}\left(  t_{i}-t_{i-1}\right)  -\sigma^{2}\int_{s}^{r}e^{-2\gamma t}dt\\
&  =I_{n}+II_{n}+III_{n}.
\end{align*}
For $III_{n}$ we have that $e^{-2\gamma\left[  \tau t_{i-1}\right]  /\tau
}\rightarrow e^{-2\gamma t_{i-1}}$ as $\tau\rightarrow\infty.$ In other words,
there exists a $\tau^{\prime}$ such that for all $\tau\geq\tau^{\prime},$
$\left\vert e^{-2\gamma\left[  \tau t_{i-1}\right]  /\tau}-e^{-2\gamma
t_{i-1}}\right\vert \leq\delta$ and by (\ref{CW2.5_D1})%
\[
\left\vert III_{n}\right\vert \leq2\delta.
\]
We also have by Condition \ref{Diag_Unit Root CLT_Cond}vii) that
\[
\tau^{-1}\sum_{l\in I_{i}}E\left[  \eta_{l}^{2}|\mathcal{G}_{\tau n,\left(
l-\min\left(  1,\tau_{0}\right)  -1\right)  n}\right]  \rightarrow\sigma
^{2}\left(  t_{i}-t_{i-1}\right)
\]
as $\tau\rightarrow\infty$ such that by $\max_{i\leq k}\left\vert
e^{2\gamma\left[  \tau t_{i-1}\right]  /\tau}\right\vert \leq e^{2\left\vert
\gamma\right\vert }$%
\[
\left\vert II_{n}\right\vert \leq e^{2\left\vert \gamma\right\vert }\left\vert
\tau^{-1}\sum_{l\in I_{i}}E\left[  \eta_{l}^{2}|\mathcal{G}_{\tau n,\left(
l-\min\left(  1,\tau_{0}\right)  -1\right)  n}\right]  -\sigma^{2}\left(
t_{i}-t_{i-1}\right)  \right\vert =o_{p}\left(  1\right)  .
\]
Finally, there exists a $\tau^{\prime}$ such that for all $\tau\geq
\tau^{\prime}$ it follows that
\begin{align*}
\max_{i\leq k}\max_{l\in I_{i}}\left\vert e^{-2\gamma l/\tau}-e^{-2\gamma
\left[  \tau t_{i-1}\right]  /\tau}\right\vert  &  \leq\max_{i\leq
k}\left\vert e^{-2\gamma\left[  \tau t_{i}\right]  /\tau}-e^{-2\gamma\left[
\tau t_{i-1}\right]  /\tau}\right\vert \\
&  \leq2\max_{i\leq k}\left\vert e^{-2\gamma\left[  \tau t_{i}\right]  /\tau
}-e^{-2\gamma t_{i}}\right\vert \\
&  +\max_{i\leq k}\left\vert e^{-2\gamma t_{i}}-e^{-2\gamma t_{i-1}%
}\right\vert \\
&  \leq2\delta+\delta=3\delta.
\end{align*}
We conclude that
\[
\left\vert I_{n}\right\vert \leq3\delta\left\vert \tau^{-1}\sum_{i=1}^{k}%
\sum_{l\in I_{i}}E\left[  \eta_{l}^{2}|\mathcal{G}_{\tau n,\left(
l-\min\left(  1,\tau_{0}\right)  -1\right)  n}\right]  \right\vert
=3\delta\sigma^{2}\left(  1+o_{p}\left(  1\right)  \right)  .
\]
The remainder of the proof is identical to Chan and Wei (1987, p. 1062).
\end{proof}

\section{Standard Error for Section \ref{OP}\label{WE}}

We state precise sufficient conditions for our example and establish that they
imply the regularity conditions of our general results in Section
\ref{Section_JointCLT}.

\begin{condition}
[EX-1]Assume that \newline i) $f\left(  y_{j,t}|\theta\right)  $ is measurable
with respect to $\mathcal{G}_{\tau n,\left(  t-\min\left(  1,\tau_{0}\right)
\right)  n+i}$ and $E\left[  \left.  f\left(  y_{j,t}|\theta_{0}\right)
\right\vert \mathcal{G}_{\tau n,\left(  t-\min\left(  1,\tau_{0}\right)
\right)  n+j-1}\right]  =0.$\newline ii) $g\left(  \nu_{s}\left(
\beta\right)  |\nu_{s-1}\left(  \beta\right)  ,\beta,\rho\right)  $ is
measurable with respect to $\mathcal{G}_{\tau n,\left(  t-\min\left(
1,\tau_{0}\right)  \right)  n+i}$ for all $i=1,...,n$ and $E\left[  \left.
g\left(  \nu_{s}|\nu_{s-1},\beta_{0},\rho_{0}\right)  \right\vert
\mathcal{G}_{\tau n,\left(  t-\min\left(  1,\tau_{0}\right)  -1\right)
n+i}\right]  =0$ for $t>T$ and all $i=1,...,n.$\newline iii) Let $y_{j,t}^{k}$
be the $k$-th element of $y_{j,t}.$ Then, for some $\delta>0$ and $C<\infty
,$\ $\sup_{it}E\left[  \left\vert y_{j,t}^{k}\right\vert ^{2+\delta}\right]
\leq C$ for all $n\geq1.$\newline iv) For some $\delta>0$ and $C<\infty
,$\ $\sup_{s\leq\tau_{0}+\tau}E\left[  \left\vert Y_{s}^{\ast}\right\vert
^{2+\delta}\right]  \leq C$ and $\sup_{s\leq\tau_{0}+\tau}E\left[  \left\vert
K_{s}^{\ast}\right\vert ^{2+\delta}\right]  \leq C$ and for all $\tau\geq
1.$\newline v) $\left\Vert E\left[  \left.  g\left(  \nu_{t}|\nu_{t-1}%
,\beta_{0},\rho_{0}\right)  \right\vert \mathcal{G}_{\tau n,\left(
t-\min\left(  1,\tau_{0}\right)  -1\right)  n+i}\right]  \right\Vert _{2}%
\leq\vartheta_{t}$\textbf{ }for $t<0$ and all $i=1,...,n$ where
\[
\vartheta_{t}\leq C\left(  \left\vert t\right\vert ^{1+\delta}\right)
^{-1/2}.
\]

\end{condition}

Conditions EX-1(i) and (ii) impose that the estimating functions are
martingale differences relative to the filtrations defined in
(\ref{Information Sets}). These filtrations accumulate information about the
time series and cross-section samples up to a common point in time $t,$ as
well as information about common shocks of the cross-section sample. The
conditions can be interpreted as imposing correct specification of the time
series and cross-section models in terms of the conditional mean. Note that
for the time series moments $E\left[  \left.  g\left(  \nu_{t}|\nu_{t-1}%
,\beta_{0},\rho_{0}\right)  \right\vert \mathcal{G}_{\tau n,\left(
t-\min\left(  1,\tau_{0}\right)  -1\right)  n+i}\right]  \neq0$ for $t\leq T$
because of possible mean dependence of $g\left(  .\right)  $ with the
aggregate shocks generating $\mathcal{C}$. The violation of the moment
conditions for $s\leq T$ leads to possible estimator bias that is being
controlled by imposing Condition (v). It is important to stress that we are
not assuming that the cross-section and time series samples are independent of
each other or the common shocks, or that conditioning on the common shocks
leads to conditionally independent samples. Nor do we assume that the
cross-section is sampled randomly. Such additional assumptions can be invoked
to ensure that laws of large numbers for sample averages hold, but are likely
much stronger than needed. In our theory we impose these laws of large numbers
as high level regularity conditions. Conditions EX-1(iii) and (iv) impose mild
regularity conditions in terms of moments of the marginal distributions of all
variables in the cross-section and time series samples. Finally, Condition (v)
imposes a mixingale condition on the common shock process. We show later that
it holds for a stationary Gaussian AR(1) model for $\nu_{s}$, although the
condition is expected to hold for much more general processes. Condition EX-1
parallels Footnote 32 of HKM20 for a different example where the focus is on
the cross-sectional parameters, while here we use an example that focuses on
the time series parameters as the main object of interest.

Under Condition EX-1 it follows that for $u_{j,t}=e_{j,t+1}^{\left(  C\right)
}+\eta_{j,t+1}$ and $f\left(  y_{j,t}|\theta_{0}\right)  =u_{j,t}z_{j,t}$ the
cross-sectional moment vector satisfies a martingale difference property such
that
\[
E\left[  u_{j,t+1}z_{j,t}|\mathcal{G}_{\tau n,\left(  t-\min\left(  1,\tau
_{0}\right)  \right)  n+j-1}\right]  =0.
\]
In our example, $z_{j,t}=\left(  1,k_{j,t-1},i_{j,t-1}\right)  ^{\prime}$
consists of lagged values that are measurable with respect to $\mathcal{G}%
_{\tau n,\left(  t-\min\left(  1,\tau_{0}\right)  \right)  n+j-1}.$ The
condition then is equivalent to imposing the martingale difference assumption
on the cross-sectional innovation $u_{j,t}.$ The construction of
$\mathcal{G}_{\tau n,\left(  t-\min\left(  1,\tau_{0}\right)  \right)  n+j-1}$
in (\ref{Information Sets}) guarantees that $u_{j,t}$ are uncorrelated both
cross-sectionally and temporally. This implies that when evaluated at the true
parameter $\theta_{0}$,
\[
\operatorname*{Var}\left(  \frac{1}{\sqrt{n}}\sum_{t=1}^{T}\sum_{j=1}%
^{n}f\left(  \left.  y_{j,t}\right\vert \theta_{0}\right)  \right)  =\frac
{1}{n}\sum_{t=1}^{T}\sum_{j=1}^{n}E\left[  u_{j,t}^{2}z_{j,t}z_{j,t}^{\prime
}\right]  .
\]
The next condition postulates that a law of large numbers holds.

\begin{condition}
[EX-2]There exist non-singular constant matrices $\Omega_{f}$ and $\Omega_{g}$
such that
\begin{align}
\Omega_{f}  &  =\operatorname*{plim}_{n\rightarrow\infty}\frac{1}{n}\sum
_{t=1}^{T}\sum_{j=1}^{n}u_{j,t}^{2}z_{j,t}z_{j,t}^{\prime},\label{Def_Omega_f}%
\\
\Omega_{g}  &  =\operatorname*{plim}_{\tau\rightarrow\infty}\frac{1}{\tau}%
\sum_{s=\tau_{0}+1}^{\tau_{0}+\tau}\left(  e_{s}^{\left(  A\right)  }\right)
^{2}\nu_{s-1}^{2}. \label{Def_Omega_g}%
\end{align}

\end{condition}

Also, let $\tilde{\Omega}_{f}=\frac{1}{n}\sum_{t=1}^{T}\sum_{j=1}^{n}\tilde
{u}_{j,t}^{2}z_{j,t}z_{j,t}^{\prime}$ with
\[
\tilde{u}_{j,t}=\mathfrak{y}_{j,t}^{\ast}-\left(  \tilde{\beta}_{0,t}^{\ast
}+\tilde{\beta}_{k}k_{j,t}+\tilde{\alpha}^{\left(  C\right)  }\left(  \phi
_{t}\left(  i_{j,t-1},k_{j,t-1}\right)  -\tilde{\beta}_{k}k_{j,t-1}\right)
\right)  .
\]
and set $W_{n}^{C}=\tilde{\Omega}_{f}^{-1},$ $W^{C}=\Omega_{f}^{-1}$ and
$F_{n}\left(  \beta,\nu\right)  =-h_{n}\left(  \beta,\nu\right)  ^{\prime
}W_{n}^{C}h_{n}\left(  \beta,\nu\right)  .$

At the true parameter values $\beta_{0}$ and $\rho_{0}$ it follows that
$g\left(  \left.  \nu_{s}\left(  \beta\right)  \right\vert \nu_{s-1}\left(
\beta\right)  ,\beta,\rho\right)  =g\left(  \left.  \nu_{s}\right\vert
\nu_{s-1},\beta,\rho\right)  =e_{s}^{\left(  A\right)  }\nu_{s-1}.$ As for the
cross-sectional error, Condition EX-1 implies that $E\left[  e_{s}^{\left(
A\right)  }\nu_{s-1}|\mathcal{G}_{\tau n,\left(  t-\min\left(  1,\tau
_{0}\right)  -1\right)  n+j}\right]  =0$ for all $j=1,...,n.$ By the same
logic as before, $\nu_{s-1}$ is measurable with respect to $\mathcal{G}_{\tau
n,\left(  t-\min\left(  1,\tau_{0}\right)  -1\right)  n+j}$ such that a
martingale difference assumption may be directly imposed on the aggregate time
series shock $e_{s}^{\left(  A\right)  }.$ The martingale difference sequences
(mds) property then implies that
\[
\operatorname*{Var}\left(  \frac{1}{\sqrt{\tau}}\sum_{s=\tau_{0}+1}^{\tau
_{0}+\tau}g\left(  \left.  z_{s}\right\vert \beta_{0},\rho_{0}\right)
\right)  =\frac{1}{\tau}\sum_{s=\tau_{0}+1}^{\tau_{0}+\tau}E\left[  \left(
e_{s}^{\left(  A\right)  }\right)  ^{2}\nu_{s-1}^{2}\right]  .
\]
Let $\tilde{\Omega}_{g}=\frac{1}{\tau}\sum_{s=\tau_{0}+1}^{\tau_{0}+\tau
}\left(  \tilde{e}_{s}^{\left(  A\right)  }\right)  ^{2}\tilde{\nu}_{s-1}^{2}$
where $\tilde{\nu}_{s}=Y_{s}^{\ast}-\tilde{\beta}_{k}K_{s}^{\ast}$ and
$\tilde{e}_{s}^{\left(  A\right)  }=\tilde{\nu}_{s}-\hat{\alpha}^{\left(
A\right)  }\tilde{\nu}_{s-1}.$ Then, set $W_{\tau}^{\tau}=\tilde{\Omega}%
_{g}^{-1},$ $W^{\tau}=\Omega_{g}^{-1}$ and let $G_{\tau}\left(  \beta
,\rho\right)  =-k_{\tau}\left(  \beta,\rho\right)  ^{\prime}W_{\tau}^{\tau
}k_{\tau}\left(  \beta,\rho\right)  $.

We now demonstrate how to obtain the distributional approximations analogous
to (\ref{rho-influence}) and (\ref{expansion-theta-hat-reflecting-tilde-rho}).
Similar arguments for a more complicated and cross-sectionally oriented
example can also be found in Section 6, Eq (35) of HKM20. To obtain explicit
formulas we turn to the derivatives of the criterion functions. We have%
\[
\frac{\partial f\left(  y_{j,t}|\theta\right)  ^{\prime}}{\partial\theta
}=\left[
\begin{array}
[c]{c}%
1\\
-\left(  k_{j,t}-\alpha^{\left(  C\right)  }k_{j,t-1}\right) \\
-\left(  \phi_{t}\left(  i_{j,t-1},k_{j,t-1}\right)  -\beta_{k}k_{j,t-1}%
\right)
\end{array}
\right]  z_{j,t}^{\prime}%
\]
and let $h\left(  \theta\right)  =\operatorname*{plim}_{n\rightarrow\infty
}h_{n}\left(  \beta,\nu\right)  $ and $k\left(  \beta,\rho\right)
=\operatorname*{plim}_{\tau\rightarrow\infty}k_{\tau}\left(  \beta
,\rho\right)  .$ It follows from standard GMM\ large sample theory that
\[
\varphi_{j,t}=\left(  \frac{\partial h\left(  \theta\right)  }{\partial
\theta^{\prime}}W^{C}\frac{\partial h\left(  \theta\right)  }{\partial\theta
}\right)  ^{-1}\frac{\partial h\left(  \theta\right)  }{\partial\theta
^{\prime}}W^{C}f\left(  \left.  y_{j,t}\right\vert \theta_{0}\right)
\]
such that the asymptotic variance covariance matrix of the cross-sectional GMM
estimator is $\left(  \frac{\partial h\left(  \theta_{0}\right)  }%
{\partial\theta^{\prime}}\Omega_{f}^{-1}\frac{\partial h\left(  \theta
_{0}\right)  }{\partial\theta}\right)  ^{-1}$. Note that $\partial h\left(
\theta\right)  /\partial\theta=\operatorname*{plim}_{n\rightarrow\infty}%
n^{-1}\sum_{t=1}^{T}\sum_{i=1}^{n}\partial f\left(  y_{j,t}|\theta\right)
/\partial\theta^{\prime}.$ Similarly, considering the time series estimator
one obtains
\[
\left[
\begin{array}
[c]{c}%
\frac{\partial g\left(  \left.  z_{s}\right\vert \beta,\rho\right)  }%
{\partial\theta}\\
\frac{\partial g\left(  \left.  z_{s}\right\vert \beta,\rho\right)  }%
{\partial\rho}%
\end{array}
\right]  =\left[
\begin{array}
[c]{c}%
0\\
-\left(  K_{s}^{\ast}-\alpha^{\left(  A\right)  }K_{s-1}^{\ast}\right)
\nu_{s-1}\left(  \beta\right)  -\left(  \nu_{s}\left(  \beta\right)
-\alpha^{\left(  A\right)  }\nu_{s-1}\left(  \beta\right)  \right)
K_{s-1}^{\ast}\\
0\\
-\alpha^{\left(  A\right)  }\left(  \nu_{s-1}\left(  \beta\right)  \right)
^{2}%
\end{array}
\right]  .
\]
It then follows again from standard theory that hypothetical estimates for the
parameter $\rho$ obtained from the time series data and using the moment
function $g\left(  .\right)  $ at the true parameter value for $\beta_{k}$
have an asymptotic variance covariance matrix equal to $\left(  \frac{\partial
k\left(  \beta_{0},\rho_{0}\right)  }{\partial\rho^{\prime}}\Omega_{g}%
^{-1}\frac{\partial k\left(  \beta_{0},\rho_{0}\right)  }{\partial\rho
}\right)  ^{-1}$.\textbf{ }Our theory formally allows to handle the case where
$\beta_{k}$ is estimated from the cross-section sample. While the expressions
for the asymptotic distribution of $\rho$ are similar to standard formulas for
two step estimators, a rigorous derivation of the asymptotic approximations is
much more involved because there are two generally dependent samples involved
in the estimation.

With these expressions it is now possible to obtain standard errors, mimicking
the procedure laid out in Section 6 of Hahn, Kuersteiner and Mazzocco (2020).
For this, we sketch how to obtain the joint limiting distribution of the
vector $\phi=\left(  \theta^{\prime},\rho^{\prime}\right)  ^{\prime}$.

The joint (with $\mathcal{C}$ measurable random variables) limiting
distribution of $D_{n\tau}^{-1}J_{n\tau}\left(  \phi_{0}\right)  $ is
established in the following Lemma.

\begin{lemma}
\label{Lemma_EX1}Assume that Conditions EX-1 and EX-2 hold, and that
(\ref{Def_Omega_f}) and (\ref{Def_Omega_g}) are well defined. Then,
\[
D_{n\tau}^{-1}J_{n\tau}\left(  \phi_{0}\right)  \rightarrow_{d}N\left(
0,\Omega\right)  \text{ }\mathcal{C}\text{-stably}%
\]
where $\Omega=\operatorname*{diag}\left(  \Omega_{y},\Omega_{\nu}\right)  $ is
the asymptotic variance covariance matrix of the moment functions defined in
(\ref{Empirical_Joint_GMM}).
\end{lemma}

Lemma \ref{Lemma_EX1} is a direct consequence of Corollary \ref{Diag_CLT} in
Section \ref{Section_JointCLT} and the fact that Condition EX-1 combined with
(\ref{Def_Omega_f}) and (\ref{Def_Omega_g}) imply that Conditions
\ref{Diag_CLT_Cond}, \ref{Omega_z} and \ref{Omega_y} hold for $r=1,$ where $r$
is defined in Condition \ref{Omega_z}.\textbf{ }

\end{document}